\documentclass[sigplan, nonacm]{acmart}

\usepackage{amsmath}
\usepackage{algorithm}
\usepackage{booktabs}
\usepackage{float}
\usepackage{multirow}
\usepackage{xcolor}
\usepackage{algorithmic}
\usepackage{balance}
\usepackage{subcaption}
\usepackage{xspace}
\usepackage{enumitem}
\usepackage{graphicx}
\graphicspath{{figures/}{Figures/}}
\newcommand{\figorbox}[2][width=\linewidth]{%
  \IfFileExists{#2}{\includegraphics[#1]{#2}}{%
    \fbox{\parbox{0.6\linewidth}{\centering\vspace{2em}\textit{[figure not included: \texttt{\detokenize{#2}}]}\vspace{2em}}}}}
\usepackage{tikz}
\usetikzlibrary{arrows.meta,positioning,fit,backgrounds,calc,decorations.pathreplacing,shapes.geometric}

\hypersetup{
  colorlinks  = true,
  linkcolor   = blue!70!black,
  citecolor   = blue!70!black,
  urlcolor    = blue!70!black,
}

\begin{document}
\title{
WISER: Systematic Design-Space Exploration of Trapped Ions with Multiplexed Control
}

\author{Scott Jones}
\affiliation{
  \department{Department of Computer Science and Technology}
  \institution{University of Cambridge}
  \city{Cambridge}
  \country{United Kingdom}
}
\email{sj665@cam.ac.uk}

\author{Song-qing-hao Yang}
\affiliation{
  \department{Department of Physics}
  \institution{University of Cambridge}
  \city{Cambridge}
  \country{United Kingdom}
}
\email{sqhy2@cam.ac.uk}

\author{Prakash Murali}
\affiliation{
  \department{Department of Computer Science and Technology}
  \institution{University of Cambridge}
  \city{Cambridge}
  \country{United Kingdom}
}
\email{pm830@cam.ac.uk}

\begin{abstract}
Trapped-ion Quantum Charge-Coupled Devices (QCCD) are a leading contender for quantum computing, but their scalability is constrained by control wiring and electronics. Wiring using Integrated Switching Electronics (WISE), a recently proposed QCCD architecture, reduces wiring through multiplexing and integrated switching hardware. However, this leaves an execution model with limited parallelism and limited flexibility in ion movement. Further, these devices need to be paired with a quantum error correction (QEC) code to enable fault-tolerant quantum computation (FTQC). Can WISE architectures efficiently support FTQC requirements? Which QEC choices, device and control parameters are practical?

We present WISER, a cross-layer design-space exploration framework for trapped-ion systems with multiplexed control. WISER combines a WISE-specific SAT-based compiler, a physics-informed noise model and logical-memory simulation to estimate logical error rates, logical clock speeds and control power across hardware parameters and QEC families. Its compiler reduces routing time by $2.6$--$18.8\times$ relative to a greedy WISE-compatible baseline.

Our analysis provides concrete design guidance. Two-ion traps with $16$-way multiplexing, $8\times$ lower than the original WISE proposal, give the fastest logical clock that meets our reliability and cold-stage power targets. With current hardware parameters, the distance-7 surface code is the only evaluated code to meet our early-FTQC screen, at $4.58\,$Hz and $2.04\,$W of DAC power per logical qubit. At this operating point, ion transport and recooling take $94$--$96\%$ of the WISE cycle, and even without them sample-and-hold electrode charging leaves millisecond-scale syndrome-extraction rounds.
\end{abstract}
\maketitle

\section{Introduction}

Quantum computing has the potential to transform fields such as cryptography, chemistry, materials science, and optimisation by solving computational problems beyond the practical reach of classical computers~\cite{Shor_1997,grover1996fastquantummechanicalalgorithm,Feynman1982}, e.g.\ RSA factoring~\cite{gidney2025factor2048bitrsa} and quantum chemistry~\cite{Babbush2018,Lee2021}. However, practical quantum advantage is a formidable challenge: useful computations require logical error rates that are many orders of magnitude below those of physical operations~\cite{beverland2022assessingrequirementsscalepractical}. To bridge this gap, we require quantum error correction (QEC). Multiple platforms are racing towards this goal with recent demonstrations of QEC on superconducting~\cite{suppressing_quantum_errors_by_scaling_surface_code}, trapped-ion~\cite{microsoft_quantinuum_paper,small_code_colour_realtime,berthusen2024experiments4dsurfacecode,ft_one_bit_addition_small} and neutral atom qubits~\cite{bluvstein2025fault,reichardt2025faulttolerantquantumcomputationneutral}. QEC comes with high qubit and performance overheads. For example, with surface codes, one logical qubit requires encoding its information across hundreds of physical qubits. For representing full applications, recent estimates indicate that millions of qubits are needed~\cite{O_Gorman_2017,beverland2022assessingrequirementsscalepractical,gidney2025factor2048bitrsa}. Improving these resource requirements through efficient, platform-specific QEC architecture design is important~\cite{GoogleCrypto2026, cain2026shorsalgorithmpossible10000, bhardwaj2026highrateqldpcprocessors, IonQWalkingCat, yoder2025tourgrossmodularquantum}.

Trapped-ion quantum computers store qubits in two long-lived states of individual ions~\cite{CiracZoller1995,Wineland_1998,Wineland_2003}. Ions are confined in a region called a \emph{trap}~\cite{leibfried2003quantum} inside which their shared motion acts as a medium for two-qubit gates~\cite{CiracZoller1995}. Trapped-ion qubits have long coherence times, reaching even multiple hours under carefully controlled conditions~\cite{Wineland_2003,Langer_2005,Wang_2017,Wang2021Coherence}, making them a strong candidate for quantum computation. Using a single trap to hold all qubits~\cite{Lin_2009,Landsman_2019,Leung_2018} becomes difficult beyond $\sim50$ qubits due to qubit addressing, crosstalk, and gate serialisation challenges~\cite{murali2020architectingnoisyintermediatescaletrapped}.

To overcome the limitations of a single trap, the Quantum Charge-Coupled Device (QCCD) architecture was proposed~\cite{Kielpinski_2002, Pino2021}. QCCD breaks the system into a collection of traps. Within each trap, single- and two-qubit gates can be performed. To perform two-qubit operations on qubits in different traps, they need to be co-located using shuttling (physical ion movement).
This paradigm~\cite{Kielpinski_2002, Lekitsch_2017} now underpins the most advanced trapped-ion systems~\cite{ Pino2021, racetrack_Moses2023, helios2025, IonQTempo2025} and has already scaled to 98 physical qubits~\cite{helios2025}. The challenge now is therefore no longer demonstrating quantum computation, but designing systems with several thousand qubits and FTQC capabilities.

Despite their advantages over single-trap systems, QCCD systems also have scaling limits. Ion transport in QCCD requires coordinated manipulation of $\sim 10{-}100$ electrodes per ion. This is implemented using a large number of digital-to-analogue converters (DAC) that individually control each channel~\cite{Malinowski2023, ohira2026cryogenictimedivisionmultiplexedvoltagecontrol}, leading to I/O wiring and power constraints that limit these QCCD to a few hundred qubits~\cite{Malinowski2023}. To overcome this, Malinowski \emph{et al.} proposed the Wiring using Integrated Switching Electronics (WISE) architecture. %
WISE shares waveform generation and distribution electronics across qubits and enables scaling to $10,000$ qubits. It is being actively pursued by IonQ/Oxford Ionics~\cite{ionq, OxfordIonics}. Other DAC-sharing proposals~\cite{ohira2026cryogenictimedivisionmultiplexedvoltagecontrol, Anmasser:2026kpn} also exist; we focus on WISE because it is the most aggressive in terms of wiring reduction and scaling.

Recent work reports that WISE leads to infeasible reductions ($\sim 25\times$) in logical clock speed compared with QCCD systems with no DAC-sharing~\cite{Jones_2026}. However, their work uses a single QEC choice (surface codes), compilation that is not natively optimised for WISE control and only coarse device models.
The choice of compilation, QEC, and device and control parameters can substantially influence architectural comparisons. Our work aims to develop a systematic framework for reasoning about the capabilities and limits of scalable trapped-ion systems and how they compare with existing QCCD architectures.

The first challenge in architecture co-design is the lack of hardware-specific noise models which are essential to understand logical qubit characteristics.
Prior works
ignore detailed trapped-ion system physics~\cite{Malinowski2023, transversality_lattice_Guti_rrez_2019,murali2020architectingnoisyintermediatescaletrapped, Jones_2026, IonQWalkingCat}, including
error channels such as the electric field fluctuations~\cite{brownnutt2015ion, talukdar2016implications}, heating and cooling~\cite{ballance2016high, turchette2000heating}, charge fluctuation from shared control lines~\cite{Malinowski2023, Cheng2024, Anmasser:2026kpn}, and shuttling errors.
Second, existing trapped-ion compilers~\cite{MOVELESS_khan2025movelessminimizingoverheadqccds, muzzletheshuttle,Jones_2026, murali2020architectingnoisyintermediatescaletrapped, Sivarajah_2020} are designed primarily for single-trap or QCCD architectures with no DAC-sharing. They ignore
WISE-specific properties, such as global odd-even control of ion movement and limited gate throughput. Third, there are a variety of QEC schemes that can be used to implement logical qubits, but prior works focus on surface codes~\cite{Jones_2026} or physical implementation of small codes~\cite{racetrack_Moses2023,valentini2024demonstrationtwodimensionalconnectivityscalable, helios2025}, providing little insight into the design space of QEC possibilities and their mappings onto devices.

We ask, \textit{
``What are the right hardware parameters (trap capacity, multiplexing factors, etc.) to target for globally controlled architectures such as WISE? Which code families are feasible under WISE constraints, and what factors ultimately limit logical qubit quality? How well can an architecture-specific compiler support QEC-to-device mappings? Overall, what is its scaling outlook for WISE?''
}

To answer these questions, we present WISER in Figure~\ref{fig:full-pipeline}. WISER sits between two styles of trapped-ion research: theory-driven architecture papers~\cite{cain2026shorsalgorithmpossible10000, bhardwaj2026highrateqldpcprocessors, IonQWalkingCat, yoder2025tourgrossmodularquantum} typically fix a small number of QEC codes and adopt simplified hardware abstractions, while experimental papers characterise individual device components without committing to a code or system-level workload~\cite{Malinowski2023, ohira2026cryogenictimedivisionmultiplexedvoltagecontrol, Anmasser:2026kpn}, leaving a gap between component characterisation and system-level design. WISER connects these two levels of the stack through three tool contributions: (1) physics-informed noise models tailored to trapped-ion QCCD architectures with multiplexed control, (2) a WISE-specific compiler, with SAT-based qubit mapping and routing whose schedules are provably optimal within its encoded routing model when the search converges and carry a proven lower bound otherwise, and (3) a configurable framework that enables sweeps over detailed hardware noise parameters, several QEC families and architectural design points. Together, this enabled our extensive design space exploration, leading to the following contributions:

\begin{enumerate}

\item While the original WISE proposal assumed $M=128$-way multiplexing~\cite{Malinowski2023} of control wires, we show that a much lower order is preferable. Across the evaluated orders, $M=16$ is the only one at which the fastest code meets every early-FTQC requirement: it keeps $98.5\%$ of the unmultiplexed logical clock speed at $6.9\%$ of the DAC power (\S\ref{sec:result2}). The same choice suits other QCCD systems with sample-and-hold multiplexed control: with transport and cooling removed, $M=16$ keeps a syndrome-extraction round within $1.5$--$2.0\times$ of the unmultiplexed round, against $5$--$10\times$ at $M=128$ (\S\ref{sec:result5_multiplexing_limits}).

\item Among the evaluated trap capacities $k\in\{2,4,6,8,16\}$, two-ion traps give the best logical clock speed and power for early FTQC (\S\ref{sec:result2}). This matches recent QCCD trap-capacity results~\cite{Jones_2026} even though the underlying control system is drastically different. Further, laser cooling before every gate block is the preferred cooling policy in $18$ of the $20$ feasible combinations of technology tier, error target and objective, consistent with current experimental practice~\cite{racetrack_Moses2023,helios2025}.

\item With models of current hardware, surface codes are the only family to meet the requirements for early FTQC: the distance-7 surface code passes at $4.58\,$Hz, using $2.04\,\mathrm{W}$ per logical qubit at $M=16$ (\S\ref{sec:result2}). Surface codes offer the fastest executions, while the high encoding rate of BB codes gives the lowest power per logical qubit (\S\ref{sec:result2}).

\item Prior works provide only architectural parameter guidance~\cite{murali2020architectingnoisyintermediatescaletrapped}. Towards FTQC, we provide design targets across several low-level physical parameters (\S\ref{sec:result3_dominant}): improving state preparation and measurement from $10^{-4}$ to $10^{-5}$ alone removes $70\%$ of the expected physical faults of a syndrome round, and ion transport and recooling, which take $94$--$96\%$ of the cycle, must be shortened together because neither alone gives more than a $2.8\times$ speed-up (\S\ref{sec:result_isolating_wise_slowdown}).

\item Beyond WISE, even with all transport and cooling removed, multiplexed control with non-overlapped charging lengthens a surface code syndrome-extraction round from $1.3\,$ms to $6.7\,$ms when going from $M=1$ to $M=128$ on current hardware. At these round times, recovering a P-256 elliptic-curve private key within 48~hours with the same number of logical rounds as a recent neutral-atom estimate~\cite{cain2026shorsalgorithmpossible10000} requires a $260\times$ reduction in logical circuit depth, rising to $4{,}100\times$ once WISE's transport and recooling are included (\S\ref{sec:result5_multiplexing_limits}).

\end{enumerate}
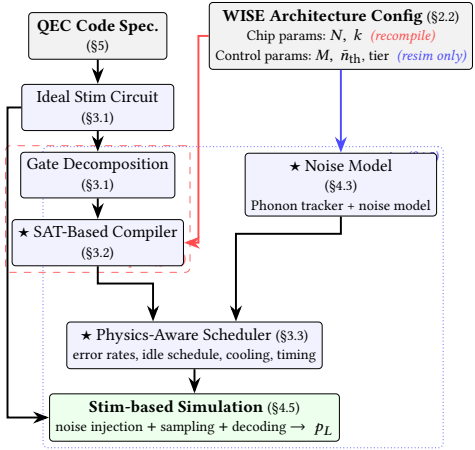
\begin{figure}[t]
\centering
\resizebox{0.75\columnwidth}{!}{
\begin{tikzpicture}[
  >=Stealth,
  every node/.style={font=\scriptsize},
  box/.style={draw, rounded corners=1.5pt, inner sep=2.5pt,
              align=center, fill=white, minimum width=22mm},
  inputbox/.style={box, fill=gray!10, minimum width=20mm},
  stagebox/.style={box, fill=blue!5},
  outbox/.style={box, fill=green!7},
  arr/.style={->, thick},
  node distance=3.5mm and 5mm,
]
\node[inputbox] (qec)
  {\textbf{QEC Code Spec.}\\ {\tiny(\S\ref{sec:evaluation})}};
\node[inputbox, right=6mm of qec, minimum width=34mm] (config)
  {\textbf{WISE Architecture Config} {\tiny(\S\ref{sec:arch-params})}\\
   {\tiny Chip params: $N,\;k$
         \;\textcolor{red!70}{\emph{(recompile)}}}\\
   {\tiny Control params: $M,\;\bar{n}_{\text{th}}$, tier
         \;\textcolor{blue!70}{\emph{(resim only)}}}};

\node[stagebox, below=3.5mm of qec] (stim)
  {Ideal Stim Circuit\\ {\tiny(\S\ref{sec:gate-decomposition})}};
\draw[arr] (qec) -- (stim);

\node[stagebox, below=3mm of stim] (decomp)
  {Gate Decomposition\\ {\tiny(\S\ref{sec:gate-decomposition})}};
\draw[arr] (stim) -- (decomp);

\node[stagebox, below=3mm of decomp] (qmr)
  {$\star$ SAT-Based Compiler\\ {\tiny(\S\ref{sec:sat-routing})}};
\draw[arr] (decomp) -- (qmr);

\node[stagebox, below=12mm of config] (noise)
  {$\star$ Noise Model\\ {\tiny(\S\ref{sec:noise-model})}\\
   {\tiny Phonon tracker + noise model}};

\draw[arr, red!70] (config.west) -- ++(-1.5mm,0) |- (qmr.east);
\draw[arr, blue!70] (config) -- (noise);

\node[stagebox, below=8mm of qmr, xshift=14mm] (sched)
  {$\star$ Physics-Aware Scheduler {\tiny(\S\ref{sec:scheduling})}\\
   {\tiny error rates, idle schedule, cooling, timing}};
\draw[arr] (qmr.south) -- ++(0,-2.5mm) -| ([xshift=-6mm]sched.north);
\draw[arr] (noise.south) -- ++(0,-2.5mm) -| ([xshift=6mm]sched.north);

\node[outbox, below=3.5mm of sched] (sim)
  {\textbf{Stim-based Simulation}
   {\tiny(\S\ref{sec:noise-injection-detail})}\\
   {\tiny noise injection $+$ sampling $+$ decoding $\to\;p_L$}};
\draw[arr] (sched) -- (sim);
\draw[arr] (stim.west) -- ++(-2mm,0) |- (sim.west);

\begin{scope}[on background layer]
  \node[draw=red!60, dashed, rounded corners=2pt, inner sep=2.5pt,
        fit=(decomp)(qmr)] (recompbox) {};
  \node[font=\tiny\itshape, text=red!70, anchor=south east,
        fill=white, inner sep=0.5pt]
    at ([xshift=-0.5mm, yshift=0.5mm]recompbox.south east)
    {recompile (\S\ref{sec:sat-routing})};
\end{scope}

\begin{scope}[on background layer]
  \node[draw=blue!60, densely dotted, rounded corners=2pt, inner sep=2.5pt,
        fit=(noise)(sched)(sim)] (resimbox) {};
  \node[font=\tiny\itshape, text=blue!70, anchor=north east,
        fill=white, inner sep=0.5pt]
    at ([xshift=-0.5mm, yshift=-0.5mm]resimbox.north east)
    {resim (\S\ref{sec:noise-injection-detail})};
\end{scope}
\end{tikzpicture}
}
\caption{WISER takes a QEC code and an architectural configuration as inputs. As output it produces estimates of the logical error rate and logical cycle time of the system. The QEC specification is used to create a noiseless Stim circuit~\cite{stim_Gidney2021}, then lowered to WISE's native gate set and then mapped and scheduled onto the device. Our noise model supplies the error channels for a noise simulation with Stim, producing the logical error rate $p_L$; the schedule supplies the logical cycle time $T_\mathrm{logical}$. Chip-parameter changes require recompilation through the red dashed box; control-parameter changes need only re-scheduling and simulation through the dotted box. ($\star$ = our contributions).
}
\Description{Pipeline from a QEC specification and a WISE configuration through gate decomposition, SAT-based routing, physics-aware scheduling and noise injection to Stim simulation of the logical error rate.}
\label{fig:full-pipeline}
\end{figure}

\begin{figure}[t]
    \centering
    \resizebox{\columnwidth}{!}{
    \begin{tikzpicture}[
        x=0.42cm, y=0.42cm,
        zodd/.style={draw, thin, minimum width=0.40cm, minimum height=0.40cm,
                     inner sep=0pt, fill=green!18},
        zeven/.style={draw, thin, minimum width=0.40cm, minimum height=0.40cm,
                      inner sep=0pt, fill=orange!18},
        jzodd/.style={draw, thin, minimum width=0.48cm, minimum height=0.48cm,
                      inner sep=0pt, fill=green!18},
        jzeven/.style={draw, thin, minimum width=0.48cm, minimum height=0.48cm,
                       inner sep=0pt, fill=orange!18},
        iondot/.style={circle, fill=gray!55, inner sep=0pt, minimum size=5pt},
        oddH/.style={<->, >={Stealth[length=2pt]}, semithick, red!70!black},
        evenH/.style={<->, >={Stealth[length=2pt]}, semithick, blue!60!black},
        oddV/.style={<->, >={Stealth[length=2pt]}, semithick, cyan!60!black},
        evenV/.style={<->, >={Stealth[length=2pt]}, semithick, magenta!60!black},
        lbl/.style={font=\sffamily\tiny},
        grp/.style={draw, densely dotted, rounded corners=1.5pt,
                    inner sep=1.5pt, gray!70},
        trapbox/.style={draw, thick, dashed, rounded corners=2pt,
                        inner sep=2pt, black!80},
    ]
    \def\rowpitch{2.0}

    \newcommand{\ionzone}[4]{
      \pgfmathtruncatemacro{\zpar}{mod(#2-1+#3,2)}
      \ifnum\zpar=0\node[zodd](#1)at(#3,#4){};\else\node[zeven](#1)at(#3,#4){};\fi
      \node[iondot](id-#1)at(#3,#4){};
    }
    \newcommand{\segzone}[4]{
      \pgfmathtruncatemacro{\zpar}{mod(#2-1+#3,2)}
      \ifnum\zpar=0\node[zodd](#1)at(#3,#4){};\else\node[zeven](#1)at(#3,#4){};\fi
    }
    \newcommand{\junczone}[4]{
      \pgfmathtruncatemacro{\zpar}{mod(#2-1+#3,2)}
      \ifnum\zpar=0\node[jzodd](#1)at(#3,#4){};\else\node[jzeven](#1)at(#3,#4){};\fi
      \draw[gray!40,very thin](#1.north west)--(#1.south east)
                              (#1.north east)--(#1.south west);
      \node[iondot](id-#1)at(#3,#4){};
    }

    \foreach \rr in {1,...,4}{
      \pgfmathsetmacro{\ry}{-(\rr-1)*\rowpitch}
      \ionzone{z-\rr-0}{\rr}{0}{\ry}
      \ionzone{z-\rr-1}{\rr}{1}{\ry}
      \ionzone{z-\rr-2}{\rr}{2}{\ry}
      \segzone{z-\rr-3}{\rr}{3}{\ry}
      \junczone{z-\rr-4}{\rr}{4}{\ry}
      \segzone{z-\rr-5}{\rr}{5}{\ry}
      \ionzone{z-\rr-6}{\rr}{6}{\ry}
      \ionzone{z-\rr-7}{\rr}{7}{\ry}
      \ionzone{z-\rr-8}{\rr}{8}{\ry}
      \segzone{z-\rr-9}{\rr}{9}{\ry}
      \junczone{z-\rr-10}{\rr}{10}{\ry}
      \segzone{z-\rr-11}{\rr}{11}{\ry}
      \ionzone{z-\rr-12}{\rr}{12}{\ry}
      \ionzone{z-\rr-13}{\rr}{13}{\ry}
      \ionzone{z-\rr-14}{\rr}{14}{\ry}
      \foreach \ii [evaluate=\ii as \jj using int(\ii+1)] in
        {0,1,2,3,4,5,6,7,8,9,10,11,12,13}{
        \draw[gray!35,very thin](z-\rr-\ii.east)--(z-\rr-\jj.west);
      }
    }

    \foreach \jxv in {4,10}{
      \foreach \rr [evaluate=\rr as \rnxt using int(\rr+1)] in {1,2,3}{
        \pgfmathsetmacro{\my}{-(\rr-1)*\rowpitch - \rowpitch/2}
        \coordinate(vs-\rr-\jxv)at(\jxv,\my);
        \pgfmathtruncatemacro{\toppar}{mod(\rr-1+\jxv,2)}
        \pgfmathtruncatemacro{\botpar}{mod(\rr+\jxv,2)}
        \ifnum\toppar=0
          \fill[green!18]([xshift=-0.20cm]vs-\rr-\jxv)
            rectangle([xshift=0.20cm,yshift=0.20cm]vs-\rr-\jxv);
        \else
          \fill[orange!18]([xshift=-0.20cm]vs-\rr-\jxv)
            rectangle([xshift=0.20cm,yshift=0.20cm]vs-\rr-\jxv);
        \fi
        \ifnum\botpar=0
          \fill[green!18]([xshift=-0.20cm,yshift=-0.20cm]vs-\rr-\jxv)
            rectangle([xshift=0.20cm]vs-\rr-\jxv);
        \else
          \fill[orange!18]([xshift=-0.20cm,yshift=-0.20cm]vs-\rr-\jxv)
            rectangle([xshift=0.20cm]vs-\rr-\jxv);
        \fi
        \draw[thin]([xshift=-0.20cm,yshift=-0.20cm]vs-\rr-\jxv)
          rectangle([xshift=0.20cm,yshift=0.20cm]vs-\rr-\jxv);
        \draw[gray!40,very thin]([xshift=-0.20cm]vs-\rr-\jxv)
          --([xshift=0.20cm]vs-\rr-\jxv);
      }
    }

    \node[trapbox, fit=(z-1-0)(z-1-2),
          label={[lbl,yshift=-1pt]above:\textbf{\textsf{trap}}\,(k{=}3)}] {};
    \node[grp, fit=(z-4-5),
          label={[lbl,gray!55!black]below:\textsf{seg}}] {};
    \node[grp, fit=(z-4-4),
          label={[lbl]below:\textsf{junc}}] {};

    \node[lbl,green!50!black,anchor=north]
      at([yshift=-2pt]z-1-12.south){\textsf{odd}};
    \node[lbl,orange!60!black,anchor=north]
      at([yshift=-2pt]z-1-13.south){\textsf{even}};

    \foreach \rr in {1,...,4}{
      \foreach \base in {0,6,12}{
        \pgfmathtruncatemacro{\bplusone}{\base+1}
        \draw[oddH]([yshift=4pt]id-z-\rr-\base.north)
          to[out=50,in=130]([yshift=4pt]id-z-\rr-\bplusone.north);
      }
    }
    \node[lbl,red!70!black,anchor=south]
      at([yshift=7pt]$(id-z-2-0)!0.5!(id-z-2-1)$){\textsf{odd-H}};

    \foreach \rr in {1,...,4}{
      \foreach \base in {1,7,13}{
        \pgfmathtruncatemacro{\bplusone}{\base+1}
        \draw[evenH]([yshift=4pt]id-z-\rr-\base.north)
          to[out=50,in=130]([yshift=4pt]id-z-\rr-\bplusone.north);
      }
    }
    \node[lbl,blue!60!black,anchor=south]
      at([yshift=7pt]$(id-z-2-7)!0.5!(id-z-2-8)$){\textsf{even-H}};

    \draw[oddV]([xshift=4pt]id-z-1-4.east)
      to[out=-15,in=15]([xshift=4pt]id-z-2-4.east);
    \draw[oddV]([xshift=4pt]id-z-1-10.east)
      to[out=-15,in=15]([xshift=4pt]id-z-2-10.east);
    \draw[oddV]([xshift=4pt]id-z-3-4.east)
      to[out=-15,in=15]([xshift=4pt]id-z-4-4.east);
    \node[lbl,cyan!60!black,fill=white,fill opacity=0.75,text opacity=1,
          inner sep=0.5pt] at (vs-1-10) {\textsf{odd-V}};

    \draw[evenV]([xshift=8pt]id-z-2-4.east)
      to[out=-15,in=15]([xshift=8pt]id-z-3-4.east);
    \draw[evenV]([xshift=8pt]id-z-2-10.east)
      to[out=-15,in=15]([xshift=8pt]id-z-3-10.east);
    \node[lbl,magenta!60!black,fill=white,fill opacity=0.75,text opacity=1,
          inner sep=0.5pt] at (vs-2-10) {\textsf{even-V}};

    \end{tikzpicture}
    }
    \Description{Grid of WISE zones in an odd/even checkerboard, with traps of three zones, transport segments, junctions, and the four broadcast swap primitives.}
    \caption{WISE chip layout~\cite{Malinowski2023}. Each small square is one \emph{zone}. Zones alternate odd
             (\protect\tikz[baseline=-0.6ex]\protect\node[draw,
             fill=green!18,inner sep=0pt,minimum size=5pt]{};)
             / even
             (\protect\tikz[baseline=-0.6ex]\protect\node[draw,
             fill=orange!18,inner sep=0pt,minimum size=5pt]{};)
             checkerboard with odd zones having even zones as neighbours. Four global reconfiguration primitives (\S\ref{sec:wise-arch}):
             \textcolor{red!70!black}{\textsf{parallel-odd-H}},
             \textcolor{blue!60!black}{\textsf{parallel-even-H}},
             \textcolor{cyan!60!black}{\textsf{parallel-odd-V}} and
             \textcolor{magenta!60!black}{\textsf{parallel-even-V}}. The dashed outline groups $k=3$ zones into one trap; grey dots are ions; transport segments and junctions connect traps. \textbf{WISE replaces per-electrode control by four globally broadcast local swap primitives.}}
    \label{fig:wise-chip}
\end{figure}
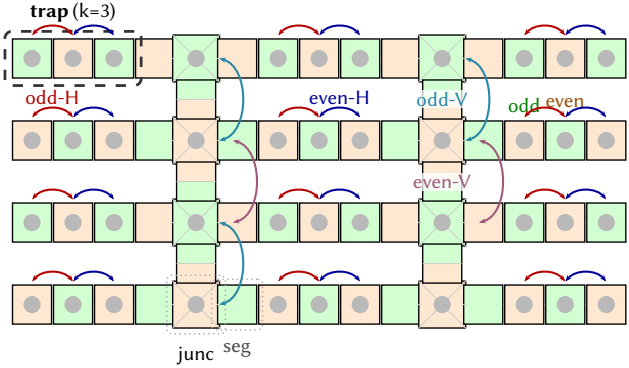

\section{Design Choices and Tradeoffs}
\label{sec:wise-arch}
\subsection{Background on WISE Architecture}
Figure~\ref{fig:wise-chip} shows an example WISE qubit chip. In WISE, an array of $m {\times} n$ ions is stored in a 2D array of traps. Each trap contains $k$ ions. The fundamental layout unit in WISE is a \emph{zone}: a region of control lines (that generate electric fields to confine, position, and move ions) holding at most one ion.

\emph{Traps} are a group of $k$ consecutive zones. Traps enable the execution of gates. Entanglement between two trapped ions is generated through their shared motional coupling, known as a M{\"o}lmer--S{\"o}rensen (MS) interaction~\cite{Sorensen1999,S_rensen_2000}. Between traps sit single-zone \emph{transport segments} and \emph{junctions} enabling two-axis routing~\cite{Malinowski2023}. Routing is performed through ion swapping. WISE uses \emph{switchable swaps}: two ions in neighbouring zones are physically exchanged when both on-chip switches are set to~$1$; otherwise ions remain stationary. To orchestrate a parallel swap step, we set an $N$-bit word appropriately and broadcast it to the QPU.

The WISE reconfiguration network implements odd-even transposition sorting~\cite{OddEvenSort_Habermann1972, Knuth1998}. WISE labels a row of zones by alternating between odd and even, and then uses parallel \emph{odd/even horizontal} swaps between adjacent ions within rows and \emph{odd/even vertical} swaps between adjacent ions that are both right next to a junction, as shown in Figure~\ref{fig:wise-chip}.
Vertical swaps differ because the $k$ ions in each trap share a junction and must cross it serially. Arbitrary reconfiguration of an $m \times n$ grid completes in at most $2m + kn$ swap steps~\cite{Malinowski2023}; at fixed $N=mn$ this is smallest, $\sqrt{8kN}$ steps, for the aspect ratio $2m=kn$.

\subsection{WISE Design Choices}
\label{sec:arch-params}

\paragraph{Multiplexing Order}
\label{sec:multiplexing-order}

In QCCD, electrodes (metal conductors) carry voltage signals and create electric fields that are used to control ions; digital-to-analogue converters (DACs) are used to generate waveforms to send to these control lines. Each zone contains $\approx10$ static electrodes for ion trapping and $\approx 10{-}20$ dynamic electrodes for ion movement.

WISE separates waveform generation from waveform distribution by routing a small set of shared waveforms through integrated cryogenic switching electronics rather than assigning dedicated control hardware to each electrode.

Dynamic transport of ions is provided by a small bank of globally distributed DACs, with local bit-select registers determining which waveform class drives each zone. Therefore, WISE has two primitives for each zone, either do nothing or swap with a neighbouring zone. We term this \emph{switched control}, which is unique to WISE. \emph{Multiplexed} control refers to quasi-static trapping voltages being supplied through time-division multiplexing. This means a single control line (\textit{e.g.,} DAC) is used to sequentially charge $M$ electrodes using an $M$-way multiplexor. This is the second half of WISE's design but it extends to other scalable trapped-ion control proposals~\cite{Anmasser:2026kpn, ohira2025multiplexed}.

Higher $M$ serialises electrode charging: when charging is not overlapped with other operations, as in our scheduler, every gate layer waits for its $M$ electrodes to be set sequentially, injecting a fixed latency overhead that lengthens every gate and increases exposure to dephasing.
Further, while electrodes are disconnected from their DAC, their stored voltages drift due to charge leakage, introducing control errors that grow with $M$.
Finally, all $M$ electrodes sharing a DAC experience identical voltage noise from that shared source. This can create spatially correlated errors across physically distant qubits, a failure mode that standard QEC analyses, which assume weakly correlated noise, do not capture. Our model uses this shared-control physics to set error rates, while sampling faults independently per gate (\S\ref{sec:live-state}). \textbf{Increasing $M$ reduces I/O requirements at the expense of longer gate layers, larger control errors and potentially correlated faults.}

\paragraph{Trap Capacity}
\label{sec:trap-capacity}
The trap capacity $k$ is the number of ions (qubits) housed in each trap. This parameter's choice controls a parallelism vs. communication trade-off similar to classical multi-cores~\cite{Jones_2026}. Smaller traps increase the number of independent execution sites, enabling more two-qubit gates to fire in parallel across the chip. However, fewer qubits per trap means more inter-trap data movement through junctions. Every gate between qubits in different traps requires physical ion transport, which is both slow and error-prone. Conversely, \textbf{larger traps reduce communication but pack ions closer together, increasing spectator crosstalk} (unintended coupling to non-target ions sharing the trap) between neighbouring qubits during gate execution~\cite{murali2020architectingnoisyintermediatescaletrapped}.

For WISE, $k$ introduces an additional asymmetric penalty not seen in QCCD. \textbf{Vertical reconfiguration, i.e., moving ions between rows, requires ions to cross through shared junction hardware, and this crossing must be serialised into $k$ sequential column buckets.}
Each additional qubit per trap adds another serial phase to every vertical routing step, yielding a vertical-to-horizontal cost ratio of $\approx 2.4k + k/2$, so QEC codes whose syndrome-extraction circuits demand frequent vertical communication are disproportionately penalised.

\paragraph{Cooling Protocol}
\label{sec:phonon-threshold}
Trapped-ion qubits accumulate excess motional energy as computation progresses (quantified as phonon number $\bar{n}$, which measures how vigorously each ion vibrates in its trap; higher phonon numbers degrade gate fidelity)~\cite{Monroe1995, turchette2000heating}.
Because two-qubit gate fidelity degrades with increasing $\bar{n}$~\cite{murali2020architectingnoisyintermediatescaletrapped}, we must periodically apply cooling pulses to return ions to near their motional ground state.

We parametrise the cooling policy by a phonon threshold $\bar n_\mathrm{th}$. With $\bar n_\mathrm{th}=0$, a cooling step precedes every gate block, mirroring the frequent recooling of Quantinuum's H-series and Helios systems~\cite{racetrack_Moses2023, helios2025}. With $\bar n_\mathrm{th}>0$, a per-ion phonon tracker estimates the motional excitation accumulated through transport, gates and idle periods, and triggers chip-wide cooling when an ion chain exceeds the threshold. \textbf{Raising this threshold eliminates unnecessary cooling stalls and directly reduces QEC cycle time, but potentially at the cost of higher motional excitation and lower fidelity.}

\subsection{QEC Design Choices for WISE}
\label{sec:QEC-trade-offs}

\paragraph{2D local codes (high qubit cost).}
Surface codes~\cite{Fowler2012} involve only nearest-neighbour qubits, making them the cheapest to route on WISE. Colour codes~\cite{Bombin2006} obey the same locality constraint but offer richer logical gate support (transversal Cliffords), at the cost of moderately deeper error-correction circuits. However, the BPT bound~\cite{BPTbound2010}, $kd^2=O(n)$ for 2D geometrically local stabiliser codes, forces their encoding rate to drop as $O(1/d^2)$ with increasing protection strength~$d$, meaning qubit costs grow rapidly.

\paragraph{2D subsystem codes (reduced check weight).}
Subsystem codes such as Bacon-Shor~\cite{BaconShor2006} measure weight-2 gauge operators rather than weight-4 stabilisers. Each weight-2 check requires only two two-qubit gates per ancilla (compared to four for a surface-code check), reducing individual circuit depth per measurement round and simplifying local routing. However, each weight-$2d$ stabiliser of a distance-$d$ Bacon--Shor code is inferred by combining $d$ weight-2 gauge outcomes, which introduces gauge-freedom overheads and lowers fault-tolerance thresholds~\cite{AliferisCross2007, NappPreskill2013}. Whether these shorter measurement circuits yield faster QEC cycle times that offset weaker per-round error suppression is an open co-design question that depends directly on the routing fabric and hardware noise mechanisms.

\paragraph{High-rate qLDPC codes (low qubit cost, long-range connectivity).}
At the opposite extreme, qLDPC codes (e.g., bivariate bicycle~\cite{bravyi2024highthreshold}, hypergraph product~\cite{TillichZemor2014}) circumvent the BPT bound via long-range interactions, encoding $5$--$10\times$ more logical qubits per physical qubit. While less practical on fixed 2D grids, mobile-qubit architectures natively support these non-local stabiliser graphs. Indeed, a recent QCCD study demonstrated that a $[[102, 22, 9]]$ qLDPC architecture can support 110 logical qubits using just $2514$ physical qubits at $\approx 10^6$ T gates per day~\cite{IonQWalkingCat}. Likewise, resource estimates for neutral atoms suggest that high-rate qLDPC codes could reduce factoring requirements to as few as $10{,}000$ reconfigurable atoms~\cite{cain2026shorsalgorithmpossible10000}.

However, $\sim 10\,\mu\text{s}$ transport latencies~\cite{IonQWalkingCat} and predicted $\sim 100$-day execution times~\cite{cain2026shorsalgorithmpossible10000} call into question whether these theoretical gains survive under realistic hardware, and in our case, multiplexed control constraints. %
\textbf{Overall, the optimal code family depends on the interplay between routing overhead (governed by the controller, the compiler, and routing fabric), noise accumulation, and per-round error suppression (dictated by code structure).}

\section{SAT-based WISE Compiler}
\label{sec:compilation}

We develop a WISE-specific compiler that performs exact, SAT-based qubit mapping and routing (QMR) for QEC circuits. Existing QCCD compilers~\cite{Jones_2026, Cyclone_khan2025cyclone, muzzletheshuttle, Sivarajah_2020, TISCC_Leblond_2023, murali2020architectingnoisyintermediatescaletrapped, MOVELESS_khan2025movelessminimizingoverheadqccds, GeneratingCompilersQMR, russon2026scalingqubitmappingrouting, SSYNC_QCCD, Schieretal2026,AdaptiveRouting2026Ovide} target different transport interfaces and do not support WISE's global odd-even routing primitives. TrapSIMD~\cite{ruan2025trapsimdsimdawarecompileroptimization} also exploits SIMD-style grouped transport on 2D trapped-ion architectures, through heuristic aggregation and scheduling of junction-based shuttling; WISER instead encodes WISE's broadcast odd-even swaps exactly in SAT.

We chose a SAT approach to reduce compiler quality as a confounding factor in architecture-level comparisons, by obtaining schedules that are optimal within our routing model, or bounded when the search times out. QMR is NP-hard in general~\cite{garey1979computers,Molavi2022MaxSATQMR}, but recent SAT formulations make exact compilation practical for many quantum circuits and improve substantially on heuristics~\cite{Molavi2022MaxSATQMR,jakobsen2025depthoptimalquantumlayoutsynthesis, shaik2024optimallayoutsynthesisdeep, yang_et_al:LIPIcs.SAT.2024.29}. However, existing SAT formulations do not encode global routing models. In WISE, a broadcast selects the direction and parity of a pass for the whole chip, while local switch bits choose which eligible comparators fire, so the swaps of a pass are globally coupled; prior solver-based formulations~\cite{Molavi2022MaxSATQMR, jakobsen2025depthoptimalquantumlayoutsynthesis} assume independently addressable SWAPs and optimise SWAP-count/depth. Moreover, WISE's vertical-to-horizontal cost asymmetry means minimising swap count is not time-optimal; our objective minimises weighted routing time under the odd-even constraints, and the scheduler then applies the multiplexed-control timing (\S\ref{sec:scheduling}). Consequently, WISE requires a new formulation that jointly models global odd-even routing and asymmetric transport costs.

\paragraph{Overview}
The compiler takes a QEC circuit, code distance $d$, trap capacity $k$, and hardware timing parameters as input, and produces a schedule of operations (native two-qubit MS gates, transport, reconfigurations).
Our QEC circuits represent memory experiments, i.e., they test the ability to preserve encoded logical qubits over time and consist of $C$ repeated rounds of syndrome extraction (typically $C=d$).
QEC circuits are represented in the Stim format~\cite{stim_Gidney2021}. The compiler decomposes the Stim circuit into native operations supported by WISE (\S\ref{sec:gate-decomposition}), then the SAT solver finds mapping and routing that co-locates interacting ion pairs while minimising the weighted routing cost (\S\ref{sec:sat-routing}). The scheduler interleaves reconfiguration passes with gate operations, inserting cooling when the tracked phonon number exceeds $\bar{n}_{\text{th}}$. Figure~\ref{fig:schedule-timeline} shows an example of a schedule for the $d{=}2$ rotated surface code on a $k{=}2$ WISE grid.

\subsection{From Stim Circuit to Native Operations}
\label{sec:gate-decomposition}

The Stim~\cite{stim_Gidney2021} circuits for each code we study here are generated from the \texttt{qldpc} Python library~\cite{qldpc}. The compiler first rewrites each circuit into Stim's $\{\texttt{H},\texttt{S},\texttt{CX},\texttt{M},\texttt{R}\}$ basis, so that, for example, an $X$-basis measurement becomes a $Z$-basis measurement conjugated by Hadamards, and then lowers these gates to native WISE operations: $\mathcal{G}_{\text{native}} = \{\text{MS}, R_X, R_Y, R_Z, M_Z, R\}$.
This is a straightforward IR transformation in our compiler (Appendix~\ref{appendix:gate-model} gives the rules). Stim then simulates the original Clifford circuit; the lowered sequence fixes the number, duration and error of the native operations charged to each gate. Since execution time is a bottleneck in WISE, we use one ancilla per QEC stabiliser unlike prior works which reuse ancillas~\cite{MOVELESS_khan2025movelessminimizingoverheadqccds}.

\subsection{SAT-based Qubit Mapping and Routing}
\label{sec:sat-routing}

To map and route QEC qubits onto physical qubits, this stage of the compiler accepts a sequence of $R$ MS-gate rounds as input. Each round is specified by the set of ion pairs that must be entangled $\mathcal{P}_r = \{(i_1, j_1), (i_2, j_2), \ldots\}$. The compiler produces a reconfiguration schedule and mapping where a sequence of parallel swap layers is determined such that for each round, the ion pairs involved in that round are placed within the same trap, allowing the two-qubit gates to happen on them.

\paragraph{Variables and Clauses}
\label{sec:sat-variables}
Intuitively, the SAT encoding answers: ``at each time step, where is every ion, and which swaps fire?'' For an $n\times m$ grid with $I$ ions and $R$ MS-gate rounds, $\ell$ indexes routing passes up to the horizon $\ell_{\max}=2(n+m)$. Three variable families encode the full ion trajectory: \emph{placement} variables $a(r,\ell,y,x,i)$ (ion $i$ occupies cell $(y,x)$ after pass $\ell$ of round $r$), \emph{swap} variables $s_h,s_v$ (a horizontal or vertical comparator exchanges neighbouring ions), and \emph{phase} variables $\phi(r,\ell)$ selecting horizontal or vertical swaps; helper variables derive the per-round pass counts $n_H^{(r)}$ and $n_V^{(r)}$ from them. Clauses enforce three invariants: (i) ions form a permutation (no ion is cloned or lost), (ii) swaps respect WISE's global odd-even parity and phase rules, and (iii) paired ions end each round co-located in the same trap. Table~\ref{tab:sat-clauses-structural} gives the full definitions (with $p$ for $\ell$).

\paragraph{Objective.}

Our goal is to minimise the total time ions spend being reconfigured between MS-gate rounds. With $t_H$ and $t_V$ denoting the costs of one horizontal and one vertical odd/even sub-pass, the total reconfiguration time is:
\begin{equation}\label{eq:reconfig-objective}
\sum_{r=1}^{R} T_r
=\sum_{r=1}^{R}\left(n_H^{(r)}t_H+n_V^{(r)}t_V\right).
\end{equation}
Our compiler jointly optimises all $R$ rounds of a complete stabiliser cycle in a single SAT call, choosing ion placements and the horizontal and vertical pass counts to minimise the total weighted routing cost.

The costs $t_H$ and $t_V$ follow from primitive transport timings in units of the ion-swap time $t_0$ (Appendix~\ref{append-optimality}). A horizontal swap is five sequential primitives (shuttle--merge--rotate--split--shuttle~\cite{walther2012controlling,bowler2012coherent,Burton_2023}), giving $t_H=2.12\,t_0$ per sub-pass; a vertical sub-pass serialises $k$ junction-crossing buckets of $5.10\,t_0$, giving $t_V=k(5.10\,t_0)+(k/2)t_H$. Thus $t_V/t_H\approx2.4k+k/2$: simply minimising raw swap count would ignore WISE's substantial vertical-routing penalty.

Because the SAT solver requires integer arithmetic, we discretise the objective as $W_r=n_H^{(r)}+w n_V^{(r)}$, where $w=\operatorname{round}(t_V/t_H)$ captures the vertical-to-horizontal cost asymmetry. The search variable $B$ is an upper bound on the total weighted routing cost. A totaliser-based cardinality constraint encodes $\sum_r W_r\leq B$~\cite{PySAT2018,RC2_Ignatiev2019,EenSorensson2003}, and parallel search tightens $B$. At convergence, a feasible schedule at $B^*$ and an UNSAT proof at $B^*-1$ certify that no schedule of the encoded model has a lower integer-weighted cost (Appendix~\ref{append-optimality}). The encoded model orders each round's passes horizontal-then-vertical within a horizon of $2(n+m)$ passes. This restriction never makes a round infeasible: for the even trap capacities we study, one horizontal phase followed by one vertical phase can always co-locate every MS pair (Appendix~\ref{append-optimality}). It can, however, exclude cheaper schedules that return to horizontal passes, so our certificates are relative to this phase order.

\paragraph{Cycle-Reuse Decomposition.}
\label{par:qmr-chaining}

Encoding all gate rounds in a single SAT formula is intractable at practical code distances. QEC memory experiments repeat a short cycle of MS rounds $C$ times (typically $C=d$), so the first SAT call encodes one copy of the cycle plus any start rounds, with a route-back round that returns the layout to the cycle start so the solution can be replayed $C$ times; cycle rounds carry weight $C$ in the objective. A second call routes the end rounds. The number of encoded placement slices is therefore independent of $C$; only the weighted cardinality constraint grows with it (Appendix~\ref{si:cycle-reuse}). The route-back round makes the cycle layout periodic, so certificates are also relative to this periodicity.

To inform architecture design, WISER needs to compile large QEC codes and code distances. We implemented a parallel SAT solving framework whose workers simultaneously tighten feasible upper bounds and prove lower bounds. QEC-specific spatial and temporal decomposition accelerates this search; the best schedule and proven lower bound are retained when the time budget expires. These tooling optimisations, motivated by prior exact mappers~\cite{shaik2024optimallayoutsynthesisdeep,yang_et_al:LIPIcs.SAT.2024.29}, are described in Appendix~\ref{appendix:hunters}.

\subsection{Physics-Aware Scheduling}
\label{sec:scheduling}
\begin{figure*}
    \centering
    \figorbox[width=0.85\linewidth]{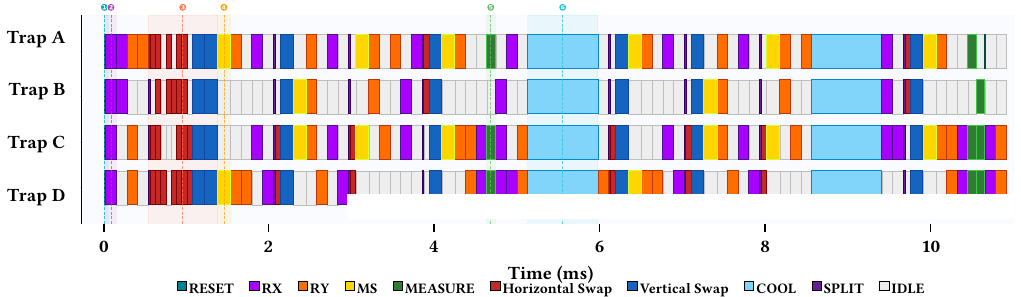}
    \caption{Compiled schedule for the $[[4,1,2]]$ rotated surface code (trap capacity $k{=}2$, $M{=}128$, Current tier, phonon threshold $\bar{n}_{\text{th}}{=}1.5$, a threshold-triggered policy chosen to illustrate mid-schedule cooling insertion), two syndrome-extraction rounds. Per-trap Gantt chart (trap labels on the y-axis): horizontal-swap (red) and vertical-swap (blue) reconfiguration passes interleave with type-homogeneous gate blocks. Light-blue cooling bars appear when $\bar{n}$ exceeds $\bar{n}_{\text{th}}$; grey bars are idle dephasing windows. MS is the native two-qubit gate. \textbf{Because control is shared, operations of one type execute together across traps, and every gate layer waits for multiplexed electrode charging. Even at this relaxed threshold, two chip-wide cooling stalls take about $15\%$ of the schedule; with cooling before every gate block ($\bar n_\mathrm{th}=0$), recooling takes about $30\%$ of the cycle (\S\ref{sec:result_isolating_wise_slowdown}).}}
    \Description{Gantt chart of four traps over about 11 milliseconds showing swap passes, single-qubit rotations, MS gates, measurements, two cooling stalls and idle periods.}
    \label{fig:schedule-timeline}
\end{figure*}

The scheduler processes the mapped and routed Stim circuit. It inserts cooling operations by tracking motional excitation rates for each trap based on the device's noise models operation by operation in a compiled Stim circuit.
If any trap's phonon number exceeds the threshold $\bar{n}_{\text{th}}$ before an MS gate or rotation, a cooling operation is inserted for \emph{all} traps simultaneously (light-blue bars in Fig.~\ref{fig:schedule-timeline}).
Cooling resets the motional excitation state of the trap.

WISE uses single-instruction, multiple data (SIMD) control: the same type of control pulse is applied across traps simultaneously. Accordingly, the scheduler also batches operations into blocks that are composed of the same operation type. In Figure~\ref{fig:schedule-timeline}, these SIMD operations are shown as colour-coded bands in the same timestep. We take the logical cycle time $T_\mathrm{logical}$ to be the duration of the full scheduled memory experiment, including state preparation and final readout (\S\ref{sec:metrics}).

\section{WISER: Architectural Model}
\label{sec:arch-model}

This section presents the WISER architectural model: how the hardware parameters $(N, k, M)$ and fabrication quality determine grid geometry, gate timing, and physical error rates.

\subsection{Architecture Instantiation and Grid Optimisation}
\label{sec:arch-instantiation}
Given an ion count $N$ and trap capacity $k$, the grid dimensions $m \times n$ must be chosen to minimise reconfiguration time.
Because vertical swap passes are physically ${\sim} 3k$ times more expensive than horizontal ones (\S\ref{sec:sat-routing}), a naive square grid selection will be suboptimal.
A candidate grid must satisfy $m \times n \geq N$ and $m \equiv 0 \pmod{k}$ (so that junction chains align with the trap structure). We search candidates near the idealised column count $m_{\text{ideal}} = \sqrt{N \cdot v_w}$, where $v_w = 3k$ is the vertical-to-horizontal cost multiplier and select the grid minimising $\text{cost}(m,n) = m + v_w \cdot n,$ which approximates worst-case reconfiguration depth.
Ties are broken by fewest spectator (unused) cells $m \cdot n - N$. The result is wider-than-square grids that exploit the cheaper horizontal axis. Power accounting (\S\ref{sec:metrics}) counts every zone of the chosen grid, including unused ones.

\subsection{Gate Time Model}
\label{sec:gate-time}
In our scheduler, every gate layer waits for the multiplexed control hardware to charge its $M$ electrode groups, without overlapping this charging with other operations. This sample-and-hold overhead adds a fixed latency $t_{\text{sc}} = M \cdot t_{\text{ec}}$ to every gate layer~\cite{Malinowski2023}, giving effective gate times $t_{\text{eff,1q}} = t_{\text{sc}} + t_{\text{1q}}$ and $t_{\text{eff,2q}} = t_{\text{sc}} + t_{\text{2q}}$.
A layer therefore takes between $t_\text{eff,1q}$ (single-qubit gates, no reconfiguration) and $t_\text{eff,2q} + t_r + t_\text{cool}$ (an MS layer preceded by reconfiguration and cooling), where~$t_r$ is the reconfiguration time computed by the compiler for a given schedule. Overlapping charging with other operations, as discussed for WISE~\cite{Malinowski2023}, would shorten these layers and is not modelled.

\subsection{Noise Modelling}
\label{sec:noise-model}

We implement a physics-informed noise model that traces the error rates of gates, idling, transport and cooling back to the architectural parameters $(N, k, M)$ and the technology tier, with single-qubit, measurement and reset errors as tier-dependent constants. Unlike prior studies that either assume simple depolarising noise~\cite{beverland2022assessingrequirementsscalepractical,vandam2024usingazurequantumresource,IonQWalkingCat} or highly abstracted hardware-derived noise~\cite{transversality_lattice_Guti_rrez_2019, murali2020architectingnoisyintermediatescaletrapped, Jones_2026}, our model captures how errors arise from and scale with low-level hardware parameters.

The hardware-level error mechanisms (heating, shared-control noise, crosstalk) are derived from device physics and projected onto Pauli channels, following the convention of parametrising depolarising noise from gate infidelities~\cite{wallman2018randomized,yang2026hardware}, extended here with WISE-specific shared-control noise. Fully calibrated channel models would require experimental device characterisation or intractable simulation~\cite{Gustiani2025virtualquantum}. We summarise the dominant noise mechanisms and their scaling behaviour below. Exact equations and detailed physics are in Appendix~\ref{append-noise}.

\textbf{Motional heating.}
Electric-field noise from the trap surface continuously drives the ions' vibrational modes, increasing the phonon number over time.
We assume a WISE layout in which ions form a linear chain with normal modes. The heating rate scales as $\Gamma_m = \frac{e^2}{4m_{\mathrm{ion}}\hbar \omega_m} \sum_l |\mathcal{E}_{l\to m}|^2 S_{E_l}(\omega_m)$, with ion mass $m_{\mathrm{ion}}$ and $S_E(\omega) \propto 1/\omega^\beta$, as shown in~\cite{brownnutt2015ion}, where the coupling $\mathcal{E}_{l\to m}$ of electrode $l$ to mode $m$ (Appendix D) sets how voltage noise on an electrode becomes motional excitation or a frequency shift. Architecturally, the heating rate determines how often cooling must be applied. In the evaluated model, phonons accumulate at a fixed reference heating rate; the heating-amplitude parameter enters only through motional dephasing (Appendix~\ref{append-noise}).

\textbf{Operational phonon generation.} Noise in the operation resonantly drives ion motion during transport, junction crossing, and swapping, resulting in residual phonons. Removing them requires recooling, and each recooling operation can scatter a photon from a data ion and decohere it, so the total cooling error grows with the number of recooling operations.

\textbf{Correlated noise.} Unintended entanglement between non-target ion pairs arises from shared motional modes during parallel gates. The crosstalk phase scales proportionally to the product of the Lamb-Dicke parameters (defined in Appendix~\ref{append-noise}) between different ions in the trap. For WISE's zone architecture, this limits parallel gate density: non-interacting qubit pairs must be assigned to different traps to reduce the correlation noise.

\textbf{Device control noise (shared DACs).} Charge injection from switch openings and voltage decay on the control lines create correlated electric field errors. The voltage step scales with the extra charge fluctuation, and drops exponentially in time. Because WISE shares DACs across multiple controls, this noise couples across the entire zone~\cite{ohira2025multiplexed}.

\textbf{Single-qubit, measurement, and reset penalties.} These arise from residual physical errors not explicitly modelled at the physical level, with state-of-the-art single-qubit infidelity as low as $10^{-7}$ per operation~\cite{smith2025single}. Architecturally, these constant penalties set the baseline below which all other noise sources (heating, crosstalk, device noise) must be suppressed to avoid dominating the total infidelity.

\subsection{Per-Gate Error Computation} \label{sec:live-state}
Unlike prior QEC studies that assign fixed error rates per gate type, our model evaluates each gate at the live phonon state of its trap at the instant it executes. Each trap's phonon number is updated in-place throughout the schedule. For instance, transportation of the ions will increase it, while cooling of the ion resets it (so two identical MS gates at different schedule points have different fidelities). Each MS gate's infidelity is decomposed into the noise contributions described in \S\ref{sec:noise-model}, evaluated at the current phonon number, and mapped to Pauli error channels. For example, device-control, heating and photon-scattering errors are mapped onto a two-qubit depolarising channel (which applies one of the 15 non-identity two-qubit Paulis, chosen uniformly, when a fault occurs), motional dephasing onto $Z$ errors on the gate ions, and crosstalk onto $Z$ errors on spectator ions~\cite{Leung_2018, Figgatt_2019}. Every channel is injected independently per gate: shared-DAC physics sets the rates, but joint faults across traps sharing a DAC are not simulated.

Single-qubit rotation gates (\textit{e.g.,} $R_X$, $R_Y$) incur heating and device-control errors computed at the current $\bar{n}$, combined into a single-qubit depolarising channel. Measurement operations incur a fixed bit-flip error $p_M$ (technology-tier-dependent) modelled as an $X$ error before the projective $Z$-basis measurement, so that the classical outcome is flipped with probability $p_M$; $X$-basis measurements are executed as $Z$-basis measurements between Hadamards (\S\ref{sec:gate-decomposition}), so the same flip applies to them. Qubit reset incurs a fixed preparation error $p_R$ modelled as an $X$ error after the reset, representing imperfect state initialisation. Each recooling operation applies a single-qubit depolarising error of total probability $\epsilon_{\text{cool}}$ ($\epsilon_{\text{cool}}/3$ on each Pauli axis) to every ion in the cooled trap, representing decoherence from the cooling laser interaction.

\subsection{Stim Noise Injection}
\label{sec:noise-injection-detail}
The final pipeline stage translates the compiled schedule into a noisy Stim circuit~\cite{stim_Gidney2021} for QEC simulation. The noise injector walks through the schedule and, for each operation, (i) computes the idle duration each ion experienced since its last operation (grey bars in Fig.~\ref{fig:schedule-timeline}), converting idle time to $Z$-dephasing errors from ion-chain frequency fluctuations and voltage drift, (ii) inserts pre- and post-gate Pauli errors at the rates computed by the noise model (\S\ref{sec:noise-model}), and (iii) attaches cooling-induced errors where cooling was inserted by the scheduler. Idling during transport accumulates motional heating rather than dephasing, which is tracked separately.

We sample syndrome histories with Stim and decode each detector error model with MWPM~\cite{PyMatching_Higgott2022} for surface codes and, for other codes, exact maximum-likelihood lookup~\cite{Poulin2006MLE} up to 30 detectors, Tesseract~\cite{beni2025tesseractdecoder} up to 60 and BP-OSD~\cite{Roffe2020,ldpc_python} beyond, with fewer iterations and lower OSD order as detector count and syndrome load grow. A shot fails if any logical observable of the block is decoded incorrectly, so $p_L$ is a block-failure probability and, for multi-qubit blocks, upper-bounds the error rate of each logical qubit. Sampling stops at an error count or at $1$--$2\times10^9$ shots; we report one-sided $95\%$ Clopper--Pearson upper bounds for each configuration, $3.0\times10^{-9}$ for no failures in $10^9$ shots, and a configuration meets a $p_L$ target only if this bound lies below it. Configurations whose average detector-flip probability exceeds $5\%$ are classed as not meeting any target without sampling. Appendix~\ref{si:evidence} lists the failure and shot counts behind every headline result.

\section{Experimental Setup}
\label{sec:evaluation}

We systematically vary architectural parameters with hardware/noise assumptions in order to determine the best design choices and whether these choices remain consistent across plausible operating regimes.

\subsection{Control Parameter Space}
\label{sec:control-params}
\label{sec:resimulation}
The WISE architecture exposes \emph{compile-time parameters} (trap capacity $k$ and grid geometry) that require SAT recompilation, and \emph{post-hoc parameters} (multiplexing order $M$, technology tier, and phonon threshold $\bar{n}_{\mathrm{th}}$) that affect timing and noise but not routing (\S\ref{sec:arch-params}). Changing a post-hoc parameter re-runs only the scheduler, with the same routing but updated noise and cooling decisions, before the noisy Stim circuit is re-sampled and decoded.

\subsection{Code Selection Space}
Code selection follows the architectural regimes of \S\ref{sec:QEC-trade-offs}. We evaluate 25 code instances across five families from distance 2 to 9 (Figure~\ref{fig:SAT-vs-heuristic}); the full parameter sweep covers 24 of them, all except the distance-9 surface code. This covers the code-block sizes relevant to early FTQC WISE system operation.

\textbf{(1)~surface codes}~\cite{Fowler2012} (rotated layouts; the $2d^2-1$ allocation includes data and syndrome ancillas).

\textbf{(2)~colour codes}~\cite{Bombin2006} (4.8.8) tessellations.

\textbf{(3)~subsystem codes}~\cite{BaconShor2006} (Bacon--Shor).

\textbf{(4)~small codes} ($[[4,2,2]]$, $[[6,2,2]]$, Steane, Shor, $[[8,3,2]]$, $[[15,1,3]]$ Reed--Muller, Hamming CSS), motivated by recent trapped-ion demonstrations of real-time QEC~\cite{small_code_colour_realtime, ft_one_bit_addition_small, berthusen2024experiments4dsurfacecode, microsoft_quantinuum_paper}.

\textbf{(5)~qLDPC codes}~\cite{bravyi2024highthreshold, TillichZemor2014, Panteleev2022, Breuckmann2021}: bivariate bicycle $[[18,8,2]]$, $[[36,8,4]]$ and $[[72,12,6]]$ and hypergraph products of the Hamming code.

\subsection{Sweep Methodology}
\label{sec:sweep-method}
Rather than reporting results at a single operating point, we sweep code and distance, trap capacity $k\in\{2,4,6,8,16\}$, multiplexing order $M\in\{1,4,16,64,256\}$, cooling policy $\bar n_{\rm th}\in\{0,0.1,0.5,1,10,100\}$ and technology tier. Nineteen code instances cover this full grid; the three largest blocks (surface $d=7$, colour $d=6$ and HGP $7\times7$) are evaluated at $k=2$, Reed--Muller at $k\leq6$, and Bacon--Shor $d=9$ at $k\leq6$ and $65$ points with $k\geq8$. This gives $9{,}425$ configurations, of which $7{,}230$ were sampled, $2{,}074$ exceeded the detector-noise cut-off (\S\ref{sec:noise-injection-detail}) and $121$ Bacon--Shor $d\in\{7,9\}$ points were excluded because their decoder runs did not terminate. The Current, Optimistic, and Pessimistic tiers vary gate timing, transport latency, and physical error parameters (Table~\ref{tab:tech-tiers}).

\subsection{Reported Metrics.}
\label{sec:metrics}

We define one \textit{logical cycle} as a memory experiment of $d$ complete rounds of syndrome extraction. Its \textit{logical cycle time} $T_{\mathrm{logical}}$ is the duration of the full schedule, including state preparation and final readout, so a round lasts $T_{\mathrm{ec}}=T_{\mathrm{logical}}/d$ on average, and the \textit{logical clock speed} is $f_L=1/T_{\mathrm{logical}}$. The \textit{logical error rate} $p_L$ is the probability that a logical cycle ends with any logical observable of the block decoded incorrectly. 

$N$ is the number of zones in a code block, each holding at most one ion (Figure~\ref{fig:wise-chip}). Each DAC dissipates $30\,\mathrm{mW}$~\cite{Malinowski2023}. An $M$-way demultiplexer reserves one of its states for disconnection, so a static-electrode DAC serves $M-1$ electrodes and a block needs $N_{\mathrm{sDAC}}=10N/(M-1)$ static DACs; $M=1$ denotes one DAC per static electrode, $N_{\mathrm{sDAC}}=10N$. The $N_{\mathrm{dDAC}}=120$ dynamic-waveform DACs broadcast to every zone of the chip~\cite{Malinowski2023}, so identical blocks running the same schedule in lockstep share one bank. For $L=100$ logical qubits in blocks encoding $K_L$ each, the DAC power per logical qubit is $P_{\mathrm{LQ}}=30\,\mathrm{mW}\times(N_{\mathrm{sDAC}}/K_L+N_{\mathrm{dDAC}}/L)$, with all DACs at the cold stage; we refer to it as power throughout.

\paragraph{Optimistic screening criteria for early fault tolerance}
\label{sec:early-ft-screen}

We screen for configurations that could achieve early fault-tolerance. We define this as preserving $100$ fault-tolerant logical qubits through $10^5$ logical time steps, modelling each step as one logical cycle. This is a deliberately lenient screening criterion, with a useful quantum computation expected to require a stronger screen~\cite{beverland2022assessingrequirementsscalepractical,Campbell2022,Huggins2025FLASQ}. It sits between the UK National Quantum Strategy's milestones of a million quantum operations by 2028 and a billion by 2032~\cite{UKGovQuantumMissions2023}.

We allocate a $10\%$ failure budget~\cite{beverland2022assessingrequirementsscalepractical,PRXQuantum.5.020101,Yoshioka2024} to preserving the entire register and allow at most $48\,\mathrm{h}$~\cite{Pino2021, racetrack_Moses2023, oxfordionics_Loschnauer2024} to finish the logical circuit on the quantum processor. Splitting the budget evenly across qubits and steps (a union bound) gives the target $p_L\leq\tfrac{0.1}{100\times10^5}=10^{-8}$, a sufficient rather than strictly necessary condition; the window requires $f_L \geq \tfrac{10^5}{48\times3600}\,\mathrm{Hz}\simeq0.58\,\mathrm{Hz}$. We apply the same $10^{-8}$ target to the block-failure probability of every code, which is conservative for blocks encoding $K_L>1$ logical qubits: splitting the budget over $\lceil100/K_L\rceil$ blocks instead permits $p_L\leq0.1/(\lceil100/K_L\rceil\times10^5)$, e.g.\ $1.1\times10^{-7}$ for nine $[[72,12,6]]$ blocks. Repeating our screens with this register-level target adds only four slower BB configurations at the Optimistic tier and none at the Current tier, so none of our conclusions depend on the choice.

We assume the cryogenic refrigerator at $\sim 10\,\mathrm{K}$ can remove a total of $600\,\mathrm{W}$ while maintaining the QPU's operating temperature, an optimistic design assumption informed by large $20\,\mathrm{K}$ helium-refrigeration systems~\cite{Michael2024, Brock2023}, and allocate $275\,\mathrm{W}$ to optical absorption in the light-delivery waveguides. Appendix~\ref{sec:early-ft-screen-si} gives the supporting calculations and evidence for all of our assumptions here. This leaves $3.25\,\mathrm{W}$ per algorithmic logical qubit for electrical trapping and reconfiguration.

The screen is a minimum bar rather than a guarantee: passing it does not imply feasibility in practice. Every configuration is evaluated under favourable conditions, with a minimal early-fault-tolerant workload, SAT-optimised (or bounded) mapping and routing and strong tractable decoders, so configurations that do not pass are poor candidates for early fault tolerance on the studied hardware. Further, we intentionally use logical memory rather than logical computational gadgets, such as magic-state preparation and consumption, in order to reduce confounding factors in resource estimates.

\section{Results}
\label{sec:results}

\subsection{WISER Compiler Quality}
\label{sec:result1}
\textbf{Method:} To evaluate compilation quality, we use the architecture configuration of the original WISE study~\cite{Malinowski2023}, $(k,M,\bar{n}_{\mathrm{th}})=(2,128,0.0)$ on the Current tier. This experiment uses the strictest cooling policy: in addition to cooling before every gate block, ions are recooled during reconfiguration whenever transport has heated them, so routing time includes the recooling that movement causes; the recool before each gate block, which both compilers pay identically, is excluded from routing time but included in cycle time. The co-design sweep (\S\ref{sec:result2}) and the comparisons that follow it use cooling before every gate block only. For each QEC scheme, we compare the WISER compiler (\S\ref{sec:sat-routing}) with a greedy WISE-compatible baseline. The heuristic baseline adapts the shortest-path co-location policy used by existing QCCD compilers~\cite{murali2020architectingnoisyintermediatescaletrapped,MOVELESS_khan2025movelessminimizingoverheadqccds}, adapted to WISE's global swapping primitives. For a given round of 2-qubit entanglement (MS) gates, where qubits must be co-located in the same trap, each pair of qubits is assigned, closest-first, to the nearest free trap, and the ions are moved there by a three-phase odd-even transposition sort (rows, columns, rows), which realises any target permutation within $p_{\max} = 2(n+m)$ swap layers~\cite{OddEvenSort_Habermann1972,Knuth1998, Malinowski2023}. 

\textbf{Result:} Figure~\ref{fig:SAT-vs-heuristic} shows the ion routing time (lower is better) for both techniques. \textbf{WISER reduces routing time by $2.6$--$18.8\times$ relative to the heuristic compiler across 6--161-qubit QEC workloads ($8.24\times$ on average).} It also lowers logical cycle time by $1.9$--$9.7\times$ (lines in Figure~\ref{fig:SAT-vs-heuristic}).
Medium-sized codes with $10$--$50$ qubits benefit the most because non-local routing overhead in the greedy heuristic is high, while the overheads of SAT decomposition are low. Workloads of $100$ or more qubits still achieve $3.0$--$6.0\times$ routing reductions, showing that the formulation remains effective beyond the regime where optimality can be proved.

\begin{figure*}[htbp]
\figorbox[width=0.95\textwidth]{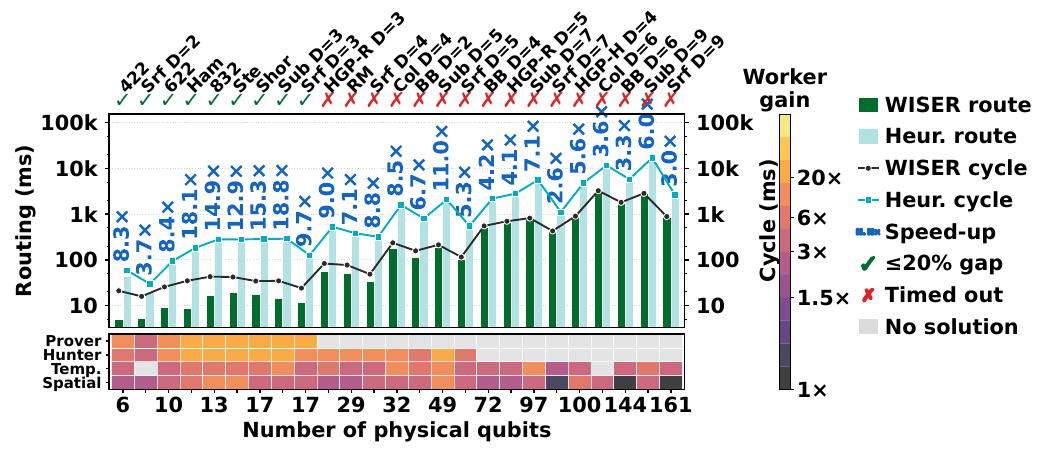}
\centering
\caption{\textbf{Top:} routing time (bars, left axis) and cycle time (lines, right axis) per syndrome-extraction round, WISER compiler vs.\ the greedy odd-even heuristic on identical hardware, at WISE's own operating point ($k=2$, $M=128$, Current tier, $\bar{n}_{\mathrm{th}}=0.0$ with recooling during reconfiguration). Labels give the code and its distance $D$. Blue labels give the routing speed-up; green ticks mark schedules within the proven lower bound, red crosses 12-hour runs with a larger gap. Routing time measures transport and the recooling that transport causes, while cycle time additionally includes gates, measurement, reset and the recooling before each gate block. \textbf{WISER cuts routing time by $8.24\times$ on average} (arithmetic mean over the 25 codes). \\ \textbf{Bottom:} routing improvement per SAT-portfolio worker class (Appendix~\ref{appendix:hunters}); grey cells are workers that found no solution in time. Provers only scale to ${\sim}17$ qubits within the 12-hour proof budget on an M4 Mac Mini. Hunters (branch-and-bound on the full problem) extend compiler scalability to 49 qubits, while spatial and temporal decomposition extends compilation to ${\sim}161$. \textbf{Our QEC-specific SAT-decompositions extend compilation to codes required for early fault-tolerance.}}
\Description{Bar and line chart of routing and cycle time for 25 QEC codes ordered by qubit count, comparing WISER with a greedy heuristic, above a heatmap of per-worker routing gains.}
\label{fig:SAT-vs-heuristic}
\end{figure*}

QEC-specific decomposition lets WISER compile memory circuits of up to 161 physical qubits, against roughly 17 for the Prover alone (Figure~\ref{fig:SAT-vs-heuristic}). Compiling larger monolithic code blocks is not necessary here: cycle time grows with block size, and at $M=16$ a single-logical-qubit block beyond ${\sim}160$ zones exceeds the power ceiling.

\subsection{Full Parameter Sweep}
\label{sec:result2}
\textbf{Method:} To determine optimal operating points, we do a series of hardware parameter sweeps. In Figure~\ref{fig:result2-stepC}, we sweep trap capacity $k$ and multiplexing order $M$ and estimate the logical clock speed and power per logical qubit, while meeting the logical-error target.
In Figure~\ref{fig:result2-stepC-fig2}, we test if our conclusions are robust to changes in model parameters by repeating $(k,M, \bar n_{\rm th})$ selection across a range of logical error rates and technology tiers.
From these sweeps, we fix the best trap capacity, multiplexing order, and cooling policy to compare code families in a fixed architecture in Figure~\ref{fig:qec_expert_trajectories}. Further, Figure~\ref{fig:q3_sensitivity} checks the sensitivity of the underlying physical-model parameters. Figure~\ref{fig:error-budget} identifies the dominant noise mechanisms, while Figure~\ref{fig:cycle-breakdown} determines the contributions of transport, recooling, multiplexing, and gates towards the logical cycle time.

\paragraph{Result:} Figure~\ref{fig:result2-stepC} shows that \textbf{the smallest ($k=2$, two-ion) traps, moderate ($M=16$) multiplexing and cooling before every gate block ($\bar{n}_{\mathrm{th}}=0$) give the fastest early fault-tolerant WISE design}, while higher multiplexing trades clock speed for lower power. Figure~\ref{fig:result2-stepC-fig2} repeats the selection across logical-error targets and hardware tiers: small traps and cooling before every gate block remain preferable, while the power benefit of high $M$ is limited by its clock-speed and error penalties.
At $k=2$, the fastest configuration meeting all three screens is the surface code $d=5$ at $M=16$, the BB $[[72,12,6]]$ code at $M=4$, and the Bacon--Shor code $d=5$ at $M\geq64$; the lowest power at every $M\geq4$ comes from the BB $[[72,12,6]]$ code, always with cooling before every gate block (Appendix~\ref{si:evidence}). No cell passes at $M=1$, and no configuration passes at $k\geq8$.

At $(k,M)=(2,16)$, optimistic tier, the fastest code, the surface code $d=5$, passes all three screens ($52.9$\,Hz, $1.12\,\mathrm{W/LQ}$), which is the highest jointly feasible clock speed in the sweep. At $M\leq4$ its static DACs exceed the power ceiling ($16.2$ and $5.4\,\mathrm{W/LQ}$); by our power model it fits the ceiling for any $M\geq7$. At $M\geq64$ its $p_L$ bound stays above the target within our shot budget (one failure in $1.2$--$1.5\times10^{8}$ shots). $M=16$ thus keeps $98.5\%$ of the unmultiplexed clock at $6.9\%$ of its power, and is the only evaluated order at which the fastest code meets every requirement.

\begin{figure*}[t]
\centering

\includegraphics[width=0.9\textwidth]{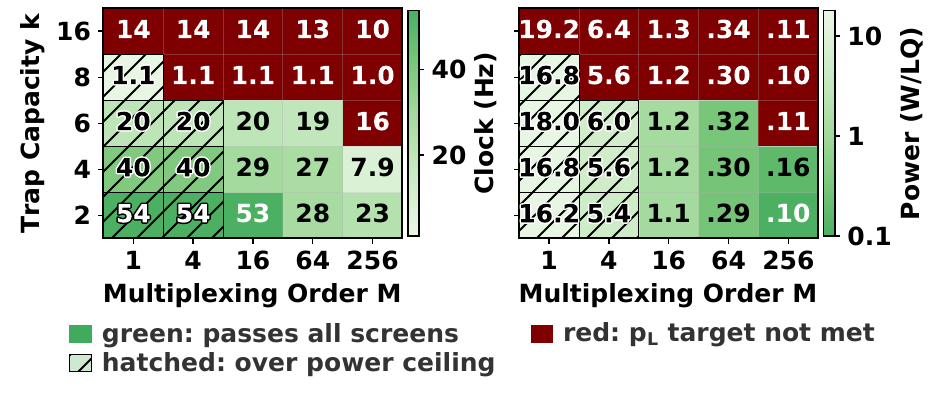}

\caption{Logical clock (left) and DAC power per logical qubit (right) of the fastest design meeting $p_L\leq10^{-8}$ in each cell of trap capacity $k$ and multiplexing order $M$, Optimistic tier, over (code, $d$, $\bar{n}_{\mathrm{th}}$); both panels describe the same configuration. Green: that configuration also meets the power ($\leq3.25$\,W per logical qubit) and clock ($\geq0.58$\,Hz) screens; hatched: it exceeds the power ceiling; dark red: no configuration meets the $p_L$ target, and the cell continues the row with the configuration of the nearest cell to its left that does (for $k=16$, the row's most reliable configuration). Along every row, clock and power both fall with $M$. At $(k,M)=(2,4)$ and $(4,4)$ a slower BB $[[72,12,6]]$ configuration also passes all screens; Table~\ref{tab:fig5-winners} lists every plotted configuration and the fastest and lowest-power passing ones. \textbf{$(k,M)=(2,16)$ is the only cell where the fastest code meets every requirement: the surface code $d=5$ at $53$\,Hz and $1.1$\,W per logical qubit. Below $M=16$ it exceeds the power ceiling, and above it Bacon--Shor $d=5$ takes over as its $p_L$ bound is no longer certified within our shot budget.} }
\Description{Two heatmaps over trap capacity and multiplexing order showing the best feasible logical clock speed and the lowest feasible power per logical qubit at the Optimistic tier.}
\label{fig:result2-stepC}
\end{figure*}

\begin{figure*}
    \centering
    \includegraphics[width=\textwidth]{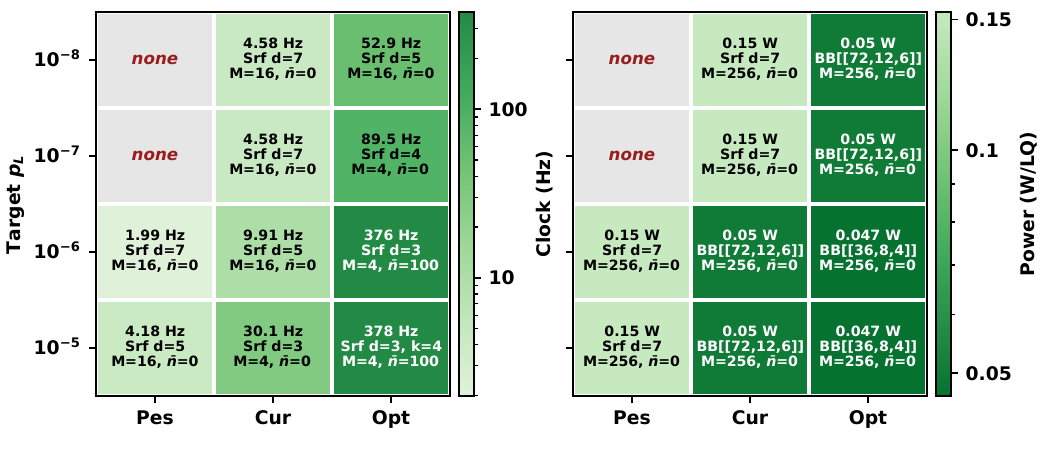}
    \caption{Best clock (left) and lowest DAC power per logical qubit (right) over all codes, distances $d$, trap capacities $k$, multiplexing orders $M$ and cooling policies $\bar{n}_{\mathrm{th}}$, per technology tier and target logical error rate, among configurations that also meet the power ($\leq3.25$\,W per logical qubit) and clock ($\geq0.58$\,Hz) screens. Each cell gives the winning value and configuration ($k=2$ unless stated). (Left) the winners for best clock are surface codes with moderate multiplexing. (Right) the winners for best power are high-rate or high-distance codes at high multiplexing. \textbf{Cooling before every gate block ($\bar{n}_{\mathrm{th}}=0$) wins 18 and $k=2$ wins 19 of the 20 feasible cells; no configuration meets targets of $10^{-7}$ and below at the Pessimistic tier, and $10^{-8}$ costs $11.6\times$ in clock at Current ($4.58$\,Hz) versus Optimistic ($52.9$\,Hz).} }
    \Description{Two grids over technology tier and target logical error rate giving the fastest and the lowest-power feasible configuration in each cell.}
    \label{fig:result2-stepC-fig2}
\end{figure*}

In Figure~\ref{fig:qec_expert_trajectories}, six codes reach $p_L\leq10^{-8}$ under the Optimistic tier, all within the cold-stage power ceiling: surface $d=5$ and $7$ ($52.9$ and $26.3$\,Hz), Bacon--Shor $d=5$ and $7$ ($29.2$ and $11.6$\,Hz), HGP $5\times5$ ($22.0$\,Hz) and the $[[72,12,6]]$ BB code ($11.0$\,Hz, $0.28\,\mathrm{W/LQ}$).  \textbf{At the Current tier only surface $d=7$ reaches $10^{-8}$ with $4.58\,$Hz and $2.04\,\mathrm{W/LQ}$}, and it continues to pass at $M=64$ and $256$ ($4.30$ and $3.46$\,Hz, down to $0.15\,\mathrm{W/LQ}$). At $M=16$ it would complete the $10^5$-cycle early-FTQC workload in $6.1$\,h, with $7.9\times$ headroom on the clock floor, $3.3\times$ on the $p_L$ bound and $1.6\times$ on the power ceiling, and by our power model it fits the ceiling for any $M\geq11$. Under the Optimistic tier the surface code $d=5$ completes the workload in $31$ minutes ($91\times$ clock headroom) and the BB code in $2.5$\,h with $12\times$ power headroom.

\begin{figure*}
\centering

\includegraphics[width=0.85\textwidth]{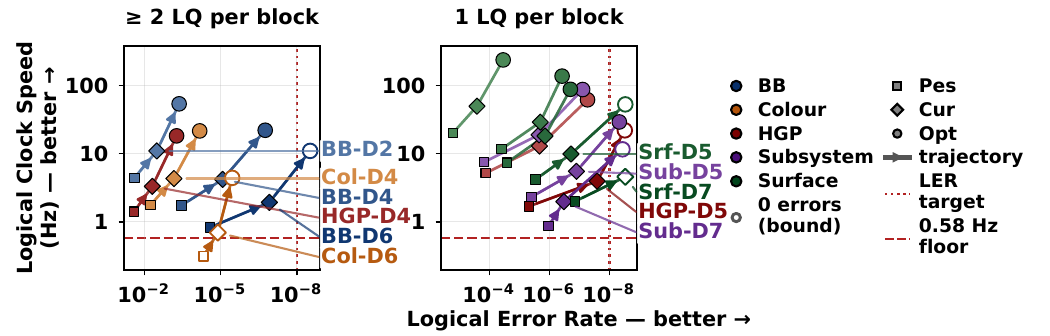}

\caption{Tier trajectories of each code (Pessimistic$\rightarrow$Current$\rightarrow$Optimistic tiers: square, diamond, circle) at fixed hardware settings $(k^*,M^*,\bar{n}_{\text{th}}^*)=(2,16,0.0)$. Left: codes whose blocks encode several logical qubits (BB, colour, HGP $7\times7$); right: single-logical-qubit codes, labelled only where they reach the target. Labels give the code family and distance $D$. Hollow markers are upper bounds from zero observed failures. The trajectories go from pessimistic (the lower-left of the axis, high error, low clock) to optimistic (the upper-right, low error, high clock). Advances in hardware improve reliability and throughput while largely preserving the ordering of the code families. Six codes at $(2,16,0.0)$ reach the target above the $0.58$\,Hz floor (\S\ref{sec:early-ft-screen}) in the Optimistic tier, but \textbf{only the distance-7 surface code meets all the requirements of \S\ref{sec:early-ft-screen} at the Current tier.}
}
\Description{Scatter of logical error rate against logical clock speed for each code across three technology tiers, with the target error rate and minimum clock marked.}

\label{fig:qec_expert_trajectories}
\end{figure*}

\paragraph{Do the benefits of WISE's all-to-all connectivity improve logical-qubit performance for QEC codes that exploit richer connectivity?}

Across all technology tiers, surface codes consistently achieve the highest logical clock speeds at a given logical error rate: every feasible clock winner in Figure~\ref{fig:result2-stepC-fig2} is a surface code, reflecting their comparatively low execution overhead despite WISE's all-to-all routing. In contrast, code families that exploit richer connectivity, particularly qLDPC codes, incur additional scheduling and reconfiguration overhead, reducing logical throughput. Consequently, all-to-all connectivity alone does not translate into faster logical execution.

\paragraph{Which QEC code families remain feasible under WISE constraints?}
Surface, Bacon--Shor, HGP and BB codes all have feasible design points. BB's high encoding rate cuts the power per logical qubit by $4\times$ relative to the surface code $d=5$ ($0.28$ vs.\ $1.12\,\mathrm{W/LQ}$) at about a fifth of its clock speed. Within the power ceiling, surface codes give the highest logical throughput and are the only family that passes at the Current tier.

\subsection{Finding Design Targets}
\label{sec:result3_dominant}
\textbf{Method:} We decompose the physical error budget of the $[[72,12,6]]$ BB code (144 ions) at the design point ($k=2$, $M=16$, $\bar{n}_{\mathrm{th}}=0.0$) by mechanism (\S\ref{sec:noise-model}), as the expected number of physical faults per syndrome-extraction round implied by the noise model (operation counts times per-operation error; Figure~\ref{fig:error-budget}). We separately vary each physical parameter over its tier range (Table~\ref{tab:tech-tiers}). The full sensitivity of each physical parameter on each noise mechanism can be found in Figure~\ref{fig:q3_sensitivity}. These budgets rank physical fault sources; they do not by themselves attribute logical failures to mechanisms.

\textbf{Result:} State preparation and measurement (SPAM) dominate the physical error budget both per operation and per syndrome extraction, carrying about $78\%$ of the expected faults of a round at the Current tier. SPAM and cooling errors also have the highest sensitivity between the Pessimistic and Optimistic tiers: moving the measurement or the reset error between its tier endpoints changes the expected faults per syndrome-extraction round by $8.9\times10^{-2}$ each, and the cooling error by $2.0\times10^{-2}$. Improving SPAM is therefore the most effective single design target: bringing measurement and reset alone from Current ($10^{-4}$) to Optimistic ($10^{-5}$) values removes $70\%$ of the expected physical faults of a round ($3.4\times$ fewer), whereas bringing every other mechanism to its Optimistic value removes only $19\%$. The remaining parameters follow, in descending order: axial frequency ($7.0\times10^{-3}$), gate
scattering ($5.2\times10^{-3}$), recool time
($3.5\times10^{-3}$), electrode charging time ($3.2\times10^{-3}$), MS gate time ($1.9\times10^{-3}$) and spectral exponent
($5.8\times10^{-4}$), with every remaining parameter below $10^{-4}$. The crosstalk parameters have no effect at $k=2$ because a two-ion trap holds no spectator ions.

\subsection{Comparison with locally controlled QCCD}
\label{sec:result3_qccd_comparison}

\textbf{Method: }On the same distance-7 surface code, noise model, Current tier, $p_L$ target and cooling policy as the sweep (cooling before every gate block), we compare the heuristic compiler at WISE's own point $(k,M,\bar n_{\mathrm{th}})=(2,128,0)$, our SAT compiler on the same hardware, our SAT compiler with $M=16$, and locally controlled QCCD ($M=1$, no DAC-sharing).

\textbf{Result:} On unchanged hardware, SAT compilation raises WISE's logical clock from $1.94$ to $3.98\,\mathrm{Hz}$ and lowers $p_L$ from $3.3\times10^{-6}$ to below the $10^{-8}$ target. Choosing $M=16$ raises the clock by a further $15\%$, to $4.58\,\mathrm{Hz}$, while keeping $p_L$ below the target and the power within the ceiling at $2.04\,\mathrm{W/LQ}$. Locally controlled QCCD reaches $12.8\,\mathrm{Hz}$, $2.8\times$ faster, but needs a DAC for every electrode: with ${\approx}10$ static and $10$--$20$ dynamic electrodes per zone (\S\ref{sec:arch-params}), the $100$ zones of this code block draw $60$--$90\,\mathrm{W}$ per logical qubit at $30\,\mathrm{mW}$ per DAC, $30$--$45\times$ the WISE design and far beyond the $3.25\,\mathrm{W}$ ceiling. Compilation substantially improves WISE but does not remove the latency imposed by its switched and multiplexed control.
Our findings are consistent with~\cite{Jones_2026}, which also reports slower logical clock speeds for WISE. Although SAT compilation could also improve QCCD baselines, it cannot remove the control-wiring bottlenecks that motivate WISE.

\subsection{Isolating WISE's Slowdown}
\label{sec:result_isolating_wise_slowdown}
\textbf{Method:} We determine how much of the ${\approx}3\times$ cycle-time gap to locally controlled QCCD is due to switched control versus multiplexed control, by decomposing the compiled WISE cycle of the same distance-7 surface code, and of the $[[72,12,6]]$ BB code, into transport, recooling, electrode charging, gate pulses and measurement/reset (Figure~\ref{fig:cycle-breakdown}).
\textbf{Result:} We show that WISE's switched control that leads to global odd-even ion swapping is the main contributor to cycle time. For the distance-7 surface code at the Current tier and $M=16$, transport ($63.8\%$) and recooling ($29.9\%$) make up $93.7\%$ of the cycle, against $2.2\%$ for multiplexed electrode charging; for the BB code the shares are $63.7\%$, $31.7\%$ and $2.4\%$. Removing charging entirely would therefore speed up these schedules by at most $2.5\%$, whereas removing transport and recooling would allow up to $16$--$22\times$. Neither suffices alone: removing only transport gives at most $2.8\times$ and removing only recooling at most $1.5\times$. At the Optimistic tier, faster transport shifts weight to recooling ($43$--$45\%$ of the cycle against $52\%$ for transport), so transport alone yields at most $2.1\times$ while removing both yields up to $21$--$33\times$; a QEC-first redesign must therefore address routing and recooling together. This shows up in the sensitivity analysis also, with ion swap time causing the largest slowdowns at $110\,\mathrm{ms}$ per syndrome-extraction round between pessimistic and optimistic tiers, followed by cooling time at $62\,\mathrm{ms}$ and static-electrode charging time due to multiplexing at $10\,\mathrm{ms}$.

\subsection{Limits of Multiplexed Control}
\label{sec:result5_multiplexing_limits}

\textbf{Method: }We re-time the same SAT-compiled schedules with all transport \emph{and} all recooling set to zero, leaving only the gate pulses, the per-layer sample-and-hold charging latency $t_{\mathrm{sc}}=M\,t_{\mathrm{ec}}$, measurement and reset
(\S\ref{sec:gate-time}). For the same gate decomposition, layer structure and non-overlapped charging policy, this bounds from below \emph{any routing fabric and any cooling policy} driven by the same multiplexed control, so it
isolates what multiplexing alone costs. Schemes that charge electrodes during other operations, as discussed for WISE~\cite{Malinowski2023}, or that use fewer gate layers, fall outside this bound.
These re-timings are performed at fixed $k=2,\bar{n}_{\mathrm{th}}=0.0$ and
for the $[[72,12,6]]$ BB code (144 ions) and the distance-7 surface code (97 qubits) only. We sweep $M\in\{1,16,128\}$ and test the Current technology tier.

\textbf{Result:} We find that raising the multiplexing order from $1$ to $128$ lengthens the
transport/cooling-free syndrome-extraction round from $1.9\to18.7$\,ms for the BB code and $1.3\to6.7$\,ms for the surface code $d=7$ at
the Current tier, in exchange for a $57$--$110\times$ reduction in power per
logical qubit. For scale, an $18.7\,$ms BB round is over $10,000\times$ the $1\,\mu$s surface-code round assumed by resource estimates that define early fault tolerance~\cite{gidney2025factor2048bitrsa,gidney2025million,Litinski2019} (demonstrated at $1.1\,\mu$s on superconducting hardware~\cite{Acharya2024}), ${\sim}190\times$ the $100\,\mu$s operation time of the Azure Quantum Resource Estimator's trapped-ion model~\cite{beverland2022assessingrequirementsscalepractical,vandam2024usingazurequantumresource}, and $19\times$ the $1$\,ms syndrome-extraction round assumed for recovering a P-256 private key within 264 days~\cite{cain2026shorsalgorithmpossible10000}.

\textbf{Our recommended multiplexing factor of $M=16$ leaves a transport/cooling-free round of $3.9\,$ms for the BB code ($2.0\,$ms for the surface code $d=7$) while reducing power $P_\mathrm{LQ}$ per logical qubit from $3.64\to0.28\,$W. Relative to $M=128$, this choice gives back $14.8\,$ms per round for only $0.21\,\mathrm{W}$ per logical qubit, demonstrating the benefit of architectural co-design.} However, even with co-design, if the balanced P-256 architecture of~\cite{cain2026shorsalgorithmpossible10000} took the same number of rounds on both platforms, recovering a private key within 2 days would require a $\tfrac{264\, \mathrm{days}}{2 \, \mathrm{days}}\times \tfrac{3.92\, \mathrm{ms}}{1.0 \, \mathrm{ms}} \approx 520\times$ reduction in logical circuit depth on BB rounds. Table~\ref{tab:application-rescaling} repeats this rescaling for the surface code and for three levels of the WISE schedule: gates, measurement and charging only (transport and recooling removed), with recooling added, and the full WISE round. \textbf{Even with transport and recooling removed, current hardware needs a $260$--$520\times$ reduction in logical depth; with recooling it needs $1{,}500$--$4{,}100\times$, and with the full WISE round $4{,}100$--$11{,}400\times$.} Optimistic hardware lowers these factors by $4$--$9\times$, to $31$--$60\times$ without transport and recooling and $500$--$2{,}000\times$ for full rounds. Recooling alone multiplies the required reduction by $6$--$16\times$, so cooling, not only routing, separates WISE from the $1$\,ms rounds such application estimates assume. 

\begin{table*}[t]
\centering
\footnotesize
\caption{Mean syndrome-round time at $(k,M,\bar n_{\mathrm{th}})=(2,16,0)$ and the logical-depth reduction needed to run the $264$-day P-256 architecture of~\cite{cain2026shorsalgorithmpossible10000} (1\,ms rounds) within 48\,h on these rounds, for the same number of logical rounds. ``gates'': gate pulses, charging, measurement and reset only; ``+cool'': adds recooling; ``full'': the complete WISE round including transport.}
\label{tab:application-rescaling}
\begin{tabular}{@{}llrrrrrr@{}}
\toprule
 & & \multicolumn{3}{c}{Round time (ms)} & \multicolumn{3}{c}{Required depth reduction} \\
\cmidrule(lr){3-5}\cmidrule(l){6-8}
Tier & Code & gates & +cool & full & gates & +cool & full \\
\midrule
Current & Surface $d{=}7$ & 1.98 & 11.29 & 31.19 & 261$\times$ & 1{,}491$\times$ & 4{,}117$\times$ \\
Current & BB $[[72,12,6]]$ & 3.92 & 31.32 & 86.32 & 518$\times$ & 4{,}135$\times$ & 11{,}394$\times$ \\
Optimistic & Surface $d{=}7$ & 0.26 & 2.59 & 5.43 & 35$\times$ & 342$\times$ & 717$\times$ \\
Optimistic & BB $[[72,12,6]]$ & 0.46 & 7.31 & 15.17 & 60$\times$ & 965$\times$ & 2{,}002$\times$ \\
Optimistic & Surface $d{=}5$ & 0.24 & 2.35 & 3.78 & 31$\times$ & 310$\times$ & 499$\times$ \\
\bottomrule
\end{tabular}
\end{table*}

\section{Conclusion}

Substantial progress has been made in trapped-ion hardware~\cite{Srinivas_2021,oxfordionics_Loschnauer2024,Mordini_2025, Corsetti_2026, Knollmann_2024, valentini2024demonstrationtwodimensionalconnectivityscalable, IonQTempo2025, helios2025, hughes2025trappediontwoqubitgates9999, sotirova2024highfidelityheraldedquantumstate}, compiler infrastructures~\cite{muzzletheshuttle, murali2020architectingnoisyintermediatescaletrapped, MOVELESS_khan2025movelessminimizingoverheadqccds,russon2026scalingqubitmappingrouting, SSYNC_QCCD, Schieretal2026, AdaptiveRouting2026Ovide}, controller architectures~\cite{Malinowski2023, Stuart2019, Tothetal2026,ohira2026cryogenictimedivisionmultiplexedvoltagecontrol}, and QEC~\cite{BaconShor2006, Bombin2006,Fowler2012, Fowler2018LowOQ, TillichZemor2014, Litinski2019, bravyi2024highthreshold, Panteleev2022, cain2026shorsalgorithmpossible10000}, but these areas have largely evolved independently. We join these lines of work through WISER, a cross-layer resource estimation framework for co-design of scalable trapped-ion systems.

Our results read at two levels. Level~(i), specific to WISE's \emph{switched control} fabric: reduced routing time through SAT-based compilation (\S\ref{sec:result1}) and co-design (\S\ref{sec:result2}) improves WISE's logical clock speed and error rate, yet on current hardware only distance-7 surface codes meet the early FTQC screen, at under $5$\,Hz. \S\ref{sec:result3_dominant} provides design targets across low-level physical parameters, \S\ref{sec:result3_qccd_comparison} quantifies the remaining gap to locally controlled QCCD, and \S\ref{sec:result_isolating_wise_slowdown} attributes it to WISE's switched control, through restricted globally broadcast transport and the recooling it requires. Level~(ii), common to \emph{all multiplexed control} with non-overlapped charging (covering scalable trapped-ion systems beyond WISE): without routing and cooling, sample-and-hold multiplexing alone lengthens a syndrome-extraction round from $1.9\,$ms to $18.7\,$ms for the BB code, and from $1.3\,$ms to $6.7\,$ms for the surface code $d=7$, when going from $M=1$ to $M=128$. We then show that architecture co-design recovers most of the cost of multiplexing (\S\ref{sec:result5_multiplexing_limits}). 

\paragraph{Scope and Limitations.} The estimates are comparative under one model, not exact hardware predictions: the noise model projects device-physics mechanisms onto Pauli channels, samples faults independently even where they share a physical source such as a DAC, and adds single-qubit, measurement and reset infidelities as constant penalties, and no device currently achieves our parameter selection simultaneously at scale. Likewise, the $325\,$W cold-stage budget for electronic control (all DACs cold, one dynamic-waveform bank per processor) and the $48$-hour uninterrupted execution window~(\S\ref{sec:early-ft-screen}) are projections and have not been validated experimentally.\label{sec:discussion-power}

The relative standing of the code families may change once logical gates along with $T$-state preparation and consumption~\cite{beverland2022assessingrequirementsscalepractical,Huggins2025FLASQ}, together with a scalable real-time decoder are included. We exclude classical decoder runtime from our screen: WISE's millisecond-scale rounds leave far more time than representative decoding latencies~\cite{Higgott_2025}, unlike superconducting architectures whose ${\sim}1\,\mu$s cycle risks a syndrome backlog~\cite{Skoric_2023,caune2024demonstratingrealtimelowlatencyquantum,Acharya2024}. Likewise, logical memory isolates encoded-qubit performance without gate-specific assumptions, providing a common baseline for comparing code families~\cite{beverland2022assessingrequirementsscalepractical}.

\paragraph{Future Work.} \emph{(i) A QEC-first fabric.} Swap routing and recooling dominate WISE's cycle time (\S\ref{sec:result_isolating_wise_slowdown}), so extending the SAT formulation of \S\ref{sec:sat-routing} to QEC-first fabrics, co-designed with a wiring scheme that reaches feasible power budgets without sacrificing logical clock speed, is the most direct follow-up; any globally controlled design with non-overlapped charging must budget for the level-(ii) latency, which at the multiplexing orders needed for a meaningful wiring reduction leaves a millisecond-scale syndrome-extraction round on its own. Gates that tolerate motional excitation~\cite{hughes2025trappediontwoqubitgates9999} could also relax the recooling that WISE schedules currently require.

\emph{(ii) Memory experiments on WISE-class hardware.} Experimental characterisation of multiplexed trapped-ion systems will further inform our noise model and subsequent resource estimates. Memory experiments are particularly appropriate because wiring and power constraints limit monolithic QCCD systems to approximately $10^4$ physical qubits~\cite{Malinowski2023,Jones_2026}, restricting near-term devices to a few logical qubits at moderate code distance, a regime in which logical-memory performance dominates system capability. Quantum memory is also directly relevant to practical near-term applications. For instance, quantum networking and key distribution using repeaters~\cite{Azuma2023QuantumRepeaters} requires nodes to maintain quantum states over long timescales.

\emph{(iii) An end-to-end resource model.} Logical computation adds further overheads~\cite{Horsman2012,Litinski2019,Fowler2018LowOQ} that need accounting for. Our estimates suggest that executing these workloads with trapped ions will likely require modular architectures~\cite{Stephenson2020IonPhoton,Mohseni2024HowTB}, such as
MUSIQC~\cite{MonroeModular2014} QCCD modules joined by photonic interconnects,
rather than a single monolithic device; since each node must independently maintain
high-quality logical memory, the intra-module optimisations studied here apply
directly, and we expect our main architectural conclusions to carry over,
namely the dominance of execution time as the bottleneck, the advantage of small
trap capacity and the viability of high-rate qLDPC codes, though absolute cycle
times and error rates will grow once computation and inter-module links are included. Extending WISER to logical architecture proposals~\cite{yoder2025tourgrossmodularquantum,IonQWalkingCat,bhardwaj2026highrateqldpcprocessors} mapped to distributed trapped ions is an important direction for future work.

\appendix
\newpage
\begin{table*}[t]
    \centering
    \footnotesize
        \caption{Clause families for the SAT formulation of QMR.
        }
    \label{tab:sat-clauses-structural}
    \setlength{\tabcolsep}{7pt}
    \begin{tabular}{p{0.47\textwidth} p{0.47\textwidth}}
        \hline

        \multicolumn{2}{c}{{\textbf{\emph{Structural validity}}}} \\
        \hline\hline

        \begin{minipage}[t]{0.47\textwidth}
        \vspace{0pt}
        \raggedright

        {\textbf{Permutation (cell): \hfill $O(R\, P_{\mathrm{max}}\, N^2)$}} \\
        For every encoded slice $(r,p,d,c)$, enforce
            $\textsc{ExactlyOne}(\{a(r,p,d,c,i): i \in I\})$.
        \textit{Every cell contains exactly one active ion at each encoded time slice; in spectator-pruned subproblems, cells may be empty but never multiply occupied, so this relaxes to $\textsc{AtMostOne}$.}

\smallskip\hrule\smallskip

        {\textbf{Permutation (ion): \hfill $O(R\, P_{\mathrm{max}}\, N^2)$}} \\
        For every encoded slice $(r,p,i)$, enforce
            $\textsc{ExactlyOne}(\{a(r,p,d,c,i): (d,c)\in\mathrm{Grid}\}).$
        \textit{Each ion occupies exactly one cell at every encoded time step.}

\smallskip\hrule\smallskip

        {\textbf{Phase monotonicity: \hfill $O(R\,P_\mathrm{max})$}} \\
        For all $r$ and $p<P_{\mathrm{max}}-1$, enforce
            $\phi(r,p)\Rightarrow\phi(r,p{+}1)$.
        \textit{Once a round switches to vertical phase, it cannot return to horizontal phase (an $H^{*}V^{*}$ phase order; see Proposition~\ref{prop:hv-colocation}).}

\smallskip\hrule\smallskip
        {\textbf{H/V gating: \hfill $O(R\,P_{\mathrm{max}}\,N)$}} \\
        For every horizontal comparator $(d,c)$ and every vertical comparator $(d,c)$, enforce $\phi(r,p)\Rightarrow \neg s_h(r,p,d,c)$ and $\neg\phi(r,p)\Rightarrow \neg s_v(r,p,d,c)$.

        \textit{Horizontal swaps may fire only in horizontal phase, and vertical swaps may fire only in vertical phase.}

        \smallskip\hrule\smallskip

        {\textbf{H/V parity gating: \hfill $O(R\,P_{\mathrm{max}}\,N)$}} \\
        For every horizontal comparator $(d,c)$ and vertical comparator $(d,c)$ whose parity does not match the pass parity, enforce
            $c\bmod 2 \neq p\bmod 2 \;\Rightarrow\; \neg s_h(r,p,d,c)$,
            $d\bmod 2 \neq p\bmod 2 \;\Rightarrow\; \neg s_v(r,p,d,c)$.
        \textit{Only the odd or even comparator subset selected by the global odd-even clock may fire in a pass.}

              \smallskip\hrule\smallskip

        {\textbf{MS-pair end in same trap: \hfill $O(R\,|\mathcal{P}|\,n + R\,|\mathcal{P}|\,B)$}} \\
        For each required pair $(i_1,i_2)\in\mathcal{P}_r$, every row $d$, and every block $b$, enforce
            $\mathrm{row\_end}(r,i_1,d)\Leftrightarrow \mathrm{row\_end}(r,i_2,d)$ and
            $w_{\mathrm{end}}(r,i_1,b)\Leftrightarrow w_{\mathrm{end}}(r,i_2,b)$.
        \textit{Paired ions must finish in the same row and the same local block, hence in the same trap.}
        \smallskip\hrule\smallskip

        {\textbf{Inter-round chain: \hfill $O(R\,N^2)$}} \\
        For every round boundary and every $(d,c,i)$, enforce
            $a(r,max,d,c,i) \Leftrightarrow a(r{+}1,0,d,c,i)$.
        \textit{The final layout of one round is the initial layout of the next.}

         \end{minipage}
        &
        \begin{minipage}[t]{0.47\textwidth}
        \vspace{0pt}
        \raggedright

           {\textbf{Initial layout: \hfill $O(N^2)$}} \\
        If input cell $(d,c)$ initially contains ion $i_0$, assert
                $a(0,0,d,c,i_0)$,
            and
            $\neg a(0,0,d,c,i)\;\; \forall i \neq i_0$.
        \textit{The first SAT slice is pinned exactly to the incoming layout.}

        \smallskip\hrule\smallskip

        {\textbf{H/V swap semantics: \hfill $O(R\,P_\mathrm{max}\,N^2)$}} \\
        For every horizontal or vertical comparator and ion $i$, enforce
        \[
            s_h(r,p,d,c)\wedge a(r,p,d,c{+}1,i)\Rightarrow a(r,p{+}1,d,c,i),
        \]
        \[
            s_h(r,p,d,c)\wedge a(r,p,d,c,i)\Rightarrow a(r,p{+}1,d,c{+}1,i),
        \]
        \[
            s_v(r,p,d,c)\wedge a(r,p,d{+}1,c,i)\Rightarrow a(r,p{+}1,d,c,i),
        \]
        \[
            s_v(r,p,d,c)\wedge a(r,p,d,c,i)\Rightarrow a(r,p{+}1,d{+}1,c,i).
        \]
        \textit{If a comparator fires, the two incident ions are forced to swap positions in the next slice.}

  \smallskip\hrule\smallskip

        {\textbf{Copy clauses: \hfill $O(R\,P_{max}\,N^2)$}} \\
        For each cell $(d,c)$ and ion $i$, let the incident comparators be $s_{\mathrm{left}}$ and $s_{\mathrm{right}}$ (or $s_{\mathrm{up}}$ and $s_{\mathrm{down}}$).
        \[
            (\neg\phi(r,p)\wedge \neg s_{\mathrm{left}}\wedge \neg s_{\mathrm{right}}\wedge a(r,p,d,c,i))
            \Rightarrow a(r,p{+}1,d,c,i),
        \]
        \[
            (\phi(r,p)\wedge \neg s_{\mathrm{up}}\wedge \neg s_{\mathrm{down}}\wedge a(r,p,d,c,i))
            \Rightarrow a(r,p{+}1,d,c,i),
        \]
        together with the converse implications from $a(r,p{+}1,d,c,i)$ back to $a(r,p,d,c,i)$ under the same phase and no-swap guards.
        \textit{An ion stays in place whenever no incident comparator of the active phase fires.}

          \smallskip\hrule\smallskip

        {\textbf{Row-end / Block-end helpers: \hfill $O(R\,N\,(n+B))$}} \\
        For every round $r$, ion $i$, row $d$, and local block $b$, enforce $\mathrm{row\_end}(r,i,d)\Leftrightarrow \bigvee_{c=0}^{m-1} a(r,max,d,c,i)$ and
            $w_{\mathrm{end}}(r,i,b)\Leftrightarrow \bigvee_{(d,c)\in\mathrm{block}(b)} a(r,max,d,c,i)$,
        for valid blocks; degenerate edge blocks that are not fully supported are forced false.
        \textit{These helper literals summarize the row and local trap block in which each ion finishes the round.}

          \smallskip\hrule\smallskip

     {\textbf{Return to start:  \hfill $O(N^2)$} \\
        For every ion $i$, row $d$, col $c$, enforce
            $a(1,0,d,c,i)\Leftrightarrow a(R,p_{max},d,c,i)$
        }

         \end{minipage}

        \\

        \\

        \hline\hline

        \begin{minipage}[t]{0.47\textwidth}
        \vspace{0pt}
        \raggedright

        {\textbf{Helpers: \hfill $O(R\,P_{\mathrm{max}}\,(N+1))$}} \\
        For every $(r,p)$, enforce
        \[
            u(r,p)\Leftrightarrow
            \bigvee_{d,c}\bigl(s_h(r,p,d,c)\vee s_v(r,p,d,c)\bigr),
        \]
        \[
            v_{\mathrm{active}}(r,p)\Leftrightarrow u(r,p)\wedge \phi(r,p).
        \]
        \textit{$u(r,p)$ marks active passes $p$, while $v_{\mathrm{active}}(r,p)$ marks passes that are both active and vertical.}
        \\

        \end{minipage}

        &    \begin{minipage}[t]{0.47\textwidth}
        \vspace{0pt}
        \raggedright
        {\textbf{Weighted bound: \hfill totaliser-sized}} \\
        If a weighted bound $B$ is requested, enforce
        \[
            \sum_{r=2}^{R}\sum_{p=0}^{max-1}
            \Bigl(u(r,p) + (w{-}1)\,v_{\mathrm{active}}(r,p)\Bigr)\leq B.
        \]

        \textit{Horizontal active passes cost 1, vertical active passes cost $w$, and binary search tightens $B$ to the minimum feasible calibrated cost. With cycle reuse, the passes of each cycle round are additionally weighted by the repetition count $C$ (Appendix~\ref{si:cycle-reuse}).}
        \\
        \end{minipage}
        \\

    \end{tabular}

\end{table*}

\section{Complete SAT-formulation}
\subsection{SAT Optimality Proof}
\label{append-optimality}

We state precisely what the SAT formulation guarantees. Its model is the odd-even sorting network of Table~\ref{tab:sat-clauses-structural}: each round orders its passes horizontal-then-vertical, uses at most $p_{\max}=2(n+m)$ passes, minimises the integer-weighted cost $\sum_r W_r$, and, under cycle reuse, returns the layout to the cycle start.

\paragraph{Tight weighted bounds via parallel search.}
The feasibility predicate in $B$ is monotone: if a schedule exists at weighted cost $B_0$, it also satisfies the formula at any bound $B \geq B_0$. We exploit this monotonicity through parallel search: exploring different $B$ values simultaneously. \textit{Hunters} find feasible QMR solutions by probing candidate upper bounds for $B$ from above, finding successively tighter SAT-feasible solutions, while \textit{Provers} systematically prove UNSAT from below at the maximum pass budget.

\paragraph{Convergence and optimality certificate.}
The search \emph{converges} when the lower bound meets the upper bound:
\begin{equation}\label{eq:convergence}
B_{\mathrm{worst\_unsat}} \;\geq\; B_{\mathrm{best\_sat}} - 1.
\end{equation}
At convergence, the solver has a SAT-certified optimal schedule at cost~$B^*$ together with an UNSAT proof that $B^* - 1$ is infeasible at $p_{\max} = 2(n{+}m)$. This constitutes an optimality certificate for the encoded model: no schedule satisfying its constraints can achieve a lower weighted cost. If the time budget expires before convergence, the solver returns the best incumbent solution with its proven lower bound, allowing the caller to assess the remaining optimality gap.

\paragraph{The phase order never makes a round infeasible.}
A general permutation of a grid needs a horizontal-vertical-horizontal sort, and some permutations cannot be realised by one horizontal phase followed by one vertical phase. An MS round, however, only requires each pair of ions to share \emph{some} trap, and for this the restricted order suffices.

\begin{proposition}\label{prop:hv-colocation}
Let a fully occupied grid have rows tiled by traps of even capacity $k$, and let the ions be partitioned into pairs (unpaired ions may be paired arbitrarily). Then one horizontal phase that permutes ions within rows, followed by one vertical phase that permutes ions within columns, places both ions of every pair in the same trap.
\end{proposition}
\begin{proof}
Let each row have $n$ zones. Form a multigraph on the rows with one edge $(r_a,r_b)$ per pair whose ions lie in rows $r_a$ and $r_b$ (a loop if $r_a=r_b$). Each row holds $n$ ions, so the multigraph is $n$-regular with $n$ even. Since every degree is even, orienting each connected component along an Euler circuit gives each row out-degree and in-degree $n/2$. The resulting bipartite multigraph from ``out'' rows to ``in'' rows is $n/2$-regular, so by K\H{o}nig's theorem its edges split into $n/2$ perfect matchings $\mu_1,\dots,\mu_{n/2}$. Assign matching $\mu_g$ to the column pair $(2g-1,2g)$, which lies inside one trap because $k$ is even. In the horizontal phase, each row places the tail of its unique out-edge in $\mu_g$ at column $2g-1$ and the head of its unique in-edge at column $2g$; this fills every row exactly once. In the vertical phase, column $2g-1$ is left unchanged and column $2g$ moves the head of each edge $(r_a\to r_b)\in\mu_g$ from row $r_b$ to row $r_a$, a permutation because $\mu_g$ is a perfect matching. Afterwards both ions of every pair sit in row $r_a$, columns $2g-1$ and $2g$, which belong to the same trap.
\end{proof}

Odd-even transposition sorts realise these two phases in at most $n$ horizontal and $m$ vertical passes, within the horizon $p_{\max}=2(n+m)$, so every round of the encoding is feasible for the even capacities $k\in\{2,4,6,8,16\}$ we study; spectator-pruned subproblems add dummy occupants. The phase order can still exclude cheaper schedules that return to horizontal passes, and the cycle-return constraint and integer weights $w=\operatorname{round}(t_V/t_H)$ further restrict the model, so our certificates bound the routing cost within the encoded model rather than over all physically admissible WISE schedules.

\subsection{Transport primitive costs}
\label{si:transport-costs}
The costs $t_H$ and $t_V$ are obtained using timings for primitive hardware operations. Since ion swapping is a fundamental WISE movement primitive, we denote its time by $t_0$. A horizontal swap consists of five sequential primitives (shuttle--merge--rotate--split--shuttle~\cite{walther2012controlling,bowler2012coherent,Burton_2023}), giving $t_H=2.12\,t_0$ per odd or even sub-pass and  $2t_H=4.24\,t_0$  per odd--even pair. A  vertical swap moves ions through junctions, organised into $k$ serialised column buckets. Each bucket costs $5.10\,t_0$, with inter-bucket horizontal repositioning giving an amortised sub-pass cost $t_V=k(5.10\,t_0)+(k/2)t_H$. Thus $t_V/t_H\approx2.4k+k/2$: simply minimising raw swap count would ignore WISE's substantial vertical-routing penalty.
\subsection{Cycle-reuse decomposition}
\label{si:cycle-reuse}
Encoding all gate rounds in a single SAT formula is intractable at practical code distances. We exploit the fact that QEC memory experiments are dominated by a short cycle of MS rounds that repeats $C$ times (typically $C=d$), with a small number of initial and final rounds. Rather than encoding all $C$ repetitions, we decompose into two SAT calls. The first call encodes one copy of the repeating cycle plus any non-repeating start rounds, and adds a route-back round that forces the end-of-cycle ion layout to match the cycle start, enabling the solution to be replayed $C$ times without re-solving. Cycle rounds carry weight $C$ in the objective so the solver preferentially optimises the repeated portion. The second call routes any non-repeating end rounds. This removes the dependence of the number of encoded placement slices on the number of cycle repetitions; only the weighted cardinality constraint still grows with $C$. Because the layout must return to the cycle start, the resulting schedules are periodic, and optimality certificates are relative to this restriction.

\section{Implementation using Parallel Search}
\label{appendix:hunters}

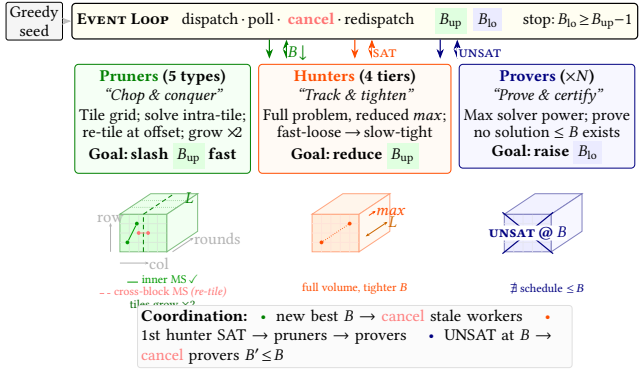
\begin{figure}[t]
\centering
\resizebox{\columnwidth}{!}{
\begin{tikzpicture}[
  >=Stealth,
  arr/.style={->, thick},
  darr/.style={<->, thick},
]

\def\colA{1.0}  \def\colB{5.5}  \def\colC{10.0}

\node[draw, rounded corners=2.5pt, fill=yellow!6, line width=0.8pt,
      inner sep=3.5pt, align=center, minimum width=92mm,
      font=\large] (loop) at (5.5, 10.2) {
  \textbf{\textsc{Event Loop}}\;\;
  {\large dispatch\,$\cdot$\,poll\,$\cdot$\,
   \textcolor{red!55}{\textbf{cancel}}\,$\cdot$\,redispatch}
  \quad
  {\large\colorbox{green!12}{$B_{\mathrm{up}}$}\,
   \colorbox{blue!8}{$B_{\mathrm{lo}}$}}
  \quad{\large stop:\,$B_{\mathrm{lo}} \!\geq\! B_{\mathrm{up}}{-}1$}};

\node[draw, rounded corners=1.5pt, fill=gray!8, inner sep=2.5pt,
      font=\large, align=center, anchor=east]
  (heur) at ([xshift=-2mm]loop.west) {Greedy\\[-0.5pt]seed};
\draw[arr, gray!50] (heur) -- (loop);

\draw[arr, green!50!black]
  ([xshift=-20mm]loop.south) -- ++(0,-0.45);
\draw[arr, orange!60!red]
  (loop.south) -- ++(0,-0.45);
\draw[arr, blue!50!black]
  ([xshift=20mm]loop.south) -- ++(0,-0.45);
\draw[darr, green!50!black]
  ([xshift=-16mm]loop.south) ++(0,-0.45) -- ++(0,0.45)
  node[pos=0.3, right=0pt, font=\large, text=green!50!black,
       fill=white, inner sep=0.3pt] {$B\!\downarrow$};
\draw[darr, orange!60!red]
  ([xshift=4mm]loop.south) ++(0,-0.45) -- ++(0,0.45)
  node[pos=0.3, right=0pt, font=\large, text=orange!60!red,
       fill=white, inner sep=0.3pt] {\textsc{sat}};
\draw[darr, blue!50!black]
  ([xshift=24mm]loop.south) ++(0,-0.45) -- ++(0,0.45)
  node[pos=0.3, right=0pt, font=\large, text=blue!50!black,
       fill=white, inner sep=0.3pt] {\textsc{unsat}};

\node[draw=green!50!black, rounded corners=2.5pt, fill=green!4,
      inner sep=3pt, align=center, line width=0.6pt,
      minimum width=34mm, anchor=north]
  (prn) at (\colA, 9.2) {
  {\large\bfseries\textcolor{green!50!black}{
    Pruners} {\large(5 types)}}\\[1pt]
  {\large\itshape ``Chop \& conquer''}\\[1pt]
  {\large Tile grid; solve intra-tile;}\\[-0.5pt]
  {\large re-tile at offset; grow $\times\!2$}\\[1pt]
  {\large\bfseries Goal:\,slash
    \colorbox{green!12}{$B_{\mathrm{up}}$} fast}};

\node[draw=orange!60!red, rounded corners=2.5pt, fill=orange!4,
      inner sep=3pt, align=center, line width=0.6pt,
      minimum width=34mm, anchor=north]
  (hnt) at (\colB, 9.2) {
  {\large\bfseries\textcolor{orange!60!red}{
    Hunters} {\large(4 tiers)}}\\[1pt]
  {\large\itshape ``Track \& tighten''}\\[1pt]
  {\large Full problem, reduced $max$;}\\[-0.5pt]
  {\large fast-loose\,$\to$\,slow-tight}\\[1pt]
  {\large\bfseries Goal:\,reduce
    \colorbox{green!12}{$B_{\mathrm{up}}$}}};

\node[draw=blue!50!black, rounded corners=2.5pt, fill=blue!3,
      inner sep=3pt, align=center, line width=0.6pt,
      minimum width=34mm, anchor=north]
  (prv) at (\colC, 9.2) {
  {\large\bfseries\textcolor{blue!50!black}{
    Provers} {\large($\times N$)}}\\[1pt]
  {\large\itshape ``Prove \& certify''}\\[1pt]
  {\large Max solver power; prove}\\[-0.5pt]
  {\large no solution $\leq B$ exists}\\[1pt]
  {\large\bfseries Goal:\,raise
    \colorbox{blue!8}{$B_{\mathrm{lo}}$}}};

\begin{scope}[shift={(0.0,4.8)},
  x={(0.22cm,0cm)}, y={(0cm,0.22cm)}, z={(0.12cm,0.08cm)}]
  \fill[green!14, draw=green!50!black!50, line width=0.3pt]
    (0,0,0) -- (5,0,0) -- (5,4,0) -- (0,4,0) -- cycle;
  \foreach \i in {1,...,4}{
    \draw[green!50!black!15, line width=0.15pt] (\i,0,0)--(\i,4,0);}
  \foreach \j in {1,...,3}{
    \draw[green!50!black!15, line width=0.15pt] (0,\j,0)--(5,\j,0);}
  \fill[green!8, draw=green!50!black!50, line width=0.3pt]
    (0,4,0)--(5,4,0)--(5,4,6)--(0,4,6)--cycle;
  \fill[green!5, draw=green!50!black!50, line width=0.3pt]
    (5,0,0)--(5,4,0)--(5,4,6)--(5,0,6)--cycle;
  \draw[green!50!black, densely dashed, line width=0.6pt]
    (2.5,0,0)--(2.5,4,0);
  \draw[green!50!black, densely dashed, line width=0.4pt]
    (2.5,4,0)--(2.5,4,6);
  \draw[green!50!black!60, densely dotted, line width=0.3pt]
    (0,4,2)--(5,4,2);
  \draw[green!50!black!60, densely dotted, line width=0.3pt]
    (0,4,4)--(5,4,4);
  \node[font=\large\itshape, text=green!50!black,
        anchor=south west, inner sep=0.5pt]
    at (5.2,4,3) {$L$};
  \node[circle,fill=green!60!black,inner sep=0.9pt]
    (innerA) at (0.8,1,0){};
  \node[circle,fill=green!60!black,inner sep=0.9pt]
    (innerB) at (1.8,3,0){};
  \draw[green!60!black, line width=0.5pt] (innerA)--(innerB);
  \node[circle,fill=red!55,inner sep=0.9pt]
    (crossA) at (2,2,0){};
  \node[circle,fill=red!55,inner sep=0.9pt]
    (crossB) at (3,2,0){};
  \draw[red!45, densely dashed, line width=0.5pt] (crossA)--(crossB);
  \draw[-{Stealth[length=1.5pt]}, gray!60, line width=0.3pt]
    (0,-1.2,0) -- (3,-1.2,0)
    node[right,font=\large,text=gray!60,inner sep=0.5pt]{col};
  \draw[-{Stealth[length=1.5pt]}, gray!60, line width=0.3pt]
    (-1.0,0,0) -- (-1.0,3,0)
    node[above,font=\large,text=gray!60,inner sep=0.5pt]{row};
  \draw[-{Stealth[length=1.5pt]}, gray!60, line width=0.3pt]
    (5.6,0,0) -- (5.6,0,4)
    node[right,font=\large,text=gray!60,inner sep=0.5pt]{rounds};
\end{scope}

\node[font=\footnotesize, text=green!60!black, anchor=north, inner sep=0.5pt]
  at (\colA, 4.3)
  {\raisebox{0.1pt}{\tikz\draw[green!60!black,line width=0.5pt](0,0)--(0.25,0);}\;inner MS\,\checkmark};
\node[font=\footnotesize, text=red!55, anchor=north, inner sep=0.5pt]
  at (\colA, 4.0)
  {\raisebox{1pt}{\tikz\draw[red!45,densely dashed,line width=0.5pt](0,0)--(0.25,0);}\;cross-block MS\,{\itshape(re-tile)}};
\node[font=\footnotesize, text=green!40!black, anchor=north]
  at (\colA, 3.8) {tiles grow $\times 2$};

\begin{scope}[shift={(4.5,4.8)},
  x={(0.22cm,0cm)}, y={(0cm,0.22cm)}, z={(0.12cm,0.08cm)}]
  \fill[orange!14, draw=orange!60!red!50, line width=0.3pt]
    (0,0,0)--(5,0,0)--(5,4,0)--(0,4,0)--cycle;
  \foreach \i in {1,...,4}{
    \draw[orange!60!red!15, line width=0.15pt] (\i,0,0)--(\i,4,0);}
  \foreach \j in {1,...,3}{
    \draw[orange!60!red!15, line width=0.15pt] (0,\j,0)--(5,\j,0);}
  \fill[orange!8, draw=orange!60!red!50, line width=0.3pt]
    (0,4,0)--(5,4,0)--(5,4,6)--(0,4,6)--cycle;
  \fill[orange!5, draw=orange!60!red!50, line width=0.3pt]
    (5,0,0)--(5,4,0)--(5,4,6)--(5,0,6)--cycle;
  \node[circle,fill=orange!60!red,inner sep=0.9pt] at (1,1,0){};
  \node[circle,fill=orange!60!red,inner sep=0.9pt] at (4,3,0){};
  \draw[orange!60!red, densely dotted, line width=0.5pt]
    (1,1,0)--(4,3,0);
\end{scope}

\begin{scope}[shift={(4.5,4.8)},
  x={(0.22cm,0cm)}, y={(0cm,0.22cm)}, z={(0.12cm,0.08cm)}]
  \draw[-{Stealth[length=1.5pt]}, orange!60!red, line width=0.6pt]
    (5.5,3.5,1) -- (5.5,3.5,2.2);
  \node[font=\large, text=orange!60!red, anchor=west, inner sep=0.5pt]
    at (5.5,3.5,2.3) {$max$};
  \draw[{Stealth[length=1.5pt]}-{Stealth[length=1.5pt]},
        orange!80!black, line width=0.5pt]
    (5.5,1.5,0.5) -- (5.5,1.5,4.5);
  \node[font=\large, text=orange!80!black, anchor=west, inner sep=0.5pt]
    at (5.5,1.5,4.6) {$L$};
\end{scope}
\node[font=\footnotesize, text=orange!60!red, anchor=north]
  at (\colB, 4.2) {full volume, tighter $B$};

\begin{scope}[shift={(9.0,4.8)},
  x={(0.22cm,0cm)}, y={(0cm,0.22cm)}, z={(0.12cm,0.08cm)}]
  \fill[blue!8, draw=blue!50!black!50, line width=0.3pt]
    (0,0,0)--(5,0,0)--(5,4,0)--(0,4,0)--cycle;
  \foreach \i in {1,...,4}{
    \draw[blue!50!black!10, line width=0.15pt] (\i,0,0)--(\i,4,0);}
  \foreach \j in {1,...,3}{
    \draw[blue!50!black!10, line width=0.15pt] (0,\j,0)--(5,\j,0);}
  \fill[blue!5, draw=blue!50!black!50, line width=0.3pt]
    (0,4,0)--(5,4,0)--(5,4,6)--(0,4,6)--cycle;
  \fill[blue!3, draw=blue!50!black!50, line width=0.3pt]
    (5,0,0)--(5,4,0)--(5,4,6)--(5,0,6)--cycle;
  \draw[blue!50!black, line width=0.8pt]
    (-0.2,-0.2,0)--(5.2,4.2,0);
  \draw[blue!50!black, line width=0.8pt]
    (-0.2,4.2,0)--(5.2,-0.2,0);
  \node[font=\large\bfseries, text=blue!50!black, fill=white,
        inner sep=0.5pt, rounded corners=0.5pt]
    at (2.5,2,0) {\textsc{unsat}\,@\,$B$};
\end{scope}

\node[font=\footnotesize, text=blue!50!black, anchor=north]
  at (\colC, 4.2) {$\nexists$ schedule $\leq\!B$};

\node[draw=gray!25, rounded corners=2pt, fill=gray!3,
      inner sep=2.5pt, align=left, font=\large,
      text width=100mm, anchor=north] at (5.5, 3.55) {
  \textbf{Coordination:}\;\,
  \textcolor{green!50!black}{\textbullet}\;
    new best $B$ $\to$
    \textcolor{red!50}{cancel} stale workers
  \quad\textcolor{orange!60!red}{\textbullet}\;
    1st hunter SAT $\to$
    pruners $\to$ provers
  \quad\textcolor{blue!50!black}{\textbullet}\;
    UNSAT at $B$ $\to$
    \textcolor{red!50}{cancel} provers $B'\!\leq\!B$};

\end{tikzpicture}
}
\caption{Parallel SAT portfolio for optimal ion routing.
Each isometric block represents the spacetime volume
(columns $\times$ rows $\times$ rounds).
\textcolor{green!50!black}{\textbf{Pruners}} tile the spatial grid
(dashed boundary on front face);
{\scriptsize\color{green!60!black}inner MS} pairs within a tile
are solved optimally, while
{\scriptsize\color{red!55}cross-block MS} pairs straddling
the boundary receive a boundary preference
and are solved on re-tiling.
Temporal rounds are chunked into windows of~$L$ (top face);
sub-problems double ($\times 2$).
\textcolor{orange!60!red}{\textbf{Hunters}} tackle the full volume
with tighter~$max$.
\textcolor{blue!50!black}{\textbf{Provers}} certify $B_{\mathrm{lo}}$
via \textsc{unsat}.
Workers promote on exhaustion.}
\Description{Diagram of the parallel SAT portfolio: an event loop dispatching pruners that tile the spacetime volume, hunters that solve the full problem, and provers that certify lower bounds.}
\label{fig:sat-portfolio}
\end{figure}

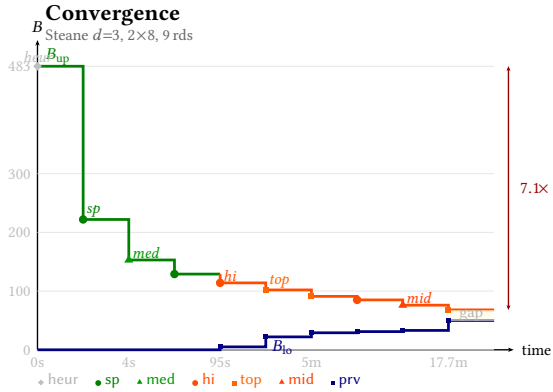
\begin{figure}[t]
\centering
\resizebox{0.88\columnwidth}{!}{
\begin{tikzpicture}[>=Stealth, every node/.style={font=\footnotesize}]

\def\panR{0.8}
\def\panW{7.0}
\def\panH{4.5}
\def\panBot{0.6}

\node[font=\normalsize\bfseries, anchor=south west]
  at (\panR, \panBot+\panH+0.35) {Convergence};
\node[font=\scriptsize, text=black!60, anchor=north west]
  at (\panR, \panBot+\panH+0.55) {Steane $d{=}3$, $2{\times}8$, 9\,rds};

\draw[-{Stealth[length=2.5pt]}, line width=0.6pt]
  (\panR, \panBot) -- (\panR, \panBot+\panH+0.25)
  node[above, font=\scriptsize] {$B$};
\draw[-{Stealth[length=2.5pt]}, line width=0.6pt]
  (\panR, \panBot) -- (\panR+\panW+0.3, \panBot)
  node[right, font=\scriptsize] {time};

\def\By#1{\panBot+\panH*(#1)/500}

\foreach \bv/\bl in {0/0, 100/100, 200/200, 300/300, 483/483}{
  \draw[gray!20, line width=0.15pt]
    (\panR, {\By{\bv}}) -- (\panR+\panW, {\By{\bv}});
  \node[font=\scriptsize, text=gray!50, anchor=east, inner sep=0.5pt]
    at (\panR-0.08, {\By{\bv}}) {\bl};}

\def\Tx#1{\panR+\panW*(#1)/10}
\foreach \ti/\tl in {0/0s, 2/4s, 4/95s, 6/5m, 9/17.7m}{
  \node[font=\scriptsize, text=gray!50, anchor=north, inner sep=0.5pt]
    at ({\Tx{\ti}}, \panBot-0.08) {\tl};}

\draw[green!50!black, line width=1.2pt]
  ({\Tx{0}}, {\By{483}}) -- ({\Tx{1}}, {\By{483}})
  -- ({\Tx{1}}, {\By{222}}) -- ({\Tx{2}}, {\By{222}})
  -- ({\Tx{2}}, {\By{153}}) -- ({\Tx{3}}, {\By{153}})
  -- ({\Tx{3}}, {\By{129}}) -- ({\Tx{4}}, {\By{129}});
\draw[orange!60!red, line width=1.2pt]
  ({\Tx{4}}, {\By{129}})
  -- ({\Tx{4}}, {\By{114}}) -- ({\Tx{5}}, {\By{114}})
  -- ({\Tx{5}}, {\By{102}}) -- ({\Tx{6}}, {\By{102}})
  -- ({\Tx{6}}, {\By{91}})  -- ({\Tx{7}}, {\By{91}})
  -- ({\Tx{7}}, {\By{85}})  -- ({\Tx{8}}, {\By{85}})
  -- ({\Tx{8}}, {\By{76}})  -- ({\Tx{9}}, {\By{76}})
  -- ({\Tx{9}}, {\By{68}})  -- ({\Tx{10}}, {\By{68}});
\node[font=\scriptsize\bfseries, text=green!50!black, anchor=south west,
      inner sep=0.5pt]
  at ({\Tx{0}+0.1}, {\By{483}}) {$B_{\mathrm{up}}$};

\draw[blue!50!black, line width=1.2pt]
  ({\Tx{0}}, {\By{0}}) -- ({\Tx{4}}, {\By{0}})
  -- ({\Tx{4}}, {\By{5}})  -- ({\Tx{5}}, {\By{5}})
  -- ({\Tx{5}}, {\By{22}}) -- ({\Tx{6}}, {\By{22}})
  -- ({\Tx{6}}, {\By{29}}) -- ({\Tx{7}}, {\By{29}})
  -- ({\Tx{7}}, {\By{31}}) -- ({\Tx{8}}, {\By{31}})
  -- ({\Tx{8}}, {\By{33}}) -- ({\Tx{9}}, {\By{33}})
  -- ({\Tx{9}}, {\By{50}}) -- ({\Tx{10}}, {\By{50}});
\node[font=\scriptsize\bfseries, text=blue!50!black, anchor=north west,
      inner sep=0.5pt]
  at ({\Tx{5.1}}, {\By{20}}) {$B_{\mathrm{lo}}$};

\fill[yellow!12, opacity=0.6]
  ({\Tx{9}}, {\By{50}}) rectangle ({\Tx{10}}, {\By{68}});
\node[font=\scriptsize, text=gray!60, inner sep=0pt]
  at ({\Tx{9.5}}, {\By{59}}) {gap};

\node[diamond, fill=gray!50, inner sep=1pt]
  at ({\Tx{0}},{\By{483}}){};
\node[circle, fill=green!50!black, inner sep=1.3pt]
  at ({\Tx{1}},{\By{222}}){};
\node[circle, fill=green!50!black, inner sep=1.3pt]
  at ({\Tx{3}},{\By{129}}){};
\node[regular polygon, regular polygon sides=3,
      fill=green!60!black, inner sep=0.8pt]
  at ({\Tx{2}},{\By{153}}){};
\node[circle, fill=orange!60!red, inner sep=1.3pt]
  at ({\Tx{4}},{\By{114}}){};
\node[circle, fill=orange!60!red, inner sep=1.3pt]
  at ({\Tx{7}},{\By{85}}){};
\node[rectangle, fill=orange!80!red, inner sep=1.3pt]
  at ({\Tx{5}},{\By{102}}){};
\node[rectangle, fill=orange!80!red, inner sep=1.3pt]
  at ({\Tx{6}},{\By{91}}){};
\node[rectangle, fill=orange!80!red, inner sep=1.3pt]
  at ({\Tx{9}},{\By{68}}){};
\node[regular polygon, regular polygon sides=3,
      fill=orange!50!red, inner sep=0.8pt]
  at ({\Tx{8}},{\By{76}}){};
\foreach \step/\bval in {4/5, 5/22, 6/29, 7/31, 8/33, 9/50}{
  \node[rectangle, fill=blue!50!black, inner sep=1pt]
    at ({\Tx{\step}},{\By{\bval}}){};}

\node[font=\scriptsize\itshape, text=gray!50, anchor=south, inner sep=0pt]
  at ({\Tx{0}}, {\By{490}}) {heur};
\node[font=\scriptsize\itshape, text=green!50!black, anchor=south west, inner sep=0pt]
  at ({\Tx{1}+0.04}, {\By{225}}) {sp};
\node[font=\scriptsize\itshape, text=green!60!black, anchor=south west, inner sep=0pt]
  at ({\Tx{2}+0.04}, {\By{156}}) {med};
\node[font=\scriptsize\itshape, text=orange!60!red, anchor=south west, inner sep=0pt]
  at ({\Tx{4}+0.04}, {\By{117}}) {hi};
\node[font=\scriptsize\itshape, text=orange!80!red, anchor=south west, inner sep=0pt]
  at ({\Tx{5}+0.04}, {\By{105}}) {top};
\node[font=\scriptsize\itshape, text=orange!50!red, anchor=south west, inner sep=0pt]
  at ({\Tx{8}+0.04}, {\By{79}}) {mid};

\draw[{Stealth[length=2pt]}-{Stealth[length=2pt]},
      red!60!black, line width=0.6pt]
  ({\Tx{10}+0.22}, {\By{68}}) -- ({\Tx{10}+0.22}, {\By{483}})
  node[midway, right=1.5pt, font=\scriptsize\bfseries, text=red!60!black]
    {$7.1\!\times$};

\node[anchor=north west, inner sep=0pt, align=left,
      font=\scriptsize, text=gray!70] at (\panR, \panBot-0.35) {
  \tikz[baseline=-0.5ex]{\node[diamond,fill=gray!50,inner sep=0.7pt]{};}
    \textcolor{gray!60}{\textsf{heur}}
  \;\;
  \tikz[baseline=-0.5ex]{\node[circle,fill=green!50!black,inner sep=0.9pt]{};}
    \textcolor{green!50!black}{\textsf{sp}}
  \;\;
  \tikz[baseline=-0.5ex]{\node[regular polygon,regular polygon sides=3,
    fill=green!60!black,inner sep=0.5pt]{};}
    \textcolor{green!60!black}{\textsf{med}}
  \;\;
  \tikz[baseline=-0.5ex]{\node[circle,fill=orange!60!red,inner sep=0.9pt]{};}
    \textcolor{orange!60!red}{\textsf{hi}}
  \;\;
  \tikz[baseline=-0.5ex]{\node[rectangle,fill=orange!80!red,inner sep=0.9pt]{};}
    \textcolor{orange!80!red}{\textsf{top}}
  \;\;
  \tikz[baseline=-0.5ex]{\node[regular polygon,regular polygon sides=3,
    fill=orange!50!red,inner sep=0.5pt]{};}
    \textcolor{orange!50!red}{\textsf{mid}}
  \;\;
  \tikz[baseline=-0.5ex]{\node[rectangle,fill=blue!50!black,inner sep=0.8pt]{};}
    \textcolor{blue!50!black}{\textsf{prv}}
};

\end{tikzpicture}
}
\caption{Convergence trace for Steane $d{=}3$
($2{\times}8$ grid, 9 rounds) over the first $17.7$ minutes.
Pruners (green) rapidly slash $B_{\mathrm{up}}$ from 483 to 129;
hunters (orange) tighten the incumbent to $B_{\mathrm{up}} = 68$ ($7.1\times$ improvement);
provers (blue) raise the proven lower bound $B_{\mathrm{lo}}$ to~50,
so at this point the optimum lies in $[51, 68]$.}
\Description{Step plot over solver wall-clock time of the incumbent upper bound falling from 483 to 68 and the proven lower bound rising from 0 to 50.}
\label{fig:sat-convergence}
\end{figure}

\paragraph{Pruners, Hunters, Provers.}
To scale up SAT solving, we introduce a parallel solving framework illustrated in Figure~\ref{fig:sat-portfolio}. A main event loop dispatches workers, polls for results, updates shared bounds ($B_{\mathrm{upper}}$, $B_{\mathrm{lower}}$), cancels obsolete workers, and redispatches into freed slots.
\emph{Pruners} decompose the problem and rapidly slash $B_{\mathrm{upper}}$.
\emph{Hunters} solve the full SAT encoding at progressively tighter targets, activated when paired pruners exhaust.
\emph{Provers} certify $B_{\mathrm{lower}}$ by proving UNSAT from below.
Workers promote on exhaustion (pruner $\to$ hunter $\to$ prover); when the first hunter finds a solution, all remaining pruners convert to provers, accelerating the proof phase.

\paragraph{Convergence.}
Figure~\ref{fig:sat-convergence} shows the search for a nine-round Steane instance ($d{=}3$, $2{\times}8$ grid): within $17.7$ minutes the incumbent falls from $B_{\mathrm{heur}} = 483$ to $68$ ($7.1\times$ improvement) while the proven lower bound reaches $50$. Across the 25 instances of Figure~\ref{fig:SAT-vs-heuristic}, the ratio of the heuristic's weighted cost to WISER's ranges from $1.0$ to $9.6\times$ (median $4.2\times$); six instances, all with at most 17 qubits, converge to a certified optimum, and three more end within $20\%$ of their proven lower bound (the ticks of Figure~\ref{fig:SAT-vs-heuristic}).

\section{WISE Gate Model and Re-write Rules}
\label{appendix:gate-model}
\subsection{Gate model}
We restrict the entangling gate set to pairwise two-qubit MS gates, executed serially within each trap. This is a deliberate design-space choice grounded in current hardware capability. Although global M\o{}lmer--S\o{}rensen (GMS) gates have been demonstrated on $10$--$20$ ions, they achieve fidelities of only $70$--$90\%$ at gate times of $\approx 1000\,\mu$s~\cite{Lu_2019,Grzesiak_2020_EASE}, compared to $>99.9\%$ fidelity for state-of-the-art two-qubit MS gates at $\sim100\,\mu$s.

Likewise, parallel pairwise MS gates within a single chain require longer, more complex pulses and have shown lower fidelity than serial execution~\cite{Figgatt_2019}. The modest parallelism gains from GMS and parallel MS gates are outweighed by their fidelity penalty. Our gate model therefore reflects the operating regime of current TI hardware~\cite{helios2025,racetrack_Moses2023}, where high-fidelity serial two-qubit MS gates remain the dominant entangling primitive; it favours small traps, and collective or parallel entangling gates would be a separate design point.

\subsection{Gate decomposition rules}
The Stim~\cite{stim_Gidney2021} circuits for each code we study here are generated from the \texttt{qldpc} Python library~\cite{qldpc}. Since execution time is a bottleneck in WISE, we use one ancilla per QEC stabiliser unlike prior works which reuse ancillas~\cite{MOVELESS_khan2025movelessminimizingoverheadqccds}.

The compiler first rewrites each circuit into Stim's $\{\texttt{H},\texttt{S},\texttt{CX},\texttt{M},\texttt{R}\}$ basis, for example $\texttt{CZ}=(I\otimes\texttt{H})\,\texttt{CX}\,(I\otimes\texttt{H})$, $\texttt{SWAP}=3\times\texttt{CX}$, $\texttt{MX}=\texttt{H}\,\texttt{M}\,\texttt{H}$ and $\texttt{MRX}=\texttt{H}\,\texttt{M}\,\texttt{R}\,\texttt{H}$, and then lowers these gates into the native WISE operations
$$\mathcal{G}_{\text{native}} = \{\text{MS}, R_X, R_Y, R_Z, M_Z, R\}.$$
With $R_P(\theta)=e^{-i\theta P/2}$ and $\text{MS}(\theta)=e^{-i\theta X\otimes X}$, and listing operations in the order they are applied, the rules, based on~\cite{figgatt2018building}, are
\begin{itemize}[nosep]
    \item $\texttt{H} \to R_Y(\pi/2),\; R_X(\pi)$ \quad (2 rotations)
    \item $\texttt{CX}_{c,t} \to R_Y^{c}(\pi/2),\; \text{MS}(\pi/4),\; R_X^{c}(-\pi/2),$\\ \hspace*{2.2em}$R_X^{t}(-\pi/2),\; R_Y^{c}(-\pi/2)$ \quad (1 MS $+$ 4 rotations)
    \item $\texttt{S} \to R_Z(\pi/2)$
    \item $\texttt{M} \to M_Z$ and $\texttt{R}\to R$,
\end{itemize}
\noindent each exact up to a global phase (checked numerically). The CNOT decomposition follows the standard trapped-ion construction: the MS gate produces the entangling interaction $XX(\pi/4)$, and the surrounding single-qubit rotations convert this to a CNOT via basis changes. The scheduler charges every single-qubit rotation the same duration $t_{1q}$ and error, independent of its angle.

\section{Noise modelling}
\label{append-noise}
 In WISE, groups of controls share DACs and work together on ions. Although this collective control makes WISE resource-efficient, it also introduces correlated noise channels that differ qualitatively from the independent local-control picture assumed in simpler models. We first describe noise modelling for trapped-ion systems in general and then tailor it to the WISE architecture where needed. We begin with correlated noise by introducing the Lamb--Dicke parameter~\cite{leibfried2003quantum} for ion $i$ and mode $m$,
\begin{equation}\label{eq:eta}
  \eta_{i,m} = k_{\mathrm{eff}}\, z_{0,m}\, \nu_m^{(i)}
  \quad\text{with}\quad
  z_{0,m} = \sqrt{\frac{\hbar}{2 m_{\mathrm{ion}} \omega_m}},
\end{equation}
where $k_{\mathrm{eff}}$ is the projection of the driving field's effective wavevector onto the mode direction ($2\pi/\lambda$ for the optical MS drive; far smaller for the microwave-frequency drive of single-qubit gates), $m_{\mathrm{ion}}$ is the ion mass, $\omega_m=2\pi f_m$ the mode's angular frequency and $\nu_m^{(i)}$ its normalised eigenvector component on ion $i$.
Crosstalk between ions is modelled directly through the coupling of these terms across the modes of a trap. Single-qubit operations, in contrast, already reach very low infidelity (down to $10^{-7}$~\cite{smith2025single}), so we do not model their underlying physics explicitly and instead add a constant infidelity penalty. The following paragraphs describe each noise source; its parameters are listed in Table~\ref{tab:tech-tiers}, and Figure~\ref{fig:q3_sensitivity} shows how each parameter affects the expected physical error per syndrome-extraction round.

\noindent\textbf{Motional heating.} This type of noise originates from electric field fluctuations near surfaces. For a homogeneously distributed layer of fluctuating dipoles on a plane with surface density $\sigma_d$, the electric field fluctuations along a direction parallel to the trap surface have autocorrelation~\cite{talukdar2016implications,turchette2000heating}
\begin{equation}
\langle E(t)E(0)\rangle = \frac{3\pi\sigma_d}{8(4\pi\epsilon_0)^2d^4}\langle \mu(t)\mu(0)\rangle,
\end{equation}
where $d$ is the ion-surface distance and $\mu(t)$ is the dipole moment. This gives electric field noise with spectral density $S_{E}(\omega) = 2\int_{-\infty}^{\infty}\langle E(t)E(0)\rangle e^{i\omega t}dt$. Fluctuations of the quadrupole field cause the frequency to fluctuate. The electric field noise spectrum typically follows $S_E(\omega) \propto 1/\omega^\beta$ with exponent $\beta$ ranging around 1 (by default we pick 0.9). To avoid divergence of the total noise power at low frequency, we include a low-frequency cutoff $\omega_{0}=2\pi f_{\mathrm{IR}}$ with $f_{\mathrm{IR}}=3\,\mathrm{kHz}$,
\begin{equation}\label{eq:spectral-density}
S_E(\omega) =
\begin{cases}
S_0, & \omega < \omega_{0} \\
S_0 \left( \dfrac{\omega_{0}}{\omega} \right)^\beta, & \omega \geq \omega_{0}
\end{cases}
\end{equation}
Here, we use the $1/f$ type of noise spectral for the idling dephasing error $\epsilon_\text{dephasing}$. For WISE, when multiple controls share a DAC, voltage fluctuations are identical at the source. Therefore, to evaluate the collective effect, we sum up the electric-field correlation function with respect to mode $m$. Heating occurs when electric field noise has spectral density at the mode frequency. The heating rate for mode $m$ is evaluated with the Fermi-Golden rule~\cite{brownnutt2015ion}. Note that here we define the rate to be the case when an ion in the ground state is excited into the first vibrational state,
\begin{equation}\label{eq:spectral-heating}
\Gamma_m = \frac{e^2}{4m_{\mathrm{ion}}\hbar \omega_m} \sum_l |\mathcal{E}_{l\to m}|^2 S_{E_l}(\omega_m).
\end{equation}
The mean phonon number grows as $\bar{n}_m(t+\Delta t) = \bar{n}_m(t) + \Gamma_m \Delta t$.

\noindent\textbf{Operational phonon.}
Motivated by various techniques in the literature on the thermal phonon generation modelling for ion movement~\cite{walther2012controlling,bowler2012coherent}, swaps~\cite{kaufmann2017fast}, junction crossing~\cite{blakestad2009high,moehring2011design,blakestad2011near} and split and merge~\cite{ruster2014experimental}, we build our heuristics. For example, the trapping potential in a trap is created by a strong radio-frequency (RF) field. In reality, the RF voltage source has some noise, small, unwanted electric fields oscillating at frequencies just above and below the main RF tone. The RF field creates a time-averaged pseudo-potential that confines the ion. The gradient of this pseudo-potential, is largest near the junctions or barriers. When the ion is located on such a slope, different frequency components interfere and creates an oscillating force at exactly the ion's natural frequency. This resonant force drives the ion's motion, increasing its energy. Therefore, phonon accumulates during the junction crossing. Also, to move the ion we need DACs, which update their output voltage at a fixed clock rate. The output is not perfectly smooth; it consists of a series of voltage steps. In the frequency domain, this stepped output contains strong harmonics of the update rate. If the ion's frequency is close to one of these harmonics, the DAC output can resonantly drive the ion's motion. This is analogous to pushing a swing at its natural frequency. The excitation occurs because each voltage step gives the ion a small `kick.' If the kicks occur at a rate matching the ion's oscillation, they add up coherently, leading to a large increase in energy, which affects the phonon count during transportation.

\noindent\textbf{Cooling model and threshold.} To maintain the chain of ions near its ground state throughout a computation, we implement a discrete cooling protocol triggered when the accumulated phonon number $\bar{n}$ exceeds a predefined threshold $\bar{n}_{\text{th}}$. This threshold is set to ensure that phonon accumulation does not cause unacceptable gate infidelity between cooling events. The cooling process itself is not error-free: for example, during a cooling operation~\cite{ozeri2007errors} a data ion scatters on average $\Gamma \tau$ photons, where $\Gamma$ is its scattering rate and $\tau$ is the cooling duration, so it scatters at least one with probability $1-e^{-\Gamma\tau}\approx\Gamma\tau$. Each such scattering event could lead to qubit decoherence. The number of operations required to reduce the phonon number from its current value $\bar{n}_{\text{current}}$ to the target ground state $\bar{n}_{\text{target}}$ is determined by the cooling efficiency. The total cooling error $\epsilon_{\text{cool}}$ incurred is then added to the cumulative infidelity whenever the phonon number crosses the threshold. This approach approximates the cooling overhead as a constant error term applied at discrete intervals, capturing the essential trade-off between heating and the finite fidelity of the cooling process itself. The threshold $\bar{n}_{\text{th}}$ can be tuned to balance the frequency of cooling events against the thermal error accumulated between them, enabling optimization of the overall gate fidelity for a given sequence.

\noindent\textbf{Heating noise.}
As mentioned in the previous section, the accumulated phonon will cause errors in the system, which leads to two types of noise in our model. The first one is the dephasing noise. Physically, if the cooling is imperfect, the mode begins with a residual mean phonon population, which leads to an under- or over-rotation of the entangling gate. For a gate of duration $\tau$, the added phonons are $\Delta\bar{n}_m = \Gamma_m\tau$, which leads to the heating error $\epsilon_\text{heating}$. And we attribute the decoherence during the idling period to the Pauli-Z dephasing error. In addition to heating, the ion chain mode can lose phase coherence due to slow fluctuations in the trapping potential. These fluctuations cause the frequency to wander during the MS gate execution, which results directly in phase error accumulations. In practice, we attribute a quadratic dependence of the mean phonon number for the error and the dephasing contributes an error proportional to the gate duration~\cite{ballance2016high}.

\noindent\textbf{Correlated noise.}
In the WISE architecture, crosstalk arises from two distinct sources: (i) collective modes of the ions in the trap coupling unintended ion pairs during entangling gates, and (ii) shared control electronics inducing correlated errors across zones. We model the former here and leave the later in the next paragraph. For a set of parallel entangling gates on ion pairs, crosstalk corresponds to unintended entangling phases for pairs $[i,j] \in I$, where $I$ contains all pairs not intended to interact. From Ref.~\cite{Cheng2024} the crosstalk entangling phase is
\begin{align}
\theta_{i,j}(\tau) = 2 \sum_m &\eta_{j,m} \eta_{i,m} \int_0^\tau dt_2 \int_0^{t_2} dt_1 \Omega_j(t_2) \Omega_{i}(t_1)\nonumber  \\
& \times \sin(\mu_j t_1) \sin(\mu_{i} t_2) \sin[\omega_m(t_1 - t_2)],
\label{eq:crosstalk_phase}
\end{align}
where $\Omega$ is the operation design down to the pulse level and $\mu$ is the detuning associated with the operations.
An unintended phase $\theta_{i,j}$ between a spectator ion $i$ and a gate ion $j$ corresponds to an error probability
\begin{equation}
p_{i,j} = \sin^2\left( \theta_{i,j}(\tau) \right).
\label{eq:crosstalk_prob}
\end{equation}
For an MS gate on ions $a$ and $b$ in a trap of capacity $k>2$, we project this crosstalk onto a dephasing channel on each spectator $s$ of the trap,
\begin{equation}
\mathcal{E}_s(\rho) = (1-q_s)\,\rho + q_s\, Z_s \rho Z_s,
\qquad q_s = p_{s,a} + p_{s,b},
\label{eq:ELD}
\end{equation}
applied after the gate; for $k=2$ there are no spectators and no crosstalk.

\noindent\textbf{Device control noise.} There are several device level noise mentioned in ~\cite{Malinowski2023}, for which we collectively denote as $\epsilon_\text{device}$. Opening a switch injects charge $\Delta Q$ into the controlling points in the trap, producing a voltage step $\Delta V$ which creates an impulsive electric field that displaces the ion’s state. Also, control lines are periodically updated via multiplexers. Once charged to a voltage $V_0$, they are disconnected from the DAC and left floating. The voltage decays through the period $\tau_{\text{hold}}$ according to $V(t) = V_0 \exp\left(-\frac{t}{\tau_{\text{hold}}}\right)$. In our model $\epsilon_{\text{device}}$ is a fixed per-gate constant ($10^{-7}$ per $20\,\mu\mathrm{s}$ of gate duration per ion in the crystal).

\noindent\textbf{Other noise.} For single-qubit unitary operations, as we mentioned in the beginning of this section, we add a constant penalty to model the error. The same goes with our modelling for the measurement/reset error.

\noindent\textbf{Simple mapping.}
Consider a physical quantum gate acting on a $d$-dimensional Hilbert space. Let the ideal gate be $U$, and suppose that the implemented gate is described by a quantum channel
\begin{equation}
\widetilde{\mathcal{U}}=\mathcal{E}\circ\mathcal{U},
\end{equation}
where $\mathcal{U}(\rho)=U\rho U^\dagger$ is the ideal operation and $\mathcal{E}$ describes the physical noise. We characterize the physical gate by its average gate fidelity,
\begin{equation}
F_{\mathrm{avg}}(\mathcal{E})=\int d\psi\langle\psi|\mathcal{E}(|\psi\rangle\langle\psi|)|\psi\rangle,
\end{equation}
where the integral is taken over the Haar measure on pure input states. Define the experimentally measured average gate infidelity as
\begin{equation}
r:=1-F_{\mathrm{avg}}(\mathcal{E}).
\end{equation}
We now replace the physical noise channel by an effective Pauli channel with the same average gate fidelity. We take the depolarizing noise channel for example. Define
\begin{equation}
\mathcal{D}_p(\rho)=p\rho+(1-p)\frac{I}{d},
\end{equation}
where $p\in[0,1]$ is the depolarizing parameter. For a pure input state
$\rho_\psi=|\psi\rangle\langle\psi|$, the output state is
\begin{equation}
\mathcal{D}(\rho_\psi)=p\rho_\psi+(1-p)\frac{I}{d}.
\end{equation}
The corresponding state fidelity is therefore
\begin{equation}
p\langle\psi|\rho_\psi|\psi\rangle+(1-p)\frac{\langle\psi|I|\psi\rangle}{d}=p+\frac{1-p}{d}.
\end{equation}
Since this quantity is independent of $|\psi\rangle$, the average gate fidelity of the depolarizing channel is
\begin{equation}
F_{\mathrm{avg}}(\mathcal{D}_p)=p+\frac{1-p}{d}.
\end{equation}
We choose $p$ such that the effective depolarizing channel has the same average gate fidelity as the experimentally characterized physical channel:
\begin{equation}
F_{\mathrm{avg}}(\mathcal{D}_p)=1-r.
\end{equation}
Thus, the effective depolarizing channel corresponding to a measured average gate infidelity $r$ is $p=1-\frac{d}{d-1}r$. Stim instead parametrises a depolarising channel by the total probability $q$ of applying a non-identity Pauli, $\mathcal{D}(\rho)=(1-q)\rho+\frac{q}{d^2-1}\sum_{P\neq I}P\rho P^\dagger$, so that $q=\frac{d^2-1}{d^2}(1-p)=\frac{d+1}{d}\,r$, i.e.\ $q=\tfrac32 r$ for one qubit and $q=\tfrac54 r$ for two. Our mechanism models (heating, scattering, device control, cooling) produce error probabilities, which we use directly as $q$.

\subsection{Sensitivity, noise budget and cycle-time decomposition}
\label{si:noise-budget}

\begin{figure*}[t]
\centering
\figorbox[width=0.6\linewidth]{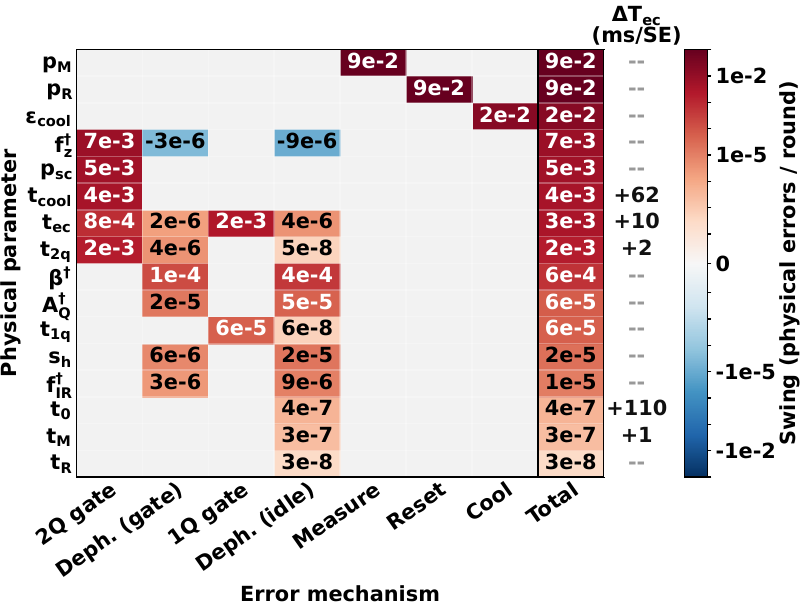}
\caption{Sensitivity of each physical parameter in Table~\ref{tab:tech-tiers}. Refer to the table to identify each physical parameter symbol on the y-axis. The right-most column of the heatmap shows the increase in the expected number of physical faults per syndrome-extraction round when changing the parameter (set by the y-axis) from its Optimistic value to its Pessimistic value (Table~\ref{tab:tech-tiers}). The x-axis separates out the change contributed by each noise mechanism. $\Delta T_\mathrm{ec} (\mathrm{ms/SE})$ column shows the change in syndrome extraction round-time in milliseconds when changing the parameter from Pessimistic to Optimistic. Blank cells: no effect on that mechanism. Physical parameters are ordered from dominant noise source (top) to weakest (bottom). \textbf{Errors due to SPAM $p_M, p_R$ and cooling $\epsilon_\mathrm{cool}$ have the largest swing from Pessimistic and Optimistic, followed by axial frequency $f_z^{\dagger}$, gate-scattering $p_\mathrm{sc}$, cooling time $t_\mathrm{cool}$ and electrode-charging time $t_\mathrm{ec}$. Round-time is most sensitive to ion-swap time $t_0$, cooling time $t_\mathrm{cool}$ and electrode-charging $t_\mathrm{ec}$.}}
\Description{Heatmap of the change in expected physical faults per round, by noise mechanism, when each physical parameter moves between its Optimistic and Pessimistic values, with a column for the change in round time.}
\label{fig:q3_sensitivity}
\end{figure*}

\paragraph{Figure~\ref{fig:q3_sensitivity} parameter ranking by total-error swing.}
Varying each parameter one at a time over its tier range tests sensitivity. We measure the change in the expected number of physical faults
per syndrome extraction: measurement error $p_M$ and reset error $p_R$
($8.9\times10^{-2}$ each) $>$ cooling error $\epsilon_{\mathrm{cool}}$
($2.0\times10^{-2}$) $>$ axial frequency $f_z$ ($7.0\times10^{-3}$) $>$ gate
scattering $p_\mathrm{sc}$ ($5.2\times10^{-3}$) $>$ recool time $t_{\mathrm{cool}}$
($3.5\times10^{-3}$) $>$ electrode charging time $t_{\mathrm{ec}}$
($3.2\times10^{-3}$) $>$ MS gate time ($1.9\times10^{-3}$) $>$ spectral exponent
$\beta$ ($5.8\times10^{-4}$), with every remaining parameter below $10^{-4}$. The crosstalk parameters have no effect at $k=2$ because a two-ion trap holds no spectator ions.

\begin{figure}[htbp]
\centering
\figorbox[width=0.8\columnwidth]{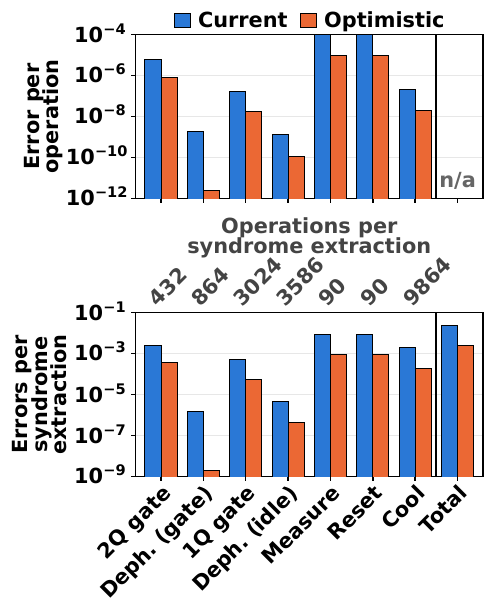}
\caption{Expected physical faults implied by the noise model,
with every parameter at the Current (blue) or Optimistic (orange) tier for the 144-qubit BB code at
the design point $k=2$, $M=16$, $\bar{n}_{\mathrm{th}}=0.0$.
(a)~Error per operation, i.e.\ the quality of one operation. (b)~The same rates
multiplied by the number of such operations in one syndrome extraction (printed
above the panel); the Total bars are the sum over mechanisms. For instance, although there are 3024 1-qubit gates compared to 432 2-qubit gates, the 2-qubit gates still have a larger physical error contribution per round of syndrome extraction than 1-qubit gates. The two dephasing
bars are the same motional-dephasing channel over different windows,
Deph.\,(gate) over the two-qubit gate duration and Deph.\,(idle) over the idle
intervals. \textbf{State preparation and measurement carry about $78\%$ of the
error budget of a syndrome extraction and dominate both per operation and per
round.}}
\Description{Bar charts of error per operation and expected faults per syndrome-extraction round, by noise mechanism, for the Current and Optimistic tiers.}
\label{fig:error-budget}
\end{figure}

\paragraph{Figure~\ref{fig:error-budget} per-operation and per-round error budget.} Operation counts are per syndrome-extraction round, averaged over the rounds of the simulated experiment (\S\ref{sec:metrics}; the $72$ final data-qubit measurements add $18$ to the $72$ ancilla measurements of each round), and a recool is counted per two-ion trap, with error $\epsilon_{\mathrm{cool}}$ on each of its ions. At the Current tier, measurement and reset are equal worst at
$10^{-4}$, an MS gate follows at $5.9\times10^{-6}$, a recool at
$2\times10^{-7}$ and a single-qubit rotation at $1.7\times10^{-7}$. The
Optimistic tier keeps the ordering: $10^{-5}$ for measurement and reset,
$8.45\times10^{-7}$ per MS gate, $2\times10^{-8}$ per recool and
$1.7\times10^{-8}$ per single-qubit gate. Weighting by the operation counts of
one syndrome extraction gives, at Current, measurement $39.1\%$ ($90$
operations), reset $39.1\%$ ($90$), MS gates $11.0\%$ ($432$), recooling
$8.6\%$ ($9864$), single-qubit gates $2.2\%$ ($3024$) and motional dephasing
below $0.1\%$, totalling $2.30\times10^{-2}$ expected physical faults per
syndrome extraction at $T_{\mathrm{ec}}=86.3$\,ms; the Optimistic tier gives
$37.3\%$, $37.3\%$, $15.1\%$, $8.2\%$ and $2.1\%$ at $2.41\times10^{-3}$ and
$15.2$\,ms. State preparation and measurement therefore account for $78\%$ of
the Current-tier and $75\%$ of the Optimistic-tier budget.

\begin{figure}[htbp]
\centering
\figorbox[width=0.8\columnwidth]{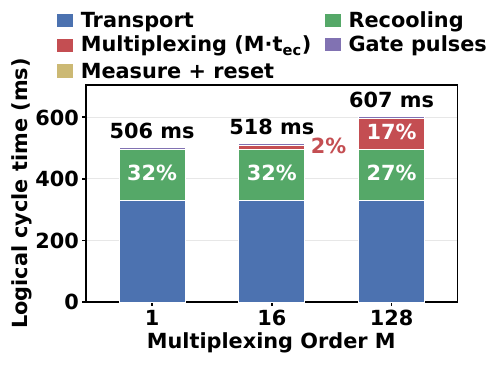}
\caption{Cycle-time decomposition for the 144-qubit BB code at $k=2$,
$\bar{n}_{\rm th}=0$, Current tier. Percentages label the recooling (white) and multiplexing (red) contributions.
Raising $M$ from $1$ to $128$ grows the charging latency from $0.2\%$ to
$16.7\%$ of the cycle. \textbf{Recooling and transport dominate the logical cycle at $M \leq 16$. Multiplexing charging latency becomes a significant contributor to cycle time at WISE's proposed $128$.} The surface code $d=7$ shows the same pattern: transport and recooling take $93.7\%$ and charging $2.2\%$ of its cycle at $M=16$, rising to $15.1\%$ charging at $M=128$.}
\Description{Stacked bars of transport, recooling, charging, gate and measurement time in the BB code's logical cycle for increasing multiplexing order.}
\label{fig:cycle-breakdown}
\end{figure}

\paragraph{Figure~\ref{fig:cycle-breakdown} cycle-time contributors.}
Transport and cooling dominate cycle time. In Figure~\ref{fig:q3_sensitivity}, we find that, per syndrome extraction, the transport scale $t_0$ swings the cycle by $110$\,ms
($149\rightarrow39$\,ms), recool time by $62$\,ms ($128\rightarrow66$\,ms),
electrode charging time $t_{\mathrm{ec}}$ by $10$\,ms ($95\rightarrow85$\,ms)
and two-qubit gate time by $2$\,ms, with everything else below $1$\,ms.

\section{Derivation of the early-fault-tolerance screen}
\label{sec:early-ft-screen-si}

\subsection{Reference workload and scope}

We search for globally controlled trapped-ion design configurations that could support $N_{\mathrm{alg}}=100$ algorithmic logical qubits
through $D_L=10^5$ logical time steps. The number of logical qubits  is motivated
by the lower end of resource estimates for scientifically relevant
phase-estimation and quantum-simulation
workloads~\cite{beverland2022assessingrequirementsscalepractical,PRXQuantum.5.020101,Campbell2022,Huggins2025FLASQ,Spagnoli2026quantumsimulationof}.
The depth is an order-of-magnitude screening target informed by
published logical schedules, ranging from relatively shallow
spin-dynamics examples to substantially deeper phase-estimation
calculations~\cite{beverland2022assessingrequirementsscalepractical,Huggins2025FLASQ,Yoshioka2024}.
We use these dimensions as an optimistic minimum lower bound reference rather than a specific useful quantum application figure, as these are continually improving~\cite{cain2026shorsalgorithmpossible10000}.

\subsection{Logical reliability.}
Following logical-spacetime error
budgeting~\cite{beverland2022assessingrequirementsscalepractical,PRXQuantum.5.020101,Yoshioka2024},
we allocate a generous failure budget
$\epsilon_{\mathrm{mem}}=0.1$ to maintaining the algorithmic register.
A union-bound allocation gives
\begin{equation}
\label{eq:eft-memory-budget}
    \begin{aligned}
        N_{\mathrm{alg}}D_Lp_L &\leq\epsilon_{\mathrm{mem}},\\
        p_L &\leq\frac{0.1}{100\times10^5}=10^{-8}.
    \end{aligned}
\end{equation}

\subsection{Execution window}

We select $T_{\mathrm{job}}=48\,\mathrm{h}$ as a generous execution
window to execute one shot of the quantum circuit experimentally.
Published workflows motivate caution about uninterrupted encoded
operation: H2 reported partial or complete ion loss in some zone
roughly every $30\,\mathrm{min}$ and checks after
$30$--$60\,\mathrm{s}$ batches of repeated circuits~\cite{racetrack_Moses2023};
Honeywell/Quantinuum reported approximately $25\%$ of system duty spent on
automated calibration~\cite{Pino2021}; and the $60\,\mathrm{h}$
Oxford Ionics data-acquisition campaign included regular automated
motional-frequency calibration~\cite{oxfordionics_Loschnauer2024}.
Maintaining one, uninterrupted, physical quantum circuit on real hardware beyond 48-hours would therefore be a significant challenge. This limit gives a minimum logical clock speed $f_L$ requirement of
\begin{equation}
\label{eq:eft-clock-budget}
\geq\frac{D_L}{T_{\mathrm{job}}}
=\frac{10^5}{48\times3600\,\mathrm{s}}
\simeq0.58\,\mathrm{Hz}.
\end{equation}
We use the Optimistic technology tier in
Table~\ref{tab:tech-tiers}.

\subsection{Cryogenic allowance}

We assume a maximum power  allowance of
$\dot Q_{\mathrm{budget}}=600\,\mathrm{W}$ at $10\,$Kelvin.

As an infrastructure
comparison, a demonstrated $600\,\mathrm{W}$, $20\,\mathrm{K}$ helium
circulation system used eight cold heads in two
cryomodules~\cite{Michael2024}. Another study describes an initial
available capacity of $860\,\mathrm{W}$ for its $20\,\mathrm{K}$ target,
but required expanded helium-refrigeration provision~\cite{Brock2023}.
Our allowances should be interpreted as optimistic
design assumption and experimental validation is necessary to evaluate any designs that pass our screening.

We must account for optical heating in cryogenic ion-trap apparatus~\cite{Wipfli2023}:

Consider an idealized delivery path whose modelled propagation loss is
entirely absorptive, followed by an output grating. Let $L_{\mathrm{abs}}$
be the numerical absorption loss in decibels and $\eta_g$ the fraction
of waveguide-output power delivered in the required gate beams.
The power-ratio definition of the decibel~\cite[Sec.~8.7]{NISTSP811}
gives
\begin{equation}
\begin{aligned}
L_{\mathrm{abs}}&=10\log_{10}
\left(\frac{P_{\mathrm{in}}}{P_{\mathrm{out}}}\right),\\
P_{\mathrm{out}}&=\frac{P_{\mathrm{pair}}}{\eta_g},\\
P_{\mathrm{abs}}&=P_{\mathrm{in}}-P_{\mathrm{out}}
=\frac{P_{\mathrm{pair}}}{\eta_g}
\left(10^{L_{\mathrm{abs}}/10}-1\right),
\end{aligned}
\end{equation}
where $P_{\mathrm{pair}}$ is the total optical power in the gate beams
addressing a pair. Note additional distributed scattering or other
propagation losses require an explicit power-flow model, rather than
identifying all insertion loss with absorption.

The time-averaged load is
\begin{equation}
\label{eq:waveguide-heat-general}
\overline P_{\mathrm{wg}}
=\sum_{j=1}^{N_{\mathrm{pair}}}
D_j\frac{P_{\mathrm{pair},j}}{\eta_{g,j}}
\left(10^{L_{\mathrm{abs},j}/10}-1\right),
\end{equation}
where $D_j$ is the actual illumination duty fraction of the corresponding
path (intuitively, the proportion of time lasers are active in an execution schedule).

\paragraph{Concrete example.} Suppose we have
$N_{\mathrm{phys}}=10^4$ ions as qubits~\cite{IonQRoadmap2026} for our early fault-tolerant workload and assume
$N_{\mathrm{pair}}=N_{\mathrm{phys}}/2$ independently supplied pair sites. Then, suppose a power draw of
$P_{\mathrm{pair}}=0.98\,\mathrm{W}$ total (the calculated requirement across three Raman beams for a
$10\,\mu\mathrm{s}$ $^{137}\mathrm{Ba}^{+}$ ground-state-qubit
M{\o}lmer--S{\o}rensen gate~\cite{Moore2023}).
Take $0.6\,\mathrm{dB}$ absorption (the loss measured for a
$369\,\mathrm{nm}$ waveguide path~\cite{Kwon_2024}, assumed here to be entirely absorptive) and a measured $50\%$ upward-radiation efficiency
(for $729\,\mathrm{nm}$ gratings~\cite{Mehta_2020}). Then choose $D_{\mathrm{ref}}=70/(70+300)\simeq0.189$ (from Helios, Quantinuum~\cite{helios2025}).
Equation~\eqref{eq:waveguide-heat-general} then gives
\begin{equation}
\label{eq:eft-waveguide-reference}
\begin{aligned}
\overline P_{\mathrm{wg}}^{\mathrm{ref}}
&=\frac{10^4}{2}\frac{0.98\,\mathrm{W}}{0.50}
\frac{70}{70+300}\left(10^{0.6/10}-1\right)\\
&=274.7\,\mathrm{W}\simeq275\,\mathrm{W}.
\end{aligned}
\end{equation}

The reference value $\overline P_{\mathrm{wg}}^{\mathrm{ref}}
\simeq275\,\mathrm{W}$ includes absorption in the modeled waveguide
section only.

\subsection{Complete cold-stage balance and electrical headroom}

Define $p_{\mathrm{el}}$ as the total modeled, time-averaged cold-stage electrical dissipation, including shared hardware, normalized by
$N_{\mathrm{alg}}$.

Let $P_{\mathrm{extra}}\geq0$ contain cold loads not
already included in $N_{\mathrm{alg}}p_{\mathrm{el}}$ or
$\overline P_{\mathrm{wg}}$: absorption associated with repumping, cooling,
and parasitic conductive and radiative heat
loads~\cite{Pino2021,Dubielzig2021}.

The necessary steady-state thermal balance is
\begin{equation}
\label{eq:eft-electrical-budget}
N_{\mathrm{alg}}p_{\mathrm{el}}+
\overline P_{\mathrm{wg}}+P_{\mathrm{extra}}
\leq\dot Q_{\mathrm{budget}}.
\end{equation}
In the reference deployment,
\begin{equation}
100p_{\mathrm{el}}+274.7\,\mathrm{W}+P_{\mathrm{extra}}
\leq600\,\mathrm{W}.
\end{equation}
Thus $p_{\mathrm{el}}=1\,\mathrm{W/LQ}$ leaves
$225.3\,\mathrm{W}$ for additional loads and margin, and the Current-tier distance-7 surface-code design point ($2.04\,\mathrm{W/LQ}$, \S\ref{sec:result2}) leaves $122\,\mathrm{W}$, whereas a design at the
$3.25\,\mathrm{W/LQ}$ ceiling leaves only $0.3\,\mathrm{W}$ for any additional load $P_\mathrm{extra}$.

\subsection{Screening procedure and interpretation}

For each architecture, we search the evaluated code, distance, and
cooling-policy choices under the Optimistic physical parameters.
A configuration must satisfy
\begin{equation}
\label{eq:eft-screen}
\begin{gathered}
p_L\leq10^{-8},\qquad f_L\geq0.58\,\mathrm{Hz},\\
p_{\mathrm{el}}\leq3.25\,\mathrm{W/LQ}.
\end{gathered}
\end{equation}
for early fault-tolerance.

\section{Parameter selection and sensitivity analysis}
\label{sec:parameter-justification}

\begin{table*}[tp]
\centering
\footnotesize
\setlength{\tabcolsep}{4pt}
\renewcommand{\arraystretch}{1.04}
\caption{Technology-tier inputs and sensitivity ranges for the WISE
architecture. Current gate-pulse and transport benchmarks follow
Malinowski et al.~\cite{Malinowski2023}. Pessimistic and optimistic entries combine experimental
benchmarks with the model-based margins (\S\ref{sec:parameter-justification}). Every row listed here is consumed by the simulated noise or timing model. The sensitivity comparison of Fig.~\ref{fig:q3_sensitivity} sweeps each row, one at a time, between its Pessimistic and Optimistic endpoints with every other row held at Current.  Starred rows ($\ast$) act only for trap capacity $k>2$. Transport primitives are fixed fractions of $t_0$ in the model. A dash denotes a quantity held fixed.}
\label{tab:tech-tiers}
\begin{tabular}{@{}l c c c l@{}}
\toprule
\textbf{Parameter}
& \textbf{Current}
& \textbf{Pessimistic}
& \textbf{Optimistic}
& \textbf{Source} \\
\midrule

\multicolumn{5}{@{}l}{\textit{Timing ($\mu\mathrm{s}$)}} \\
Electrode charging $t_{ec}$
& 1 & 5 & 0.1 & \S\ref{par:ec} \\
Swap/reorder scale $t_0$
& 70 & 150 & 10 & \S\ref{par:swap} \\
Single-qubit pulse $t_{1q}$
& 1 & 2 & 0.1 & \S\ref{par:g1} \\
Two-qubit MS pulse $t_{2q}$
& 100 & 200 & 10 & \S\ref{par:g2} \\
Measurement $t_M$
& 200 & 400 & 50 & \S\ref{par:tm} \\
Reset $t_R$
& 30 & 50 & 5 & \S\ref{par:tr} \\
Recooling $t_{\mathrm{cool}}$
& 200 & 500 & 50 & \S\ref{par:tc} \\
\midrule

\multicolumn{5}{@{}l}{\textit{Measurement, reset, and cooling errors}} \\
Measurement error $p_M$
& $10^{-4}$ & $10^{-3}$ & $10^{-5}$ & \S\ref{par:pm} \\
Reset error $p_R$
& $10^{-4}$ & $10^{-3}$ & $10^{-5}$ & \S\ref{par:pr} \\
Cooling-induced error $\epsilon_{\mathrm{cool}}$
& $10^{-7}$ & $10^{-6}$ & $10^{-8}$ & \S\ref{par:pcool} \\
\midrule

\multicolumn{5}{@{}l}{\textit{Transport primitives, fixed fractions of $t_0$ ($\mu\mathrm{s}$)}} \\
Split / merge $t_{\mathrm{sm}}$ ($0.80\,t_0$)
& 56 & 120 & 8 & \S\ref{par:split} \\
Shuttle $t_{\mathrm{sh}}$ ($0.05\,t_0$)
& 3.5 & 7.5 & 0.5 & \S\ref{par:shuttle} \\
Junction crossing $t_{\mathrm{j}}$ ($0.50\,t_0$)
& 35 & 75 & 5 & \S\ref{par:junction} \\
Rotation $t_{\mathrm{rot}}$ ($0.42\,t_0$)
& 29.4 & 63 & 4.2 & \S\ref{par:rotation} \\
Crossing swap $t_{\mathrm{xs}}$ ($1.00\,t_0$)
& 70 & 150 & 10 & \S\ref{par:swap} \\
\midrule

\multicolumn{5}{@{}l}{\textit{Additional tiered noise inputs}} \\
Heating-amplitude multiplier $s_h$
& 0.05 & 0.20 & 0.005 & \S\ref{par:heating} \\
Gate-scattering contribution $p_{\mathrm{sc}}$
& $10^{-6}$ & $2\times10^{-5}$ & $10^{-8}$
& \S\ref{par:scattering} \\
\midrule

\multicolumn{5}{@{}l}{
\textit{Fixed-default noise parameters: \textbf{not varied outside of sensitivity analysis in Figure~\ref{fig:q3_sensitivity}}}} \\
Spectral exponent $\beta$
& 0.9 & 1.6 & 0.4 & \S\ref{par:beta} \\
Low-frequency cutoff $f_{\mathrm{IR}}$ (kHz)
& 3 & 0.3 & 30 & \S\ref{par:cutoff} \\
Axial reference frequency $f_z$ (MHz)
& 2 & 1 & 3 & \S\ref{par:axial} \\
Gradient-noise coefficient $A_Q$
& $1.43\times10^{-3}$ & $1.43\times10^{-2}$
& $1.43\times10^{-4}$ & \S\ref{par:gradient} \\
Crosstalk drive offset $f_\mu^{\ast}$ (MHz)
& 3 & 2.5 & 4 & \S\ref{par:offset} \\
Crosstalk phase coefficient $A_{\mathrm{ct}}^{\ast}$
& 30 & 60 & 15 & \S\ref{par:act} \\
Pulse segments $N_{\mathrm{seg}}^{\ast}$
& 100 & 50 & 200 & \S\ref{par:segments} \\
\midrule

\multicolumn{5}{@{}l}{
\textit{Shared architecture and physical reference settings}} \\
Height $h$ ($\mu\mathrm{m}$) & $40$ & -- & -- & \S\ref{par:geometry} \\

\bottomrule
\end{tabular}
\end{table*}

In this section, we explain the motivations and outline the literature support for our parameter choices of the noise modelling. References support architectural defaults, component benchmarks, or physical mechanisms; they do not imply the simultaneous experimental realization of all parameters within a technology tier. For each parameter, we pick three values according to (the current (experimental) demonstration, the pessimistic scenario, the optimistic prediction in the next couple of years).

\subsection{Operation durations}

\subsubsection{Electrode charging}
\label{par:ec}
We use $t_{ec}=(1,5,0.1)\,\mu\mathrm{s}$, spanning faster and slower
settling around WISE's $3\,\mu\mathrm{s}$ reference~\cite{Malinowski2023}.
Nanosecond changes of trapping potentials motivate low-impedance fast
switching~\cite{Alonso2016}. For
$t_{ec}=R_{\mathrm{eff}}C_{\mathrm{hold}}\ln(1/\delta_{\mathrm{set}})$ with
a $30\,\mathrm{pF}$ storage capacitance, settling to
$\delta_{\mathrm{set}}=10^{-4}$ requires effective charging-path
resistances of approximately $(3.62,18.1,0.362)\,\mathrm{k\Omega}$.

\subsubsection{Swap and reorder scale}
\label{par:swap}
We set $t_0=(70,150,10)\,\mu\mathrm{s}$. The current value lies between
$42\,\mu\mathrm{s}$ experimental ion-order reversal and WISE's
$100\,\mu\mathrm{s}$ swap reference~\cite{kaufmann2017fast,Malinowski2023}.
The pessimistic endpoint adds more control allowance; the optimistic
endpoint assumes improved transport control, motivated by inverse-engineered
rotations and few-microsecond theoretical separation
protocols~\cite{Tobalina2021,Guglielmo2024}. The transport primitives of Table~\ref{tab:tech-tiers} are fixed fractions of $t_0$ in the model; a physical crossing swap (two ions exchanging places within one zone) costs $1.00\,t_0$, one horizontal swap sub-pass (shuttle, merge, rotation, split, shuttle) $2.12\,t_0$, and one vertical column bucket $5.10\,t_0$.

\subsubsection{Single-qubit gate duration}
\label{par:g1}
We use $t_{1q}=(1,2,0.1)\,\mu\mathrm{s}$, retaining the WISE intrinsic
pulse reference~\cite{Malinowski2023}. The endpoints represent weaker and
stronger effective drives. Demonstrated trapped-ion rotations in less than
$50\,\mathrm{ps}$ motivate the feasibility of substantially faster atomic
control underlying the $100\,\mathrm{ns}$ target~\cite{Campbell2010}.

\subsubsection{Two-qubit gate duration}
\label{par:g2}
We use $t_{2q}=(100,200,10)\,\mu\mathrm{s}$ around WISE's pulse
reference~\cite{Malinowski2023}. Experimental entangling-gate durations of
$3.8$--$520\,\mu\mathrm{s}$ support this timing range~\cite{ballance2016high},
while fault-tolerance projections include $15\,\mu\mathrm{s}$
M\o lmer--S\o rensen operations~\cite{Bermudez2019}. The optimistic endpoint
assumes faster entangling control alongside the separately specified error
improvements.

\subsubsection{Measurement duration}
\label{par:tm}
We set $t_M=(200,400,50)\,\mu\mathrm{s}$. Adaptive readout averages of
$145\,\mu\mathrm{s}$ at approximately $10^{-4}$ error and
$46\,\mu\mathrm{s}$ at $9(1)\times10^{-4}$ error motivate the current and
optimistic timing scales~\cite{Myerson2008,Todaro2021}; the pessimistic
value follows from~\cite{Bermudez2019}. We use effective scheduling
durations, with low-latency decision electronics assumed.

\subsubsection{Reset duration}
\label{par:tr}
We use $t_R=(30,50,5)\,\mu\mathrm{s}$. The current value interpolates
between the $50\,\mu\mathrm{s}$ reference and $10\,\mu\mathrm{s}$ anticipated reset of ~\cite{Bermudez2019}. Demonstrated $5\,\mu\mathrm{s}$ calcium population-preparation pulses motivate the optimistic pumping timescale, with improved polarization and state preparation assumed for the associated
error target~\cite{Barreto2023}.

\subsubsection{Recooling duration}
\label{par:tc}
We set $t_{\mathrm{cool}}=(200,500,50)\,\mu\mathrm{s}$. Two-stage
sympathetic EIT cooling in $45\,\mu\mathrm{s}$, compared with  $\approx 600\,\mu\mathrm{s}$ for earlier sequential Raman-sideband cooling,
motivates the range~\cite{Lin2013}. The current duration also lies between
published $100$ and $400\,\mu\mathrm{s}$ recooling benchmarks~\cite{Bermudez2019}. Routing overhead is accounted for separately.

\subsection{Measurement, reset, and cooling errors}

\subsubsection{Measurement error}
\label{par:pm}
We use $p_M=(10^{-4},10^{-3},10^{-5})$. The first two values follow
high-fidelity readout and fault-tolerance reference
budgets~\cite{Myerson2008,Bermudez2019}. Heralded combined state-preparation and measurement error of $5(4)\times10^{-6}$ motivates the optimistic fidelity scale~\cite{sotirova2024highfidelityheraldedquantumstate}. Improved photon collection, background
rejection, and state transfer provide routes toward the fast, low-error
target~\cite{Todaro2021}.

\subsubsection{Reset error}
\label{par:pr}
We select $p_R=(10^{-4},10^{-3},10^{-5})$ around the improved reset
budget explicitly studied in~\cite{Bermudez2019}. The endpoints allow
tenfold degradation or improvement. Fast optical pumping and low-error
heralded preparation motivate the improvement direction, with suppression
of polarization, transfer, and preparation-error floors assumed for the
unconditional reset target~\cite{Barreto2023,sotirova2024highfidelityheraldedquantumstate}.

\subsubsection{Cooling-induced qubit error}
\label{par:pcool}
We use $\epsilon_{\mathrm{cool}}=(10^{-7},10^{-6},10^{-8})$ for direct
internal-state disturbance per data qubit per recooling operation.
Spectral isolation motivates these budgets: Barrett et al.
calculate approximately $4\times10^{-11}$ spontaneous-emission probability
on a data ion per cooling Raman pulse~\cite{Barrett2003}. An illustrative
$100$-pulse sequence contributes $4\times10^{-9}$ from this mechanism;
residual motional excitation is treated separately.

\subsection{Transport primitives}

\subsubsection{Split and merge}
\label{par:split}
In the model, split and merge are tied to the swap scale as $t_{\mathrm{sm}}=0.80\,t_0=(56,120,8)\,\mu\mathrm{s}$; the fraction was calibrated at the WISE reference $t_0=100\,\mu\mathrm{s}$ against the benchmarks below. Demonstrated
$80\,\mu\mathrm{s}$ ion separation anchors the calibration~\cite{ruster2014experimental};
the pessimistic value adds $50\%$. The $30\,\mu\mathrm{s}$ separation/merge projection of~\cite{Bermudez2019} and optimized separation protocols~\cite{Guglielmo2024} support improvement below the current value. Equal split and merge durations
are a scheduling assumption.

\subsubsection{Linear shuttle}
\label{par:shuttle}
In the model $t_{\mathrm{sh}}=0.05\,t_0=(3.5,7.5,0.5)\,\mu\mathrm{s}$. Transport over
$280\,\mu\mathrm{m}$ in $3.6\,\mu\mathrm{s}$ with approximately $0.10$
added motional quanta has been demonstrated~\cite{walther2012controlling}. The Optimistic value would require $560\,\mathrm{m/s}$.

\subsubsection{Junction crossing}
\label{par:junction}
In the model $t_{\mathrm{j}}=0.50\,t_0=(35,75,5)\,\mu\mathrm{s}$. For an assumed
$200\,\mu\mathrm{m}$ junction segment, the $4\,\mathrm{m/s}$ experimental
transport benchmark gives $50\,\mu\mathrm{s}$~\cite{Burton_2023}. The three tier values correspond to $5.7$, $2.67$, and $40\,\mathrm{m/s}$; the Optimistic speed lies an order of magnitude beyond the benchmark.

\subsubsection{Crystal rotation}
\label{par:rotation}
In the model $t_{\mathrm{rot}}=0.42\,t_0=(29.4,63,4.2)\,\mu\mathrm{s}$; the fraction reproduces the demonstrated $42\,\mu\mathrm{s}$ two-ion order-reversal duration~\cite{kaufmann2017fast} at $t_0=100\,\mu\mathrm{s}$;
the pessimistic value adds $50\%$. The $20\,\mu\mathrm{s}$ projection of~\cite{Bermudez2019} and improved
inverse-engineered rotations~\cite{Tobalina2021} motivate our optimistic projection.

\subsection{Additional tiered noise inputs}

\subsubsection{Heating-amplitude multiplier}
\label{par:heating}
We use $s_h=(0.05,0.20,0.005)$, giving nominal reference rates
$\dot{\bar n}_{\mathrm{ref}}=10^3s_h=(50,200,5)$ quanta/s.
Measured rates of $63(6)$ quanta/s at $2\,\mathrm{MHz}$ and approximately
$180$ quanta/s in a multispecies mode anchor the first two
scales~\cite{Todaro2021,Delaney2024}. The optimistic tenfold reduction is
motivated by approximately hundredfold electric-field-noise suppression
after electrode cleaning~\cite{Hite2012}. In the implemented model $s_h$ scales the electric-field spectral density $S_E$ that enters the motional-dephasing filter integral (\S\ref{par:beta}); the phonon accumulation during transport and gates uses a fixed reference heating rate of $40$ quanta/s at the axial frequency, independent of $s_h$.

\subsubsection{Scattering during an entangling gate}
\label{par:scattering}
We use $p_{\mathrm{sc}}=(10^{-6},2\times10^{-5},10^{-8})$ for the
Raman-scattering contribution, with atomic structure, detuning, and laser
power setting the improvement route~\cite{ozeri2007errors,Moore2023}. The
optimistic endpoint assumes a suitable electronic-ground-manifold encoding
and sufficient power at increased red detuning, motivated by the absence
of a nonzero asymptotic Raman-scattering floor in the refined
model~\cite{Moore2023}.

\subsection{Fixed-default parameters and sensitivity probes}

\subsubsection{Spectral exponent}
\label{par:beta}
We use $\beta=(0.9,1.6,0.4)$. Talukdar et al.
measure an exponent near $0.9$ and discuss experimental values spanning
$0.4$--$1.6$~\cite{talukdar2016implications}. Larger $\beta$ increases low frequency noise from our spectral-density model (Eq~\ref{eq:spectral-density}).

\subsubsection{Low-frequency cutoff}
\label{par:cutoff}
We retain $f_{\mathrm{IR}}=(3,0.3,30)\,\mathrm{kHz}$ as rollover
scenarios. An inferred $300\,\mathrm{Hz}$ rollover under the assumption
$\beta=1.5$ below $100\,\mathrm{kHz}$ motivates the lower
endpoint~\cite{talukdar2016implications}; the default and upper endpoint probe higher
rollovers within the same spectral model. At fixed normalization and
$\beta=0.9$, the endpoints give plateau noise $7.94$ and $0.126$ times
the default.

\subsubsection{Axial reference frequency}
\label{par:axial}
We use $f_z=(2,1,3)\,\mathrm{MHz}$. Surface-trap operation at
$1\,\mathrm{MHz}$ and heating measurements at $2\,\mathrm{MHz}$ anchor
the lower scales~\cite{Allcock2012,Todaro2021}. Operation at
$5.3\,\mathrm{MHz}$ in another zone of the Todaro trap motivates the
stronger-confinement endpoint~\cite{Todaro2021}. Mode frequencies and
couplings are recomputed for each probe.

\subsubsection{Gradient-noise normalization}
\label{par:gradient}
We retain $A_Q=(1.43\times10^{-3},1.43\times10^{-2},1.43\times10^{-4})$
as model-specific normalizations. The electric-field-gradient mechanism is
motivated by surface-noise models of motional dephasing~\cite{talukdar2016implications}.
The endpoints probe tenfold stronger or weaker gradient-noise coupling
at fixed $S_E$ and $h$; their numerical normalization is an assumption
of our model.

\subsubsection{Crosstalk drive offset}
\label{par:offset}
We use $f_\mu=(3,2.5,4)\,\mathrm{MHz}$. Programmable ion interactions
with $\mu/2\pi=2.043\,\mathrm{MHz}$ motivate MHz-scale drive
placement~\cite{Lu2025Ising}. The endpoints explore altered detuning
from the collective modes to control unwanted couplings~\cite{Cheng2024}.
The response uses all $\delta_m=2\pi(f_\mu-f_m)$, including spectator-mode
resonances.

\subsubsection{Crosstalk phase coefficient}
\label{par:act}
We use the model-specific phase normalization $A_{\mathrm{ct}}=(30,60,15)$.
Collective-mode control and programmable pairwise interactions motivate
suppression of unwanted phases~\cite{Cheng2024,Lu2025Ising}.
The endpoints assume twice or half the residual phase, with the latter
targeting improved pulse optimization. For $p_{ij}=\sin^2\theta_{ij}$,
they imply approximately four times or one-quarter of the current error
in the small-phase regime.

\subsubsection{Pulse segmentation}
\label{par:segments}
We select $N_{\mathrm{seg}}=(100,50,200)$, with $100$ phase segments
used experimentally for programmable three- and four-ion
interactions~\cite{Lu2025Ising}. The endpoints halve or double the
control resolution, with finer optimized shaping motivated by demonstrated
modulated-gate error suppression~\cite{Hayes2012}. For a
$100\,\mu\mathrm{s}$ pulse, segment durations are $(1,2,0.5)\,\mu\mathrm{s}$;
the resulting errors are evaluated through the pulse model.

\subsection{Shared architecture and physical settings}

\subsubsection{Reference height}
\label{par:geometry}
We retain $h=40\,\mu\mathrm{m}$ from WISE, close to the approximately
$39\,\mu\mathrm{m}$ height in the regular trapping region of the Todaro
device~\cite{Malinowski2023,Todaro2021}. The factor
$(40\,\mu\mathrm{m}/h)^4$ follows measured approximately
inverse-fourth-power surface-noise scaling~\cite{Sedlacek2018}.

\section{Evidence for the headline results}
\label{si:evidence}

Table~\ref{tab:design-evidence} lists the sampling evidence behind every code that meets all three screens at the design point $(k,M,\bar n_{\mathrm{th}})=(2,16,0)$, and Table~\ref{tab:fig5-winners} identifies the configuration behind every cell of Figure~\ref{fig:result2-stepC}. Both tables, and Figures~\ref{fig:result2-stepC}--\ref{fig:qec_expert_trajectories}, are generated from the same sweep results with the logical-cycle and power definitions of \S\ref{sec:metrics}. Bounds are one-sided $95\%$ Clopper--Pearson upper bounds on the block-failure probability of one logical cycle; with zero failures in $n$ shots the bound is $1-0.05^{1/n}$.

\begin{table*}[t]
\centering
\footnotesize
\caption{Configurations meeting all early-FTQC screens at $(k,M,\bar n_{\mathrm{th}})=(2,16,0)$: observed failures, shots, upper bound on $p_L$, logical clock speed and power per logical qubit.}
\label{tab:design-evidence}
\begin{tabular}{@{}llrrrrr@{}}
\toprule
Tier & Code & Fail. & Shots & $p_L$ 95\% UB & $f_L$ (Hz) & $P_{\mathrm{LQ}}$ (W) \\
\midrule
Current & Surface $d{=}7$ & 0 & $1.00\times10^{9}$ & $3.00\times10^{-9}$ & 4.58 & 2.04 \\
Optimistic & Surface $d{=}5$ & 0 & $1.00\times10^{9}$ & $3.00\times10^{-9}$ & 52.93 & 1.12 \\
Optimistic & Bacon--Shor $d{=}5$ & 1 & $1.00\times10^{9}$ & $4.74\times10^{-9}$ & 29.23 & 1.12 \\
Optimistic & Surface $d{=}7$ & 0 & $1.00\times10^{9}$ & $3.00\times10^{-9}$ & 26.29 & 2.04 \\
Optimistic & HGP $5{\times}5$ $d{=}5$ & 0 & $9.63\times10^{8}$ & $3.11\times10^{-9}$ & 22.02 & 1.72 \\
Optimistic & Bacon--Shor $d{=}7$ & 0 & $8.08\times10^{8}$ & $3.71\times10^{-9}$ & 11.59 & 2.04 \\
Optimistic & BB $[[72,12,6]]$ & 0 & $1.00\times10^{9}$ & $3.00\times10^{-9}$ & 10.99 & 0.28 \\
\bottomrule
\end{tabular}
\end{table*}

\begin{table*}[t]
\centering
\footnotesize
\caption{Configurations behind every cell of Figure~\ref{fig:result2-stepC}, Optimistic tier: the plotted configuration and the screens it meets, followed, where they differ from it, by the fastest and the lowest-power configurations meeting all three screens in that $(k,M)$ cell.}
\label{tab:fig5-winners}
\begin{tabular}{@{}lllrrrrl@{}}
\toprule
$(k,M)$ & Role & Configuration & $\bar n_{\mathrm{th}}$ & Fail./Shots & $f_L$ (Hz) & $P_{\mathrm{LQ}}$ (W) & Screens \\
\midrule
$(2,1)$ & plotted & Surface $d{=}5$ & 0 & 2/$1.00\times10^{9}$ & 53.73 & 16.236 & fails power \\
$(2,4)$ & plotted & Surface $d{=}5$ & 0 & 0/$1.00\times10^{9}$ & 53.57 & 5.436 & fails power \\
$(2,4)$ & fastest, lowest power & BB $[[72,12,6]]$ & 0 & 0/$1.00\times10^{9}$ & 11.11 & 1.236 & all \\
$(2,16)$ & plotted & Surface $d{=}5$ & 0 & 0/$1.00\times10^{9}$ & 52.93 & 1.116 & all \\
$(2,16)$ & lowest power & BB $[[72,12,6]]$ & 0 & 0/$1.00\times10^{9}$ & 10.99 & 0.276 & all \\
$(2,64)$ & plotted & Bacon--Shor $d{=}5$ & 0 & 1/$1.00\times10^{9}$ & 27.81 & 0.293 & all \\
$(2,64)$ & lowest power & BB $[[72,12,6]]$ & 0 & 0/$1.00\times10^{9}$ & 10.55 & 0.093 & all \\
$(2,256)$ & plotted & Bacon--Shor $d{=}5$ & 0 & 2/$1.00\times10^{9}$ & 23.29 & 0.100 & all \\
$(2,256)$ & lowest power & BB $[[72,12,6]]$ & 0 & 0/$1.00\times10^{9}$ & 9.09 & 0.050 & all \\
$(4,1)$ & plotted & Surface $d{=}5$ & 0 & 3/$1.00\times10^{9}$ & 39.81 & 16.836 & fails power \\
$(4,4)$ & plotted & Surface $d{=}5$ & 0 & 2/$1.00\times10^{9}$ & 39.67 & 5.636 & fails power \\
$(4,4)$ & fastest, lowest power & BB $[[72,12,6]]$ & 0 & 0/$1.00\times10^{9}$ & 4.67 & 1.236 & all \\
$(4,16)$ & plotted & Bacon--Shor $d{=}5$ & 0 & 1/$1.00\times10^{9}$ & 28.64 & 1.156 & all \\
$(4,16)$ & lowest power & BB $[[72,12,6]]$ & 0 & 1/$1.00\times10^{9}$ & 4.64 & 0.276 & all \\
$(4,64)$ & plotted & Bacon--Shor $d{=}5$ & 0 & 1/$8.16\times10^{8}$ & 27.16 & 0.303 & all \\
$(4,64)$ & lowest power & BB $[[72,12,6]]$ & 0 & 0/$1.00\times10^{9}$ & 4.54 & 0.093 & all \\
$(4,256)$ & plotted & Bacon--Shor $d{=}7$ & 0 & 0/$1.00\times10^{9}$ & 7.93 & 0.163 & all \\
$(4,256)$ & lowest power & HGP $5{\times}5$ $d{=}5$ & 0.1 & 1/$1.00\times10^{9}$ & 1.16 & 0.135 & all \\
$(6,1)$ & plotted & Bacon--Shor $d{=}5$ & 0 & 7/$2.00\times10^{9}$ & 20.19 & 18.036 & fails power \\
$(6,4)$ & plotted & Bacon--Shor $d{=}5$ & 0 & 8/$2.00\times10^{9}$ & 20.12 & 6.036 & fails power \\
$(6,16)$ & plotted, lowest power & Bacon--Shor $d{=}5$ & 0 & 10/$2.00\times10^{9}$ & 19.85 & 1.236 & all \\
$(6,64)$ & plotted, lowest power & Bacon--Shor $d{=}5$ & 0 & 11/$2.00\times10^{9}$ & 18.85 & 0.322 & all \\
$(6,256)$ & plotted & Bacon--Shor $d{=}5$ & 0 & 17/$2.00\times10^{9}$ & 15.67 & 0.107 & fails $p_L$ \\
$(8,1)$ & plotted & Bacon--Shor $d{=}5$ & 0.1 & 11/$2.00\times10^{9}$ & 1.07 & 16.836 & fails power \\
$(8,4)$ & plotted & Bacon--Shor $d{=}5$ & 0.1 & 19/$2.00\times10^{9}$ & 1.07 & 5.636 & fails $p_L$ \\
$(8,16)$ & plotted & Bacon--Shor $d{=}5$ & 0.1 & 32/$2.00\times10^{9}$ & 1.07 & 1.156 & fails $p_L$ \\
$(8,64)$ & plotted & Bacon--Shor $d{=}5$ & 0.1 & 35/$2.00\times10^{9}$ & 1.06 & 0.303 & fails $p_L$ \\
$(8,256)$ & plotted & Bacon--Shor $d{=}5$ & 0.1 & 112/$2.00\times10^{9}$ & 1.05 & 0.102 & fails $p_L$ \\
$(16,1)$ & plotted & Bacon--Shor $d{=}5$ & 0 & 173/$5.00\times10^{7}$ & 13.88 & 19.236 & fails $p_L$ \\
$(16,4)$ & plotted & Bacon--Shor $d{=}5$ & 0 & 179/$5.00\times10^{7}$ & 13.82 & 6.436 & fails $p_L$ \\
$(16,16)$ & plotted & Bacon--Shor $d{=}5$ & 0 & 92/$5.00\times10^{7}$ & 13.61 & 1.316 & fails $p_L$ \\
$(16,64)$ & plotted & Bacon--Shor $d{=}5$ & 0 & 157/$5.00\times10^{7}$ & 12.81 & 0.341 & fails $p_L$ \\
$(16,256)$ & plotted & Bacon--Shor $d{=}5$ & 0 & 1003/$5.00\times10^{7}$ & 10.37 & 0.111 & fails $p_L$ \\
\bottomrule
\end{tabular}
\end{table*}

\bibliographystyle{ACM-Reference-Format}
\balance
\bibliography{paper}

\end{document}